\documentclass[acmsmall,nonacm]{acmart}

\makeatletter
\@ifundefined{lhs2tex.lhs2tex.sty.read}%
  {\@namedef{lhs2tex.lhs2tex.sty.read}{}%
   \newcommand\SkipToFmtEnd{}%
   \newcommand\EndFmtInput{}%
   \long\def\SkipToFmtEnd#1\EndFmtInput{}%
  }\SkipToFmtEnd

\newcommand\ReadOnlyOnce[1]{\@ifundefined{#1}{\@namedef{#1}{}}\SkipToFmtEnd}
\usepackage{amstext}
\usepackage{amssymb}
\usepackage{stmaryrd}
\DeclareFontFamily{OT1}{cmtex}{}
\DeclareFontShape{OT1}{cmtex}{m}{n}
  {<5><6><7><8>cmtex8
   <9>cmtex9
   <10><10.95><12><14.4><17.28><20.74><24.88>cmtex10}{}
\DeclareFontShape{OT1}{cmtex}{m}{it}
  {<-> ssub * cmtt/m/it}{}

\DeclareFontShape{OT1}{cmtt}{bx}{n}
  {<5><6><7><8>cmtt8
   <9>cmbtt9
   <10><10.95><12><14.4><17.28><20.74><24.88>cmbtt10}{}
\DeclareFontShape{OT1}{cmtex}{bx}{n}
  {<-> ssub * cmtt/bx/n}{}

\newcommand{\Conid}[1]{\mathit{#1}}
\newcommand{\Varid}[1]{\mathit{#1}}
\newcommand{\anonymous}{\kern0.06em \vbox{\hrule\@width.5em}}
\newcommand{\plus}{\mathbin{+\!\!\!+}}

\renewcommand{\leq}{\leqslant}

\usepackage{polytable}

\@ifundefined{mathindent}%
  {\newdimen\mathindent\mathindent\leftmargini}%
  {}%

\def\resethooks{%
  \global\let\SaveRestoreHook\empty
  \global\let\ColumnHook\empty}
\newcommand*{\savecolumns}[1][default]%
  {\g@addto@macro\SaveRestoreHook{\savecolumns[#1]}}
\newcommand*{\restorecolumns}[1][default]%
  {\g@addto@macro\SaveRestoreHook{\restorecolumns[#1]}}
\newcommand*{\aligncolumn}[2]%
  {\g@addto@macro\ColumnHook{\column{#1}{#2}}}

\resethooks

\newcommand{\onelinecommentchars}{\quad-{}- }
\newcommand{\commentbeginchars}{\enskip\{-}
\newcommand{\commentendchars}{-\}\enskip}

\newcommand{\visiblecomments}{%
  \let\onelinecomment=\onelinecommentchars
  \let\commentbegin=\commentbeginchars
  \let\commentend=\commentendchars}

\newcommand{\invisiblecomments}{%
  \let\onelinecomment=\empty
  \let\commentbegin=\empty
  \let\commentend=\empty}

\visiblecomments

\newlength{\blanklineskip}
\newcommand{\hsindent}[1]{\quad}
\let\hspre\empty
\let\hspost\empty

\EndFmtInput
\makeatother
\ReadOnlyOnce{polycode.fmt}%
\makeatletter

\newcommand{\hsnewpar}[1]%
  {{\parskip=0pt\parindent=0pt\par\vskip #1\noindent}}

\newcommand{\hscodestyle}{}

\newcommand{\sethscode}[1]%
  {\expandafter\let\expandafter\hscode\csname #1\endcsname
   \expandafter\let\expandafter\endhscode\csname end#1\endcsname}

  {\par\noindent
   \advance\leftskip\mathindent
   \hscodestyle
   \let\\=\@normalcr
   \let\hspre\(\let\hspost\)%
   \pboxed}%
  {\endpboxed\)%
   \par\noindent
   \ignorespacesafterend}

  {\hsnewpar\abovedisplayskip
   \advance\leftskip\mathindent
   \hscodestyle
   \let\hspre\(\let\hspost\)%
   \pboxed}%
  {\endpboxed%
   \hsnewpar\belowdisplayskip
   \ignorespacesafterend}

  {\hsnewpar\abovedisplayskip
   \advance\leftskip\mathindent
   \hscodestyle
   \let\\=\@normalcr
   \(\pboxed}%
  {\endpboxed\)%
   \hsnewpar\belowdisplayskip
   \ignorespacesafterend}

\newcommand{\plainhs}{\sethscode{plainhscode}}

\plainhs

  {\hsnewpar\abovedisplayskip
   \advance\leftskip\mathindent
   \hscodestyle
   \let\\=\@normalcr
   \(\parray}%
  {\endparray\)%
   \hsnewpar\belowdisplayskip
   \ignorespacesafterend}

\newcommand{\arrayhs}{\sethscode{arrayhscode}}

  {\parray}{\endparray}

  {\(\parray}{\endparray\)}

\def\codeframewidth{\arrayrulewidth}
\RequirePackage{calc}

  {\parskip=\abovedisplayskip\par\noindent
   \hscodestyle
   \arrayrulewidth=\codeframewidth
   \tabular{@{}|p{\linewidth-2\arraycolsep-2\arrayrulewidth-2pt}|@{}}%
   \hline\framedhslinecorrect\\{-1.5ex}%
   \let\endoflinesave=\\
   \let\\=\@normalcr
   \(\pboxed}%
  {\endpboxed\)%
   \framedhslinecorrect\endoflinesave{.5ex}\hline
   \endtabular
   \parskip=\belowdisplayskip\par\noindent
   \ignorespacesafterend}

\newcommand{\framedhslinecorrect}[2]%
  {#1[#2]}

  {\(\def\column##1##2{}%
   \let\>\undefined\let\<\undefined\let\\\undefined
   \newcommand\>[1][]{}\newcommand\<[1][]{}\newcommand\\[1][]{}%
   \def\fromto##1##2##3{##3}%
   }{\) }%

  {\let\orighscode=\hscode
   \let\origendhscode=\endhscode
   \def\endhscode{\def\hscode{\endgroup\def\@currenvir{hscode}\\}\begingroup}
   \orighscode\def\hscode{\endgroup\def\@currenvir{hscode}}}%
  {\origendhscode
   \global\let\hscode=\orighscode
   \global\let\endhscode=\origendhscode}%

\makeatother
\EndFmtInput
\arrayhs

\usepackage[UKenglish]{isodate}

\makeatletter
\newcommand{\shorteq}{%
  \settowidth{\@tempdima}{-}%
  \resizebox{\@tempdima}{\height}{=}%
}
\makeatother

\setcopyright{none}

\usepackage{booktabs}   
\usepackage{subcaption} 

\usepackage{microtype}

\usepackage{multicol}
\usepackage{amssymb}
\usepackage{mathtools}
\usepackage{stmaryrd}
\usepackage{esvect}
\usepackage{listings}
\usepackage{zi4}
\usepackage[T1]{fontenc}
\usepackage{upquote}
\usepackage{fancybox}
\usepackage{hyperref}
\usepackage{pifont}
\usepackage{textcomp}

\usepackage{longtable}

\usepackage{tikz}

\usepackage{xcolor}

\definecolor{scarlet}{RGB}{ 141, 27, 53}
\definecolor{steelblue}{RGB}{ 11, 65, 108}
\hypersetup{
  unicode,
  bookmarks=true,
  bookmarksnumbered=true,
  bookmarksopen=true,
  linktocpage,
  colorlinks=true,
  citecolor=scarlet,
  linkcolor=steelblue,
  urlcolor=steelblue,
  setpagesize=false,
  linktoc=all
}
\definecolor{light-gray}{gray}{0.92}

\newcommand\numberthis{\addtocounter{equation}{1}\tag{\theequation}}

\newcommand*{\scName}[1]{\textsc{#1}}

\newcommand*{\seman}[1]{[\kern-.16em[ #1 ]\kern-.16em]}

\newcommand{\pfun}[0]{\mathrel{\ooalign{\hfil$\mapstochar\mkern5mu$\hfil\cr$\to$\cr}}}

\newcommand{\myset}[1]{\left\{\,#1\,\right\}}
\newcommand{\validLink}[3]{#1 \xleftrightarrow[]{#3} #2}

\newcommand{\mywhere}[0]{\quad\text{ \textbf{where} } }
\newcommand{\myspace}[0]{\quad\phantom{\text{ \textbf{where} }} }
\newcommand{\myindent}[0]{\phantom{xxxxxxxxxxx}}
\newcommand{\myindentS}[0]{\phantom{xxxx}}
\newcommand{\myindentSS}[0]{\phantom{xx}}

\newcommand*{\dom}{{\normalfont\textsc{dom}}}
\newcommand*{\ldom}{{\normalfont\textsc{ldom}}}
\newcommand*{\rdom}{{\normalfont\textsc{rdom}}}

\newcommand*{\numcircledmod}[1]{\raisebox{.1pt}{\textcircled{\raisebox{-.6pt} {#1}}}}
\newcommand*{\mcond}[1]{\{\small{\ #1 \ \}}}

\newcommand*{\bb}{\begin{array}{lllll}}
\newcommand*{\ee}{\end{array}}
\newcommand*{\sem}[1]{[ \!\! \ensuremath{\mathbin{\char92 !}[\mskip1.5mu \mathbin{\#}\mathrm{1}\mskip1.5mu]\mathbin{\char92 !\char92 !}\mskip1.5mu]\mathbin{\char92 !}} } 

\newenvironment{centerTab}
  {
  \begingroup
  \setlength{\parskip}{0.5ex}
  \par
  \begingroup
  \centering
  \begin{tabular}{c}
  }
  {
  \end{tabular}
  \par
  \endgroup
  \endgroup
  \setlength{\parskip}{1.5ex}
  \noindent
  }

\AtBeginDocument{
\setlength{\jot}{0ex}
\setlength{\abovedisplayskip}{1.25ex plus 0.25ex minus 0.25ex}
\setlength{\belowdisplayskip}{1.25ex plus 0.25ex minus 0.25ex}

}

\allowdisplaybreaks

\lstdefinestyle{dimUnnecessaries}{
  morekeywords = {Lit,Num},
  keywordstyle=\color{lightgray},
  escapeinside={(*@}{@*)}
}

\DeclareMathSizes{10}{9}{7}{5}

\usepackage{enumitem}
\setlist{leftmargin=*,topsep=0ex}

\usepackage{todonotes}

\begin{document}

\setlength{\mathindent}{\parindent}

\title[Retentive Lenses]{Retentive Lenses}         
\begin{anonsuppress}
\titlenote{Draft manuscript (\today)}   
\end{anonsuppress}


\author{Zirun Zhu}
\affiliation{
  \position{}
  \institution{National Institute of Informatics}
}
\email{zhu@nii.ac.jp}

\author{Zhixuan Yang}
\affiliation{
  \position{}
  \institution{National Institute of Informatics}
}
\email{yzx@nii.ac.jp}

\author{Hsiang-Shang Ko}
\orcid{0000-0002-2439-1048}             
\affiliation{
  \position{Assistant Research Fellow}
  \institution{Institute of Information Science, Academia Sinica}
}
\email{joshko@iis.sinica.edu.tw}

\author{Zhenjiang Hu}
\affiliation{
  \position{Professor}
  \institution{Peking University}
  \department{Department of Computer Science and Technology}
}
\email{huzj@pku.edu.cn}

\begin{abstract}
Based on Foster et al.'s lenses, various bidirectional programming languages and systems have been developed for helping the user to write correct data synchronisers.
The two well-behavedness laws of lenses, namely Correctness and Hippocraticness, are usually adopted as the guarantee of these systems.
While lenses are designed to retain information in the source when the view is modified, well-behavedness says very little about the retaining of information: Hippocraticness only requires that the source be unchanged if the view is not modified, and nothing about information retention is guaranteed when the view is changed.
To address the problem, we propose an extension of the original lenses, called \emph{retentive lenses}, which satisfy a new Retentiveness law guaranteeing that if parts of the view are unchanged, then the corresponding parts of the source are retained as well.
As a concrete example of retentive lenses, we present a domain-specific language for writing tree transformations; we prove that the pair of $\textit{get}$ and $\textit{put}$ functions generated from a program in our DSL forms a retentive lens. We demonstrate the practical use of retentive lenses and the DSL by presenting case studies on code refactoring, Pombrio and Krishnamurthi's resugaring, and XML synchronisation.

\end{abstract}



\maketitle

\section{Introduction}
\label{sec:introduction}

We often need to write pairs of transformations to synchronise data.
Typical examples include view querying and updating in relational databases~\cite{Bancilhon1981Update} for keeping a database and its view in sync,
text file format conversion~\cite{Macfarlane2013Pandoc} (e.g.~between Markdown and HTML) for keeping their content and common formatting in sync,
and parsers and printers as front ends of compilers~\cite{Rendel2010Invertible} for keeping program text and its abstract representation in sync.
Asymmetric \emph{lenses}~\cite{Foster2007Combinators} provide a framework for modelling such pairs of programs and discussing what laws they should satisfy; among such laws, two \emph{well-behavedness} laws (explained below) play a fundamental role.
Based on lenses, various bidirectional programming languages and systems (\autoref{sec:relatedWork}) have been developed for helping the user to write correct synchronisers, and the well-behavedness laws have been adopted as the minimum---and in most cases the only---laws to guarantee.
In this paper, we argue that well-behavedness is not sufficient, and a more refined law, which we call \emph{Retentiveness}, should be developed.

To see this, let us first review the definition of well-behaved lenses, borrowing some of \citeauthor{Stevens2008Bidirectional}{'s} terminologies~\cite{Stevens2008Bidirectional}.
Lenses are used to synchronise two pieces of data respectively of types $S$~and~$V$, where $S$~contains more information and is called the \emph{source} type, and $V$~contains less information and is called the \emph{view} type.
Here, being synchronised means that when one piece of data is changed, the other piece of data should also be changed such that consistency is \emph{restored} among them, i.e.~a \emph{consistency relation}~$R$ defined on $S$~and~$V$ is satisfied.
Since $S$~contains more information than~$V$, we expect that there is a function \ensuremath{\Varid{get}\mathbin{:}\Conid{S}\to \Conid{V}} that extracts a consistent view from a source, and this \ensuremath{\Varid{get}} function serves as a consistency restorer in the source-to-view direction: if the source is changed, to restore consistency it suffices to use \ensuremath{\Varid{get}} to recompute a new view.
This \ensuremath{\Varid{get}} function should coincide with~\ensuremath{\Conid{R}} extensionally~\cite{Stevens2008Bidirectional}---that is, \ensuremath{\Varid{s}\mathbin{:}\Conid{S}} and \ensuremath{\Varid{v}\mathbin{:}\Conid{V}} are related by~\ensuremath{\Conid{R}} if and only if $\ensuremath{\Varid{get}}(s) = \ensuremath{\Varid{v}}$.
(Therefore it is only sensible to consider functional consistency relations in the asymmetric setting.)
Consistency restoration in the other direction is performed by another function $\ensuremath{\Varid{put}} : S \times V \to S$, which produces an updated source that is consistent with the input view and can retain some information of the input source.
Well-behavedness consists of two laws regarding the restoration behaviour of \ensuremath{\Varid{put}} with respect to \ensuremath{\Varid{get}} (i.e.~the consistency relation~\ensuremath{\Conid{R}}):
\begin{align}
\ensuremath{\Varid{get}} \ensuremath{(\Varid{put}} \ensuremath{(\Varid{s},\Varid{v}))} &= \ensuremath{\Varid{v}}
\tag{Correctness}\label{equ:correctness} \\
\ensuremath{\Varid{put}} \ensuremath{(\Varid{s},\Varid{get}} \ensuremath{(\Varid{s}))} &= \ensuremath{\Varid{s}}
\tag{Hippocraticness}\label{equ:hippocraticness}
\end{align}
Correctness states that necessary changes must be made by \ensuremath{\Varid{put}} such that the updated source is consistent with the view; Hippocraticness says that if two pieces of data are already consistent, \ensuremath{\Varid{put}} must not make any change.
A pair of \ensuremath{\Varid{get}} and \ensuremath{\Varid{put}} functions, called a \emph{lens}, is \emph{well-behaved} if it satisfies both laws.

Despite being concise and natural, these two properties do not sufficiently characterise the result of an update performed by \ensuremath{\Varid{put}}, and well-behaved lenses may exhibit unintended behaviour regarding what information is retained in the updated source.
Let us illustrate this with a very simple example, in which \ensuremath{\Varid{get}} is a projection function that extracts the first element from a tuple of an integer and a string.
(Hence a source and a view are consistent if the first element of the source tuple is equal to the view.)
\begin{centerTab}
\begin{lstlisting}
get :: (Int, String) -> Int
get (i, s) = i
\end{lstlisting}
\end{centerTab}%
Given this \lstinline{get}\footnote{In this paper, we use \scName{Haskell} notations to write functions, and concrete examples are always typeset in \lstinline{typewriter} font.}, we can define \lstinline[mathescape]{put$_1$} and \lstinline[mathescape]{put$_2$}, both of which are well-behaved with this \lstinline{get} but have rather different behaviour: \lstinline[mathescape]{put$_1$} simply replaces the integer of the source tuple with the view, while \lstinline[mathescape]{put$_2$} also sets the string empty when the source tuple is not consistent with the view.
\begin{centerTab}
\begin{lstlisting}[mathescape]
put$_1$ :: (Int, String) -> Int -> (Int, String)
put$_1$ (i, s) i' = (i', s)
\end{lstlisting}\\
\begin{lstlisting}[mathescape]
put$_2$ :: (Int, String) -> Int -> (Int, String)
put$_2$ src     i'  | get src == i'  = src
put$_2$ (i, s)  i'  | otherwise      = (i',  "")
\end{lstlisting}
\end{centerTab}%
From another perspective, \lstinline[mathescape]{put$_1$} retains the string from the old source when performing the update, while \lstinline[mathescape]{put$_2$} chooses to discard that string---which is not desired but `perfectly legal', for the string does not contribute to the consistency relation.
In fact, unexpected behaviour of this kind of well-behaved lenses could even lead to disaster in practice. For instance, relational databases can be thought of as tables consisting of rows of tuples, and well-behaved lenses used for maintaining a database and its view may erase important data after an update, as long as the data does not contribute to the consistency relation (in most cases this is because the data is simply not in the view). This fact seems fatal, as asymmetric lenses have been considered a satisfactory solution to the longstanding view update problem (stated at the beginning of Foster et~al.'s seminal paper~\cite{Foster2007Combinators}).

The root cause of the information loss (after an update) is that while lenses are designed to retain information, well-behavedness actually says very little about the retaining of information: the only law guaranteeing information retention is Hippocraticness, which merely requires that the \emph{whole} source should be unchanged if the \emph{whole} view is.
In other words, if we have a very small change on the view, we are free to create any source we like.
This is too `global' in most cases, and it is desirable to have a law that makes such a guarantee more `locally'.

To have a finer-grained law, we propose \emph{retentive lenses}, an extension of the original lenses, which can guarantee that if parts of the view are unchanged, then the corresponding parts of the source are retained as well.
Compared with the original lenses, the \ensuremath{\Varid{get}} function of a retentive lens is enriched to compute not only the view of the input source but also a set of \emph{links} relating corresponding parts of the source and the view.
If the view is modified, we may also update the set of links to keep track of the correspondence that still exists between the original source and the modified view.
The \ensuremath{\Varid{put}} function of the retentive lens is also enriched to take the links between the original source and the modified view as input, and it satisfies a new law, \emph{Retentiveness}, which guarantees that those parts in the original source having correspondence links to some parts of the modified view are retained at the right places in the updated source.

The main contributions of the paper are as follows:
\begin{itemize}
  \item We develop a formal definition of retentive lenses for tree-shaped data~(\autoref{sec:framework}).

  \item We present a domain-specific language (DSL) for writing tree synchronisers and prove that any program written in our DSL gives rise to a retentive lens~(\autoref{sec:retentiveDSL}).

  \item We demonstrate the usefulness of retentive lenses in practice by presenting case studies on code refactoring, resugaring, and XML synchronisation (\autoref{sec:application}), with the help of several view editing operations that also update the links between the view and the original source (\autoref{sec:editOperations}).
\end{itemize}
We will start from a high-level sketch of what retentive lenses do~(\autoref{sec:intuitiveRetentiveness}), and after presenting the technical contents, we will discuss related work~(\autoref{sec:relatedWork}) regarding various alignment strategies for lenses,
provenance and origin between two pieces of data, and operational-based bidirectional transformations, before concluding the paper (\autoref{sec:conclusion}).

\begin{figure}[t]
\setlength{\mathindent}{0em}
\begin{small}
\begin{minipage}[t]{0.45\textwidth}
\begin{center}
\begin{lstlisting}[xleftmargin=0pt]
type Annot  =  String
data Expr   =  Plus  Annot Expr Term
            |  Minus Annot Expr Term
            |  FromT Annot Term

data Term   =  Lit   Annot Int
            |  Neg   Annot Term
            |  Paren Annot Expr

data Arith  =  Add Arith Arith
            |  Sub Arith Arith
            |  Num Int
\end{lstlisting}
\end{center}
\end{minipage}
\begin{minipage}[t]{0.5\textwidth}
\begin{center}
\begin{lstlisting}
getE :: Expr -> Arith
getE (Plus   _ e  t) = Add (getE e) (getT t)
getE (Minus  _ e  t) = Sub (getE e) (getT t)
getE (FromT  _    t) = getT t

getT :: Term -> Arith
getT (Lit     _ i  ) = Num i
getT (Neg     _ t  ) = Sub (Num 0) (getT t)
getT (Paren   _ e  ) = getE e
\end{lstlisting}
\end{center}
\end{minipage}
\end{small}
\caption{Data types for concrete and abstract syntax of arithmetic expressions and the consistency relations between them as \lstinline{getE} and \lstinline{getT} functions in \scName{Haskell}.}
\label{fig:runningExpDataTypeDef}
\end{figure}

\section{A Sketch of Retentiveness}
\label{sec:intuitiveRetentiveness}

We will use the synchronisation of concrete and abstract representations of arithmetic expressions as the running example throughout the paper.
The representations are defined in \autoref{fig:runningExpDataTypeDef} (in \scName{Haskell}).
The concrete representation is either an expression of type \lstinline{Expr}, containing additions and subtractions; or a term of type \lstinline{Term}, including numbers, negated terms, and expressions in parentheses.
Moreover, all the constructors have an annotation field of type \lstinline{Annot} mocking up data that exist solely in the concrete representation like code comments and spaces.
The two concrete types \lstinline{Expr} and \lstinline{Term} coalesce into the abstract representation type \lstinline{Arith}, which does not include annotations, explicit parentheses, and negations---negations are considered \emph{syntactic sugar} and represented in the AST by \lstinline{Sub}.

As mentioned in \autoref{sec:introduction}, the core idea of Retentiveness is to use links to relate parts of the source and view.
For data of algebraic data types (which we call `trees' or `terms'), a straightforward interpretation of a `part' is a subtree of the data.
But it is too restrictive in most cases, and a more useful interpretation of a `part' is a \emph{region} of a tree, i.e.~a partial subtree.
Partial trees are trees where some subtrees can be missing.
We will describe the content of a partial tree with a pattern that contains wildcards at the positions of missing subtrees.
In \autoref{fig:swapAndPut}, all grey areas are examples of regions; the topmost region in \lstinline{cst} is located at the root of the whole tree, and its content has the pattern \lstinline{Plus "a plus" _ _}\,, which says that the region includes the \lstinline{Plus} node and the annotation \lstinline{"a plus"}, but not the other two subtrees with roots \lstinline{Minus} and \lstinline{Neg} matched by the wildcards.

\begin{figure}[t]
\centering
\includegraphics[scale=0.65,trim={5.1cm 5.1cm 5.1cm 4.3cm},clip]{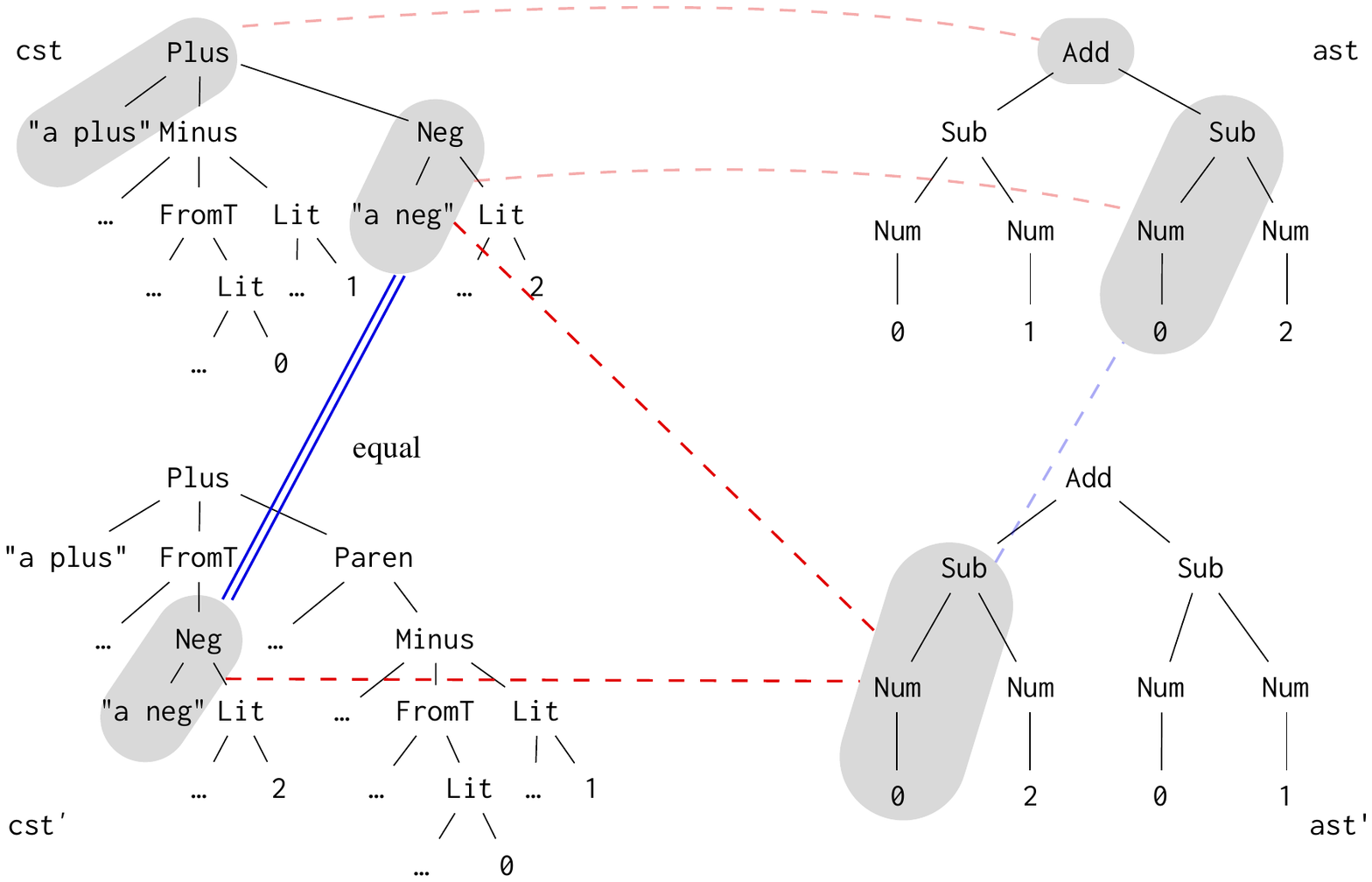}
\caption{Regions, links, and the triangular guarantee.}
\label{fig:swapAndPut}
\end{figure}

Having broken up source and view trees into regions, we can put in \emph{links} to record the correspondences between source and view regions.
In \autoref{fig:swapAndPut}, for example, the light red dashed lines between the source \lstinline{cst} and the view \lstinline{ast = getE cst} represent two possible links.
The topmost region of pattern \lstinline{Plus "a plus" _ _} in \lstinline{cst} corresponds to the topmost region of pattern \lstinline{Add _ _} in \lstinline{ast}, and the region of pattern \lstinline{Neg "a neg" _} in the right subtree of \lstinline{cst} corresponds to the region of pattern \lstinline{Sub (Num 0) _} in \lstinline{ast}.
The \ensuremath{\Varid{get}} function of a retentive lens will be responsible for producing an initial set of links between a source and its view.

As the view is modified, the links between the source and view should also be modified to reflect the latest correspondences between regions.
For example, in \autoref{fig:swapAndPut}, if we change \lstinline{ast} to \lstinline{ast'} by swapping the two subtrees under \lstinline{Add}, then there should be a new link (among others) recording the fact that the \lstinline{Neg "a neg" _} region and the \lstinline{Sub (Num 0) _} region are still related.
We will describe a way of computing new links from old ones in \autoref{sec:editOperations}.

When it is time to put the modified view back into the source, the links between the source and the modified view are used to guide what regions in the old source should be retained in the new one and at what positions.
In addition to the source and view, the \ensuremath{\Varid{put}} function of a retentive lens also takes a collection of links, and provides what we call the \emph{triangular guarantee}, as illustrated in \autoref{fig:swapAndPut}:
when updating \lstinline{cst} with \lstinline{ast'}, the region \lstinline{Neg "a neg" _} (i.e.~syntactic sugar negation) connected by the red dashed link is guaranteed to be preserved in the result \lstinline{cst'} (as opposed to changing it to a \lstinline{Minus}), and the preserved region will be linked to the same region \lstinline{Sub (Num 0) _} of \lstinline{ast'} if we run \lstinline{getE cst'}.
The Retentiveness law will be a formalisation of the triangular guarantee.

\section{Formal Definitions}
\label{sec:framework}

Here we formalise what we described in \autoref{sec:intuitiveRetentiveness}.
Besides the definition of retentive lenses~(\autoref{sec:retentive-lens-definition}), we will also briefly discuss how retentive lenses compose~(\autoref{sec:lensComp}).

\subsection{Retentive Lenses}
\label{sec:retentive-lens-definition}

We start with some notations.
Relations from set~$A$ to set~$B$ are subsets of $A \times B$, and we denote the type of these relations by $A \sim B$.
Given a relation $r : A \sim B$, define its \emph{converse} $r^\circ : B \sim A$ by $r^\circ = \{\, (b, a) \mid (a, b) \in r \,\}$, its \emph{left domain} by $\ldom(r) = \{\,a \in A \mid \exists b.\ (a, b) \in r \,\}$, and its \emph{right domain} by $\rdom(r) = \ldom(r^\circ)$.
The composition $r \cdot s : A \sim C$ of two relations $r : A \sim B$ and $s : B \sim C$ is defined as usual by $r \cdot s = \{\, (a, c) \mid \exists b.\ (a, b) \in r \mathrel\wedge (b, c) \in s\,\}$.
The type of partial functions from~$A$ to~$B$ is denoted by $A \pfun B$.
The \emph{domain} $\dom(f)$ of a function $f : A \pfun B$ is the subset of~$A$ on which $f$~is defined; when $f$~is total, i.e.~$\dom(f) = A$, we write $f : A \to B$.
We will allow functions to be implicitly lifted to relations: a function $f : A \pfun B$ also denotes a relation $f : B \sim A$ such that $(f\,x, x) \in f$ for all $x \in \dom(f)$\footnote{This flipping of domain and codomain (from $A \pfun B$ to $B \sim A$) makes function composition compatible with relation composition: a function composition $\ensuremath{\Varid{g}} \circ \ensuremath{\Varid{f}}$ lifted to a relation is the same as $\ensuremath{\Varid{g}} \cdot \ensuremath{\Varid{f}}$, i.e.~the composition of \ensuremath{\Varid{g}}~and~\ensuremath{\Varid{f}} as relations.}.

We will work within a universal set \ensuremath{\Conid{Tree}} of trees, which is inductively built from all possible finitely branching constructors.
(The semantics of an algebraic data type is then the subset of \ensuremath{\Conid{Tree}} that consists of those trees built with only the constructors of the data type.)
Similarly, the set \ensuremath{\Conid{Pattern}} is inductively built from all possible finitely branching constructors, variables, and a distinguished wildcard element \ensuremath{\anonymous }\,.
We will also need a set \ensuremath{\Conid{Path}} of all possible paths for navigating from the root of a tree to one of its subtrees.
The exact representation of paths is not crucial: paths are only required to support some standard operations such as $\ensuremath{\Varid{sel}} : \ensuremath{\Conid{Tree}} \times \ensuremath{\Conid{Path}} \pfun \ensuremath{\Conid{Tree}}$ such that $\ensuremath{\Varid{sel}}(t, p)$ is the subtree of~\ensuremath{\Varid{t}} at the end of path~\ensuremath{\Varid{p}} (starting from the root), or undefined if \ensuremath{\Varid{p}}~does not exist in~\ensuremath{\Varid{t}}; we will mention these operations in the rest of the paper as the need arises.
But, when giving concrete examples, we will use one particular representation: a path is a list of natural numbers indicating which subtree to go into at each node---for instance, starting from the root of \lstinline{cst} in \autoref{fig:swapAndPut}, the empty path \lstinline{[]} points to the root node \lstinline{Plus}, the path \lstinline{[0]} points to \lstinline{"a plus"} (which is the first subtree under the root), and the path \lstinline{[2,0]} points to \lstinline{"a neg"}.

We define a collection of links between two trees as a relation of type $\ensuremath{\Conid{Region}} \sim \ensuremath{\Conid{Region}}$, where $\ensuremath{\Conid{Region}} = \ensuremath{\Conid{Pattern}} \times \ensuremath{\Conid{Path}}$: a region is identified by a path leading to a subtree and a pattern describing the part of the subtree included in the region.
Briefly, a link is a pair of regions, and a collection of links is a relation between regions of two trees.
For brevity we will write \ensuremath{\Conid{Links}} for $\ensuremath{\Conid{Region}} \sim \ensuremath{\Conid{Region}}$.


An arbitrary collection of links may not make sense for a given pair of trees though---a region mentioned by some link may not exist in the trees at all.
We should therefore characterise when a collection of links is valid for two trees.
\begin{definition}[Region Containment]
\label{def:sat}
For a tree~\ensuremath{\Varid{t}} and a set of regions $\Phi \subseteq \ensuremath{\Conid{Region}}$, we say that $t \models \Phi$ (read `\ensuremath{\Varid{t}} contains $\Phi$') exactly when
\[ \forall (\ensuremath{\Varid{pat}}, \ensuremath{\Varid{path}}) \in \Phi. \quad \ensuremath{\Varid{sel}}(t, \ensuremath{\Varid{path}}) \text{ matches } \ensuremath{\Varid{pat}}\text. \]
\end{definition}
\begin{definition}[Valid Links]
\label{def:validLinks}
Given $\ensuremath{\Varid{ls}} : \ensuremath{\Conid{Links}}$ and two trees \ensuremath{\Varid{t}}~and~\ensuremath{\Varid{u}}, we say that \ensuremath{\Varid{ls}} is \emph{valid} for \ensuremath{\Varid{t}} and \ensuremath{\Varid{u}}, denoted by $\validLink{\ensuremath{\Varid{t}}}{\ensuremath{\Varid{u}}}{\ensuremath{\Varid{ls}}}$, exactly when
\[ \ensuremath{\Varid{t}} \models \ldom(\ensuremath{\Varid{ls}}) \quad\text{and}\quad \ensuremath{\Varid{u}} \models \rdom(\ensuremath{\Varid{ls}}) \text{.} \]
\end{definition}



Now we have all the ingredients for the formal definition of retentive lenses.
\begin{definition}[Retentive Lenses]
\label{def:retLens}
For a set~\ensuremath{\Conid{S}} of source trees and a set~\ensuremath{\Conid{V}} of view trees, a retentive lens between \ensuremath{\Conid{S}}~and~\ensuremath{\Conid{V}} is a pair of functions
\begin{align*}
\ensuremath{\Varid{get}} &: \ensuremath{\Conid{S}} \pfun \ensuremath{\Conid{V}} \times \ensuremath{\Conid{Links}} \\
\ensuremath{\Varid{put}} &: \ensuremath{\Conid{S}} \times \ensuremath{\Conid{V}} \times \ensuremath{\Conid{Links}} \pfun S
\end{align*}
satisfying
\begin{itemize}
\item \emph{Hippocraticness:} if $\ensuremath{\Varid{get}\;\Varid{s}\mathrel{=}(\Varid{v},\Varid{ls})}$, then $(\ensuremath{\Varid{s}}, \ensuremath{\Varid{v}}, \ensuremath{\Varid{ls}}) \in \dom(\ensuremath{\Varid{put}})$ and
\begin{align}
\label{law:get}
\ensuremath{\Varid{put}\;(\Varid{s},\Varid{v},\Varid{ls})\mathrel{=}\Varid{s}} \text{ ;}
\end{align}

\item \emph{Correctness:} if $\ensuremath{\Varid{put}\;(\Varid{s},\Varid{v},\Varid{ls})\mathrel{=}\Varid{s'}}$, then $\ensuremath{\Varid{s'}} \in \dom(\ensuremath{\Varid{get}})$ and
\begin{align}
\label{law:correct}
\ensuremath{\Varid{get}\;\Varid{s'}\mathrel{=}(\Varid{v},\Varid{ls'})} \quad \text{ for some \ensuremath{\Varid{ls'}} ;}
\end{align}

\item \emph{Retentiveness:}
\begin{align}
\label{law:retain}
\ensuremath{\Varid{fst}} \cdot \ensuremath{\Varid{ls}} \subseteq \ensuremath{\Varid{fst}} \cdot \ensuremath{\Varid{ls'}}
\end{align}
where $\ensuremath{\Varid{fst}} : A \sim A \times B$ is the first projection function (lifted to a relation).
\end{itemize}
\end{definition}
Modulo the handling of links, Hippocraticness and Correctness remain the same as their original forms (in the definition of well-behaved lenses).
Retentiveness further states that the input links $\ensuremath{\Varid{ls}}$ must be preserved, except for the location of source regions (i.e.~$\rdom(\ensuremath{\Varid{snd}} \cdot \ensuremath{\Varid{ls}})$ in the compact relational notation). 
The region patterns (data) and the location of the view region, which are $\ensuremath{\Varid{fst}} \cdot \ensuremath{\Varid{ls}}$ in the relational notation, must be exactly the same.
Retentiveness formalises the triangular guarantee in a compact way, and we can expand it pointwise to see that it indeed specialises to the triangular guarantee.

\begin{proposition}[Triangular Guarantee]
Given a retentive lens, suppose \ensuremath{\Varid{put}\;(\Varid{s},\Varid{v},\Varid{ls})\mathrel{=}\Varid{s'}} and \ensuremath{\Varid{get}\;\Varid{s'}\mathrel{=}(\Varid{v},\Varid{ls'})}. If $((\ensuremath{\Varid{spat}}, \ensuremath{\Varid{spath}}), (\ensuremath{\Varid{vpat}}, \ensuremath{\Varid{vpath}})) \in \ensuremath{\Varid{ls}}$, then for some \ensuremath{\Varid{spath'}} we have $\ensuremath{\Varid{s'}} \models \ensuremath{\{\mskip1.5mu (\Varid{spat},\Varid{spath'})\mskip1.5mu\}}$ and $\ensuremath{((\Varid{spat},\Varid{spath'}),(\Varid{vpat},\Varid{vpath}))} \in \ensuremath{\Varid{ls'}}$.
\end{proposition}

\begin{example}
In \autoref{fig:swapAndPut}, if the \lstinline{put} function takes \lstinline{cst}, \lstinline{ast'}, and links \lstinline[mathescape]{ls${}={}$$\{$((Neg "a neg" _  , [2]) , (Sub (Num 0) _ , [0]))$\}$} as arguments and successfully produces an updated source~\lstinline{s'}, then \lstinline{get s'} will succeed.
Let \lstinline[mathescape]{(v,ls')${}={}$get s'}; we know that we can find a link in \lstinline{ls'} with the path of its source region removed: \lstinline[mathescape]{c${}={}$(Neg "a neg" _ , (Sub (Num 0) _ , [0]))${}\in{}$fst${}\cdot{}$ls'}.
So the view region referred to by \lstinline{c} is indeed the same as the one referred to by the input link, and having \lstinline{c}${}\in{}$\lstinline{fst}${}\cdot{}$\lstinline{ls'} means that the region in~\lstinline{s'} corresponding to the view region will match the pattern \lstinline{Neg "a neg" _}\,.
\end{example}

Finally, we note that retentive lenses are an extension of well-behaved lenses: every well-behaved lens between trees can be directly turned into a retentive lens (albeit in a trivial way).
\begin{example}[Well-behaved Lenses are Retentive Lenses]
Given a well-behaved lens defined by $g : \ensuremath{\Conid{S}} \rightarrow \ensuremath{\Conid{V}}$ and $p : \ensuremath{\Conid{S}} \times \ensuremath{\Conid{V}} \rightarrow \ensuremath{\Conid{S}}$, we define $\ensuremath{\Varid{get}} : \ensuremath{\Conid{S}} \pfun \ensuremath{\Conid{V}} \times \ensuremath{\Conid{Links}}$ and $\ensuremath{\Varid{put}} : \ensuremath{\Conid{S}} \times \ensuremath{\Conid{V}} \times \ensuremath{\Conid{Links}} \pfun \ensuremath{\Conid{S}}$ as follows:
\[ \begin{array}{llcl}
\ensuremath{\Varid{get}} & \ensuremath{\Varid{s}}          & = & \ensuremath{(\Varid{g}\;\Varid{s},\emptyset)} \\
\ensuremath{\Varid{put}} & \ensuremath{(\Varid{s},\Varid{v},\Varid{ls})} & = & \ensuremath{\Varid{p}\;(\Varid{s},\Varid{v})} \ \text{.}
\end{array} \]
In the definition, $\dom(\ensuremath{\Varid{put}})$ is restricted to $\big\{\, (\ensuremath{\Varid{s}}, \ensuremath{\Varid{v}}, \emptyset) \,\big\}$.
Hippocraticness and Correctness hold because the underlying $g$ and $p$ are well-behaved.
Retentiveness is also satisfied vacuously since the input link of \ensuremath{\Varid{put}} is empty.
\end{example}

\subsection{Composition of Retentive Lenses}
\label{sec:lensComp}

\begin{figure}[t]
\centering
\includegraphics[scale=0.75,trim={4cm 7.9cm 10cm 3.9cm},clip]{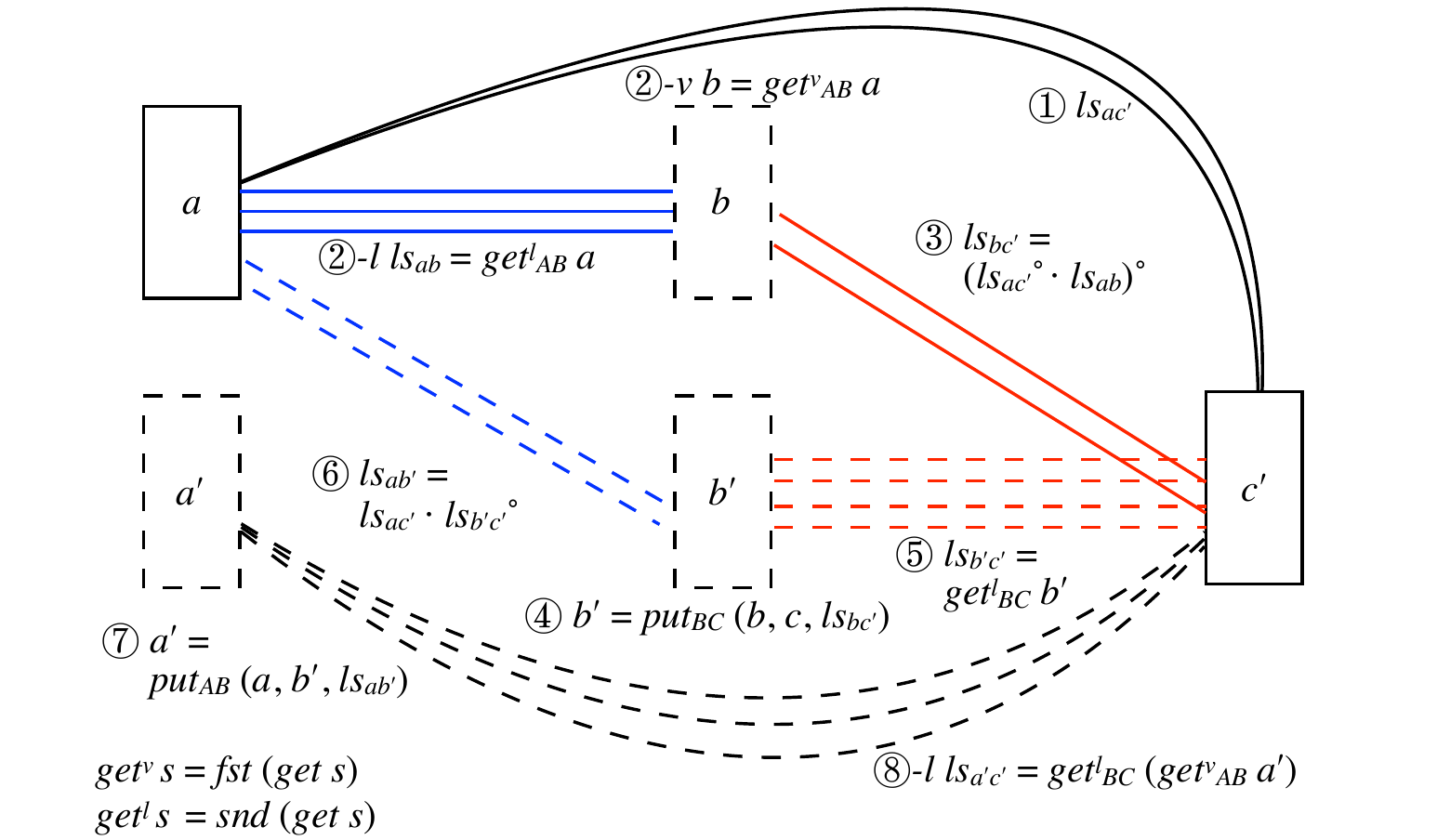}
\caption{The \ensuremath{\Varid{put}} behaviour of a composite retentive lens, divided into steps \numcircledmod{1} to \numcircledmod{8}. Step \numcircledmod{8} produces consistency links for showing the triangular guarantee.}
\label{fig:rlensComp}
\end{figure}

It is standard to provide a composition operator for composing large lenses from small ones.
Here we discuss this operator for retentive lenses, which basically follows the definition of composition for well-behaved lenses, except that we need to deal with links carefully.
Below we use \ensuremath{lens_{AB}} to denote a retentive lens that synchronises trees of sets \ensuremath{\Conid{A}} and \ensuremath{\Conid{B}}, \ensuremath{get_{AB}} and \ensuremath{put_{AB}} the \ensuremath{\Varid{get}} and \ensuremath{\Varid{put}} functions of the lens, \ensuremath{l_{ab}} a link between tree \ensuremath{\Varid{a}} (of set \ensuremath{\Conid{A}}) and tree \ensuremath{\Varid{b}} (of set \ensuremath{\Conid{B}}), and \ensuremath{ls_{ab}} a collection of links between \ensuremath{\Varid{a}} and \ensuremath{\Varid{b}}.

\begin{definition}[Retentive Lens Composition]
\label{def:rlensComp}
Given two retentive lenses \ensuremath{lens_{AB}} and \ensuremath{lens_{BC}}, define the \ensuremath{\Varid{get}} and \ensuremath{\Varid{put}} functions of their composition by\\
\begin{minipage}[t]{0.4\textwidth}
\begin{center}
\begin{hscode}\SaveRestoreHook
\column{B}{@{}>{\hspre}l<{\hspost}@{}}%
\column{3}{@{}>{\hspre}l<{\hspost}@{}}%
\column{10}{@{}>{\hspre}l<{\hspost}@{}}%
\column{21}{@{}>{\hspre}l<{\hspost}@{}}%
\column{33}{@{}>{\hspre}l<{\hspost}@{}}%
\column{E}{@{}>{\hspre}l<{\hspost}@{}}%
\>[B]{}get_{AC}\;\Varid{a}\mathrel{=}(\Varid{c},ls_{ab}\;{\cdot}\;ls_{bc}){}\<[E]%
\\
\>[B]{}\hsindent{3}{}\<[3]%
\>[3]{}\mathbf{where}\;{}\<[10]%
\>[10]{}(\Varid{b},ls_{ab}{}\<[21]%
\>[21]{})\mathrel{=}get_{AB}\;{}\<[33]%
\>[33]{}\Varid{a}{}\<[E]%
\\
\>[10]{}(\Varid{c},ls_{bc}{}\<[21]%
\>[21]{})\mathrel{=}get_{BC}\;{}\<[33]%
\>[33]{}\Varid{b}{}\<[E]%
\ColumnHook
\end{hscode}\resethooks
\end{center}
\end{minipage}
\begin{minipage}[t]{0.6\textwidth}
\begin{hscode}\SaveRestoreHook
\column{B}{@{}>{\hspre}l<{\hspost}@{}}%
\column{3}{@{}>{\hspre}l<{\hspost}@{}}%
\column{12}{@{}>{\hspre}l<{\hspost}@{}}%
\column{24}{@{}>{\hspre}c<{\hspost}@{}}%
\column{24E}{@{}l@{}}%
\column{27}{@{}>{\hspre}l<{\hspost}@{}}%
\column{30}{@{}>{\hspre}l<{\hspost}@{}}%
\column{E}{@{}>{\hspre}l<{\hspost}@{}}%
\>[B]{}put_{AC}\;(\Varid{a},\Varid{c'},ls_{ac^{'}})\mathrel{=}{}\<[30]%
\>[30]{}\Varid{a'}{}\<[E]%
\\
\>[B]{}\hsindent{3}{}\<[3]%
\>[3]{}\mathbf{where}\;{}\<[12]%
\>[12]{}(\Varid{b},ls_{ab}){}\<[24]%
\>[24]{}\mathrel{=}{}\<[24E]%
\>[27]{}get_{AB}\;\Varid{a}{}\<[E]%
\\
\>[12]{}ls_{bc^{'}}{}\<[24]%
\>[24]{}\mathrel{=}{}\<[24E]%
\>[27]{}(ls_{ac^{'}}^{\circ}\;{\cdot}\;ls_{ab})^{\circ}{}\<[E]%
\\
\>[12]{}\Varid{b'}{}\<[24]%
\>[24]{}\mathrel{=}{}\<[24E]%
\>[27]{}put_{BC}\;(\Varid{b},\Varid{c'},ls_{bc^{'}}){}\<[E]%
\\
\>[12]{}ls_{b^{'}c^{'}}{}\<[24]%
\>[24]{}\mathrel{=}{}\<[24E]%
\>[27]{}\Varid{fst}\;(get_{BC}\;\Varid{b'}){}\<[E]%
\\
\>[12]{}ls_{ab^{'}}{}\<[24]%
\>[24]{}\mathrel{=}{}\<[24E]%
\>[27]{}ls_{ac^{'}}\;{\cdot}\;ls_{b^{'}c^{'}}^{\circ}{}\<[E]%
\\
\>[12]{}\Varid{a'}{}\<[24]%
\>[24]{}\mathrel{=}{}\<[24E]%
\>[27]{}put_{AB}\;(\Varid{a},\Varid{b'},ls_{ab^{'}}).{}\<[E]%
\ColumnHook
\end{hscode}\resethooks
\end{minipage}
\end{definition}
The \ensuremath{\Varid{get}} behaviour of a composite retentive lens is straightforward; the \ensuremath{\Varid{put}} behaviour, on the other hand, is a little complex and can be best understood with the help of \autoref{fig:rlensComp}.
Let us first recap the composite behaviour of \ensuremath{\Varid{put}} of traditional lenses: in \autoref{fig:rlensComp}, if we need to propagate changes from data \ensuremath{\Varid{c'}} back to data \ensuremath{\Varid{a}} without links, we will first construct the \emph{intermediate} data \ensuremath{\Varid{b}} (by running \ensuremath{get_{AB}\;\Varid{a}}), propagate changes from \ensuremath{\Varid{c'}} to \ensuremath{\Varid{b}} and produce \ensuremath{\Varid{b'}}, and finally use \ensuremath{\Varid{b'}} to update \ensuremath{\Varid{a}}.
The composition of retentive lenses is similar:
besides the intermediate data \ensuremath{\Varid{b}}, we also need to construct intermediate links \ensuremath{ls_{bc^{'}}} (\numcircledmod{3} in the figure) for retaining information when updating \ensuremath{\Varid{b}} to \ensuremath{\Varid{b'}}, so that we can further construct intermediate links \ensuremath{ls_{ab^{'}}} (\numcircledmod{6} in the figure) for retaining information when updating \ensuremath{\Varid{a}} to \ensuremath{\Varid{a'}} using \ensuremath{\Varid{b'}}.

\begin{theorem}
\label{thm:retPrsv}
The composition of two retentive lenses is still a retentive lens.
\end{theorem}
The proof is available in the appendix (\autoref{prf:retPrsv}).

\section{A DSL for Retentive Bidirectional Tree Transformations}
\label{sec:retentiveDSL}

The definition of retentive lenses is somewhat complex, but we can ease the task of constructing retentive lenses with a declarative domain-specific language.
Our DSL is designed to describe consistency relations between algebraic data types, and from each consistency relation defined in the DSL, we can obtain a pair of \ensuremath{\Varid{get}} and \ensuremath{\Varid{put}} functions forming a retentive lens.
Below we will give an overview of the DSL and how retentive lenses are derived from programs in the DSL using the arithmetic expression example (\autoref{sec:dslByExamples}), the syntax (\autoref{sec:syntax}) and semantics (\autoref{sec:DSLSem}) of the DSL, and finally the theorem stating that the generated lenses satisfy the required laws (\autoref{thm:mainthm}).
Due to limited space, we can only provide the proof of the theorem in the appendix~(\autoref{app:proof}), but the essence is given in the last part of \autoref{sec:dslByExamples}.
Also some more programming examples other than syntax tree synchronisation can be found in the appendix (\autoref{sec:programmingExamples}).

\subsection{Overview of the DSL}
\label{sec:dslByExamples}

\begin{figure}[t]
\begin{small}
\begin{minipage}{0.45\textwidth}
\begin{centerTab}
\begin{lstlisting}
Expr <---> Arith
  Plus   _  x  y  ~  Add  x  y
  Minus  _  x  y  ~  Sub  x  y
  FromT  _  t     ~  t
\end{lstlisting}
\end{centerTab}%
\end{minipage}
\begin{minipage}[t]{0.45\textwidth}
\begin{centerTab}
\begin{lstlisting}
Term <---> Arith
  Lit    _  i  ~  Num  i
  Neg    _  r  ~  Sub  (Num 0)  r
  Paren  _  e  ~  e
\end{lstlisting}
\end{centerTab}%
\end{minipage}
\end{small}
\caption{The program in our DSL for synchronising data types defined in \autoref{fig:runningExpDataTypeDef}.}
\label{fig:dsl_example}
\end{figure}

Recall the arithmetic expression example (\autoref{fig:runningExpDataTypeDef}).
In our DSL, we define data types in \scName{Haskell} syntax and describe consistency relations between them that bear some similarity to \ensuremath{\Varid{get}} functions.
For example, the data type definitions for \lstinline{Expr} and \lstinline{Term} written in our DSL remain the same as those in \autoref{fig:runningExpDataTypeDef}, and the consistency relations between them (i.e.~\lstinline{getE} and \lstinline{getT} in \autoref{fig:runningExpDataTypeDef}) are expressed as the ones in \autoref{fig:dsl_example}.
Here we specify two consistency relations similar to \lstinline{getE} and \lstinline{getT}: one between \lstinline{Expr} and \lstinline{Arith}, and the other between \lstinline{Term} and \lstinline{Arith}.
Each consistency relation is further defined by a set of \emph{inductive rules}, stating that if the subtrees matched by the same variable appearing on the left-hand side (i.e.~source side) and right-hand side (i.e.~view side) are consistent, then the larger pair of trees constructed from these subtrees are also consistent.
Take
\begin{centerTab}
\begin{lstlisting}
Plus _ x y  ~  Add x y
\end{lstlisting}
\end{centerTab}%
for example:
it means that if \ensuremath{\Varid{x}_{\Varid{s}}} is consistent with \ensuremath{\Varid{x}_{\Varid{v}}}, and \ensuremath{\Varid{y}_{\Varid{s}}} is consistent with \ensuremath{\Varid{y}_{\Varid{v}}}, then \lstinline{Plus}~\ensuremath{\Varid{a}\;\Varid{x}_{\Varid{s}}\;\Varid{y}_{\Varid{s}}} and \lstinline{Add}~\ensuremath{\Varid{x}_{\Varid{v}}\;\Varid{y}_{\Varid{v}}} are consistent for any value~\ensuremath{\Varid{a}}, where \ensuremath{\Varid{a}}~corresponds to the `don't-care' wildcard in \lstinline{Plus _ x y}. So the meaning of \lstinline{Plus _ x y ~ Add x y} can be better understood as the following proof rule:
\[ \frac{\ensuremath{\Varid{x}_{\Varid{s}}\mathrel\sim\Varid{x}_{\Varid{v}}} \quad \ensuremath{\Varid{y}_{\Varid{s}}\mathrel\sim\Varid{y}_{\Varid{v}}}}
        {\text{\lstinline{Plus}\;\ensuremath{\Varid{a}\;\Varid{x}_{\Varid{s}}\;\Varid{y}_{\Varid{s}}\mathrel\sim} \lstinline{Add}\;\ensuremath{\Varid{x}_{\Varid{v}}\;\Varid{y}_{\Varid{v}}}}} \]

Each consistency relation is translated to a pair of \ensuremath{\Varid{get}} and \ensuremath{\Varid{put}} functions defined by case analysis generated from the inductive rules.
Detail of the translation will be given in \autoref{sec:DSLSem}, but the idea behind the translation is a fairly simple one which establishes Retentiveness by construction.
For \ensuremath{\Varid{get}}, the rules themselves are already close to function definitions by pattern matching, so what we need to add is only the computation of output links.
For \ensuremath{\Varid{put}}, we use the rules backwards and define a function that turns the regions of an input view into the regions of the new source, reusing regions of the old source wherever required: when there is an input link connected to the current view region, \ensuremath{\Varid{put}} grabs the source region at the other end of the link in the old source; otherwise, \ensuremath{\Varid{put}} creates a new source region as described by the left-hand side of an appropriate rule.

For example, suppose that the \ensuremath{\Varid{get}} and \ensuremath{\Varid{put}} functions generated from the consistency relation \lstinline{Expr <---> Arith} are named \lstinline{getEA} and \lstinline{putEA} respectively.
The inductive rule \lstinline{Plus _ x y ~ Add x y} generates the definition for \lstinline{getEA s} when \lstinline{s} matches \lstinline{Plus _ x y}:
\lstinline{getEA (Plus _ x y)} computes a view recursively in the same way as \lstinline{getE} in \autoref{fig:runningExpDataTypeDef}; furthermore, it produces a new link between the top regions \lstinline{Plus} and \lstinline{Add}, and keeps the links produced by the recursive calls \lstinline{getEA x} and \lstinline{getEA y}.
In the \ensuremath{\Varid{put}} direction, the inductive rule \lstinline{Plus _ x y ~ Add x y} leads to a case \lstinline{putEA s (Add x y) ls}, under which there are two subcases:
if there is any link in \lstinline{ls} that is connected to the \lstinline{Add} region at the top of the view, \lstinline{putEA} grabs the region at the other end of the link in the old source and tries to use it as the top part of the new source;
if such a link does not exist, \lstinline{putEA} uses a \lstinline{Plus} with a default annotation as a substitute for the top part of the new source.
In either case, the subtrees of the new source at the positions marked by \lstinline{x} and \lstinline{y} are computed recursively from the view subtrees \lstinline{x} and \lstinline{y}.

While the core idea is simple, there are cases in which the translated functions do not constitute valid retentive lenses, and the crux of \autoref{thm:mainthm} is finding suitable ways of computation or reasonable conditions to circumvent all such cases (some of which are rather subtle).
The following cases should give a good idea of what is involved in the correctness of the theorem.
\begin{enumerate}[label=\Roman*.,leftmargin=2em]
\item \emph{The translated functions may not be well-defined.}
For example, in the \ensuremath{\Varid{get}} direction, an arbitrary set of rules may assign zero or more than one view to a source, making \ensuremath{\Varid{get}} partial (which, though allowed by the definition, we want to avoid) or ill-defined, and we will impose (fairly standard) restrictions on patterns to preclude such rules.
These restrictions are sufficient to guarantee that exactly one rule is applicable in the \ensuremath{\Varid{get}} direction but not in the \ensuremath{\Varid{put}} direction, in which we need to carefully choose a rule among the applicable ones or risk non-termination (e.g.~producing an infinite number of parentheses by alternating between the \lstinline{Paren} and \lstinline{FromT} rules).
\item \emph{A region grabbed by \ensuremath{\Varid{put}} from the old source may not have the right type.}
For example, if \ensuremath{\Varid{put}} is run on \lstinline{cst}, \lstinline{ast'}, and the link between them in \autoref{fig:swapAndPut}, it has to grab the source region \lstinline{Reg "a neg" _}\,, which has type \lstinline{Term}, and install it as the second argument of \lstinline{Plus}, which has to be of type \lstinline{Expr}.
In this case there is a way out since we can convert a \lstinline{Term} to an \lstinline{Expr} by wrapping the \lstinline{Term} in the \lstinline{FromT} constructor.
We will formulate conditions under which such conversions are needed and can be synthesised automatically.
\item \emph{Hippocraticness may be accidentally invalidated by \ensuremath{\Varid{put}}.}
Suppose that there is another parenthesis constructor \lstinline{Brac} that has the same type as \lstinline{Paren} and for which a similar rule \lstinline{Brac _ e ~ e} is supplied.
Given a source that starts with \lstinline{Brac "" (Paren "" ...)}, \ensuremath{\Varid{get}} will produce two links (among others) relating both the \lstinline{Brac} and \lstinline{Paren} regions with the empty region at the top of the view.
If \ensuremath{\Varid{put}} is immediately invoked on the same source, view, and links, it may choose to process the link attached to the \lstinline{Paren} region first rather than the one attached to the \lstinline{Brac} region, so that the new source starts with \lstinline{Paren "" (Brac "" ...)}, invalidating Hippocraticness.
Therefore \ensuremath{\Varid{put}} has to carefully process the links in the right order for Hippocraticness to hold.
\item \emph{Retentiveness may be invalidated if \ensuremath{\Varid{put}} does not correctly reject invalid input links.}
Unlike \ensuremath{\Varid{get}}, which can easily be made total, \ensuremath{\Varid{put}} is inherently partial since input links may well be invalid and make Retentiveness impossible to hold.
For example, if there is an input link relating a \lstinline{Neg} region and an \lstinline{Add} region, then it is impossible for \ensuremath{\Varid{put}} to produce a result that satisfies Retentiveness since \ensuremath{\Varid{get}} does not produce a link of this form.
Instead, \ensuremath{\Varid{put}} must correctly reject invalid links for Retentiveness to hold.
Apart from checking that input links have the right forms as specified by the rules, there are more subtle cases where the view regions referred to by a set of input links are overlapping---for example, in a view starting with \lstinline{Sub (Num 0) ...} there can be links referring to both the \lstinline{Sub _ _} region and the \lstinline{Sub (Num 0) _} region at the top.
Our \ensuremath{\Varid{get}} cannot produce overlapping view regions, and therefore such input links must be detected and rejected as well.
\end{enumerate}

In the rest of this section we will describe the DSL in more detail.

\subsection{Syntax}
\label{sec:syntax}

\begin{figure}[t]
\centerline{
\framebox[10cm]{\vbox{\hsize=14cm
\[
\begin{array}{llllllll}
\colorbox{light-gray}{\textsf{Program}} \\
\bb
  \ensuremath{\Conid{Prog}} &\Coloneqq& \ensuremath{\Conid{TypeDef}}^* \ \ensuremath{\Conid{RelDef}}^+ \\
\ee \\[0.5em]
\colorbox{light-gray}{\textsf{Type Definition}} \\
\bb
  \ensuremath{\Conid{TypeDef}}  &\Coloneqq& \texttt{data} \ \ensuremath{\Conid{Type}}\ \texttt{=}\ \ensuremath{\Conid{Con}} \ \ensuremath{\Conid{Type}}^* \ \{ \texttt{|}\ \ensuremath{\Conid{Con}}\ \ensuremath{\Conid{Type}}^* \}^* \\
\ee \\[0.5em]
\colorbox{light-gray}{\textsf{Consistency Relation Definition}} \\
\bb
  \ensuremath{\Conid{RelDef}} &\Coloneqq& \ensuremath{\Conid{Type}}_s \longleftrightarrow \ensuremath{\Conid{Type}}_v\ \ensuremath{\Conid{Rule}}^+
\ee \\[0.5em]
\colorbox{light-gray}{\textsf{Inductive Rule}}\\
\bb
   \ensuremath{\Conid{Rule}}  &\Coloneqq& \ensuremath{\Conid{Pat}}_s \ \texttt{\textasciitilde}\ \ensuremath{\Conid{Pat}}_v \\
\ee \\[0.5em]
\colorbox{light-gray}{\textsf{Pattern}}\\
\bb
  \ensuremath{\Conid{Pat}} &\Coloneqq& \texttt{\_}  \quad|\quad  \ensuremath{\Conid{Var}}  \quad|\quad  \ensuremath{\Conid{Con}}\; \ensuremath{\Conid{Pat}}  \\
\ee
\end{array}
\]
}}}
\caption[Syntax of the DSL.]{Syntax of the DSL.}
\label{fig:dsl_syntax}
\end{figure}

The syntax of our DSL is summarised in \autoref{fig:dsl_syntax}, where nonterminals are in \textit{italic}; terminals are typeset in \texttt{typewriter} font; $\{ \}$ is for grouping; $^?$, $^*$, and~$^+$ represent zero-or-one occurrence, zero-or-more occurrence, and one-or-more occurrence respectively, and \ensuremath{\Conid{Type}}, \ensuremath{\Conid{Con}}, and \ensuremath{\Conid{Var}} are syntactic categories (whose definitions are omitted) for the names of types, constructors, and variables respectively.
We sometimes additionally attach a subscript $s$ or $v$ to a symbol to mean that the symbol is related to sources or views.
A program consists of two parts: data types definitions and consistency relations between these data types.
We adopt the \scName{Haskell} syntax for data type definitions---a data type is defined by specifying a set of data constructors and their argument types.
As for the definitions of consistency relations, each of them starts with \ensuremath{\Conid{Type}_{\Varid{s}}\leftrightarrow\Conid{Type}_{\Varid{v}}}, declaring the source and view types for the relation.
The body of each consistency relation is a list of inductive rules, each of which defined by a pair of source and view patterns \ensuremath{\Conid{Pat}_{\Varid{s}}\mathrel\sim\Conid{Pat}_{\Varid{v}}}, where a pattern can include wildcards, variables, and constructors.

\subsubsection{Syntactic Restrictions}
\label{sec:synres}

We impose some syntactic restrictions to guarantee that programs in our DSL indeed give rise to retentive lenses (\autoref{thm:mainthm}).

On \emph{patterns}, we require (i)~pattern coverage:
for any consistency relation $\ensuremath{\Conid{S}\leftrightarrow\Conid{V}} = \{\, p_i \sim q_i \mid 1 \leq i$ $\leq n \,\}$ defined in a program, $\myset{p_i}$ should cover all possible cases of type \ensuremath{\Conid{S}}, and  $\myset{q_i}$ should cover all cases of type \ensuremath{\Conid{V}}.
We also require (ii)~source pattern disjointness:
any distinct $p_i$ and $p_j$ should not be matched by the same tree.
Finally, (iii)~a bare variable pattern is not allowed on the source side (e.g.~\ensuremath{\Varid{x}\mathrel\sim\Conid{D}\;\Varid{x}}), and (iv) wildcards are not allowed on the view side (e.g.~\ensuremath{\Conid{C}\;\Varid{x}\mathrel\sim\Conid{D}\;\anonymous \;\Varid{x}}), and (v)~the source side and the view side must use exactly the same set of variables.
These conditions ensure that \ensuremath{\Varid{get}} is total and well-defined (ruling out Case~I in \autoref{sec:dslByExamples}).%

To state the next requirement we need a definition: two data types \ensuremath{\Conid{S}_{\mathrm{1}}} and \ensuremath{\Conid{S}_{\mathrm{2}}} defined in a program are \emph{interchangeable} in data type \ensuremath{\Conid{S}} exactly when (i) there are some data type \ensuremath{\Conid{V'}} and \ensuremath{\Conid{V}} for which consistency relations \ensuremath{\Conid{S}_{\mathrm{1}}\leftrightarrow\Conid{V'}}, \ensuremath{\Conid{S}_{\mathrm{2}}\leftrightarrow\Conid{V'}} and \ensuremath{\Conid{S}\leftrightarrow\Conid{V}} are defined in the program, and (ii) \ensuremath{\Conid{S}} may have subterms of type \ensuremath{\Conid{S}_{\mathrm{1}}} and \ensuremath{\Conid{S}_{\mathrm{2}}}, and \ensuremath{\Conid{V}} may have subterms of type \ensuremath{\Conid{V'}}.
%
%
%
If \ensuremath{\Conid{S}_{\mathrm{1}}} and \ensuremath{\Conid{S}_{\mathrm{2}}} are interchangeable, then Case~II (\autoref{sec:dslByExamples}) may happen:
when doing \ensuremath{\Varid{put}} on \ensuremath{\Conid{S}} and \ensuremath{\Conid{V}} there might be input links dictating that values of type~\ensuremath{\Conid{S}_{\mathrm{2}}} should be retained in a context where values of type~\ensuremath{\Conid{S}_{\mathrm{1}}} are expected, or vice versa.
When this happens, we need two-way conversions between \ensuremath{\Conid{S}_{\mathrm{1}}}~and~\ensuremath{\Conid{S}_{\mathrm{2}}}.
%

We choose a simple way to ensure the existence of conversions:%
for any interchangeable types \ensuremath{\Conid{S}_{\mathrm{1}}} and \ensuremath{\Conid{S}_{\mathrm{2}}} with \ensuremath{\Conid{S}_{\mathrm{1}}\leftrightarrow\Conid{V'}} and \ensuremath{\Conid{S}_{\mathrm{2}}\leftrightarrow\Conid{V'}} defined, we require that there exists a sequence of data types in the program
\[ S_1 = T_1,\ T_2,\ \cdots,\ T_{n-1},\ T_n = S_2\]
with $n >= 2$ such that for any $1 \leq i < n$, consistency relation \ensuremath{\Conid{T}_{\Varid{i}}\leftrightarrow\Conid{V'}} is defined and has a rule \ensuremath{\Conid{Pat}_{\Varid{i}}\mathrel\sim\Varid{x}} whose source pattern \ensuremath{\Conid{Pat}_{\Varid{i}}} contains exactly one variable, and its type in \ensuremath{\Conid{Pat}_{\Varid{i}}} is $T_{i+1}$ (we also require such a sequence with the roles of \ensuremath{\Conid{S}_{\mathrm{1}}} and \ensuremath{\Conid{S}_{\mathrm{2}}} switched).
With rule \ensuremath{\Conid{Pat}_{\Varid{i}}\mathrel\sim\Varid{x}}, we immediately get a function $t_i : T_{i+1} \rightarrow T_{i}$ contructing a \ensuremath{\Conid{T}_{\Varid{i}}} from a term \ensuremath{\Varid{v}} of $\ensuremath{\Conid{T}}_{i+1}$ by substituting \ensuremath{\Varid{v}} for \ensuremath{\Varid{x}} in \ensuremath{\Conid{Pat}_{\Varid{i}}} (and filling wildcard positions with default values).
Then we have the needed conversion function:
\begin{equation}\label{equ:inj}
\ensuremath{\Varid{inj}}_{{\ensuremath{\Conid{S}_{\mathrm{2}}}} \rightarrow {\ensuremath{\Conid{S}_{\mathrm{1}}}}@\ensuremath{\Conid{V'}}} = t_{n-1} \circ \cdots \circ t_2 \circ t_1 
\end{equation}
(and similary $\ensuremath{\Varid{inj}}_{{\ensuremath{\Conid{S}_{\mathrm{1}}}} \rightarrow {\ensuremath{\Conid{S}_{\mathrm{2}}}}@\ensuremath{\Conid{V'}}}$).
For example, \lstinline{FromT _ t ~ t} gives rise to a function 
\[\text{\lstinline{inj}}_{\text{\lstinline{Term}} \rightarrow \text{\lstinline{Expr}}@\text{\lstinline{Arith}}}\ x \;=\; \text{\lstinline{FromT ""}}\; x\]
and it can be used to convert \lstinline{Term} to \lstinline{Expr} whenever needed when doing \ensuremath{\Varid{put}} with view type \lstinline{Arith}.

\subsection{Semantics}
\label{sec:DSLSem}

We give the semantics of our DSL in terms of a translation into `pseudo-\scName{Haskell}', where we may replace chunks of \scName{Haskell} code with natural language descriptions to improve readability.
As in \autoref{sec:retentive-lens-definition}, let \ensuremath{\Conid{Tree}} be the set of values of any algebraic data type, and \ensuremath{\Conid{Pattern}} the set of all patterns.
For a pattern $\ensuremath{\Varid{p}} \in \ensuremath{\Conid{Pattern}}$, \ensuremath{\Conid{Vars}\;\Varid{p}} denotes the set of variables in $p$.
For each $\ensuremath{\Varid{v}} \in \ensuremath{\Conid{Vars}\;\Varid{p}}$, \ensuremath{\Conid{TypeOf}\;(\Varid{p},\Varid{v})} is (the set of all values of) the type of $\ensuremath{\Varid{v}}$ in pattern \ensuremath{\Varid{p}}, and \ensuremath{\Varid{path}\;(\Varid{p},\Varid{v})} is the path of variable $\ensuremath{\Varid{v}}$ in pattern $\ensuremath{\Varid{p}}$.
We use the following functions (two of which are dependently typed) to manipulate patterns:
\begin{align*}
\ensuremath{\Varid{isMatch}} &: \phantom{(\ensuremath{\Varid{p}} \in {}} \ensuremath{\Conid{Pattern}} \phantom{)} \times \ensuremath{\Conid{Tree}} \rightarrow \ensuremath{\Conid{Bool}} \\
\ensuremath{\Varid{decompose}} &: (\ensuremath{\Varid{p}} \in \ensuremath{\Conid{Pattern}}) \times \ensuremath{\Conid{Tree}} \pfun \big(\ensuremath{\Conid{Vars}\;\Varid{p}} \rightarrow \ensuremath{\Conid{Tree}}\big) \\
\ensuremath{\Varid{reconstruct}} &: (\ensuremath{\Varid{p}} \in \ensuremath{\Conid{Pattern}}) \times \big(\ensuremath{\Conid{Vars}\;\Varid{p}} \rightarrow \ensuremath{\Conid{Tree}}\big) \pfun \ensuremath{\Conid{Tree}} \\
\ensuremath{\Varid{fillWildcards}} &: \phantom{(\ensuremath{\Varid{p}} \in {}} \ensuremath{\Conid{Pattern}} \phantom{)} \times \ensuremath{\Conid{Tree}} \pfun \ensuremath{\Conid{Pattern}} \\
\ensuremath{\Varid{fillWildcardsWD}} &: \phantom{(\ensuremath{\Varid{p}} \in {}} \ensuremath{\Conid{Pattern}} \phantom{)} \rightarrow \ensuremath{\Conid{Pattern}} \\
\ensuremath{\Varid{eraseVars}} &: \phantom{(\ensuremath{\Varid{p}} \in {}} \ensuremath{\Conid{Pattern}} \phantom{)} \rightarrow \ensuremath{\Conid{Pattern}} ~\text.
\end{align*}
Given a pattern~\ensuremath{\Varid{p}} and a tree~\ensuremath{\Varid{t}}, \ensuremath{\Varid{isMatch}\;(\Varid{p},\Varid{t})} tests whether \ensuremath{\Varid{t}}~matches~\ensuremath{\Varid{p}}.
If the match succeeds, \ensuremath{\Varid{decompose}\;(\Varid{p},\Varid{t})} returns a function mapping every variable in \ensuremath{\Varid{p}} to its corresponding matched subtree of~\ensuremath{\Varid{t}}.
Conversely, \ensuremath{\Varid{reconstruct}\;(\Varid{p},\Varid{f})} produces a tree matching~\ensuremath{\Varid{p}} by replacing every occurrence of $\ensuremath{\Varid{v}} \in \ensuremath{\Conid{Vars}\;\Varid{p}}$ in \ensuremath{\Varid{p}} with \ensuremath{\Varid{f}\;\Varid{v}}, provided that \ensuremath{\Varid{p}}~does not contain any wildcard.
To remove wildcards, we can use \ensuremath{\Varid{fillWildcards}\;(\Varid{p},\Varid{t})} to replace all the wildcards in~\ensuremath{\Varid{p}} with the corresponding subtrees of~\ensuremath{\Varid{t}} (coerced into patterns) when \ensuremath{\Varid{t}}~matches~\ensuremath{\Varid{p}}, or use \ensuremath{\Varid{fillWildcardsWD}} to replace all the wildcards with the default values of their types.
Finally, \ensuremath{\Varid{eraseVars}\;\Varid{p}} replaces all the variables in~\ensuremath{\Varid{p}} with wildcards.
The definitions of these functions are straightforward and omitted here.%

\subsubsection{Get Semantics}
\label{sec:getsem}
For a consistency relation $\ensuremath{\Conid{S}} \leftrightarrow \ensuremath{\Conid{V}}$ defined in our DSL with a set of inductive rules $R = \{\, \ensuremath{\Varid{spat}_{\Varid{k}}} \sim \ensuremath{\Varid{vpat}_{\Varid{k}}} \mid 1 \leq \ensuremath{\Varid{k}} \leq n\,\}$, its corresponding $\ensuremath{\Varid{get}}_{\ensuremath{\Conid{SV}}}$ function has the following type:
\begin{align*}
\ensuremath{\Varid{get}}_{\ensuremath{\Conid{SV}}} : \ensuremath{\Conid{S}} \rightarrow  \ensuremath{\Conid{V}} \times \ensuremath{\Conid{Links}}
\end{align*}
The idea of computing \ensuremath{\Varid{get}\;\Varid{s}} is to use a rule $\ensuremath{\Varid{spat}_{\Varid{k}}} \sim \ensuremath{\Varid{vpat}}_k \in R$ such that~$s$ matches \ensuremath{\Varid{spat}_{\Varid{k}}}---the restrictions on patterns imply that such a rule uniquely exists for all $s$---to generate the top portion of the view with \ensuremath{\Varid{vpat}_{\Varid{k}}}, and then recursively generate subtrees for all variables in \ensuremath{\Varid{spat}_{\Varid{k}}}.
The \ensuremath{\Varid{get}} function also creates links in the recursive procedure: when a rule $\ensuremath{\Varid{spat}_{\Varid{k}}} \sim \ensuremath{\Varid{vpat}}_k \in R$ is used, it creates a link relating the matched parts/regions in the source and view, and extends the paths in the recursively computed links between the subtrees.
In all, the \ensuremath{\Varid{get}} function defined by $R$ is:
\begin{align*}
&\ensuremath{\Varid{get}}_\ensuremath{\Conid{SV}}~s = (\ensuremath{\Varid{reconstruct}}\;(\ensuremath{\Varid{vpat}_{\Varid{k}}}, \ensuremath{\Varid{fst}} \circ \ensuremath{\Varid{vls}}),\; l_{\mathit{root}} \cup \ensuremath{\Varid{links}}) \numberthis \label{equ:get} \\
&\mywhere \text{find } k \text{ such that } \ensuremath{\Varid{spat}_{\Varid{k}}} \sim \ensuremath{\Varid{vpat}}_k \in R \text{ and } \ensuremath{\Varid{isMatch}}(\ensuremath{\Varid{spat}_{\Varid{k}}}, \ensuremath{\Varid{s}}) \\
&\myspace \ensuremath{\Varid{vls}} = (\ensuremath{\Varid{get}} \circ \ensuremath{\Varid{decompose}\;(\Varid{spat}_{\Varid{k}},\Varid{s})}) \in \ensuremath{\Conid{Vars}\;\Varid{spat}_{\Varid{k}}} \rightarrow \ensuremath{\Conid{V}} \times \ensuremath{\Conid{Links}} \\
&\myspace \ensuremath{\Varid{spat'}} = \ensuremath{\Varid{eraseVars}\;(\Varid{fillWildcards}\;(\Varid{spat}_{\Varid{k}},\Varid{s}))} \\
&\myspace l_{\mathit{root}} = \myset{((\ensuremath{\Varid{spat'}} , \ensuremath{[\mskip1.5mu \mskip1.5mu]}) , (\ensuremath{\Varid{eraseVars}\;\Varid{vpat}_{\Varid{k}}}, [\,])) } \\
&\myspace \ensuremath{\Varid{links}} = \big\{\, ((\ensuremath{\Varid{spat}}, \ensuremath{\Varid{path}\;(\Varid{spat}_{\Varid{k}},\Varid{v})\plus \Varid{spath}}), (\ensuremath{\Varid{vpat}}, \ensuremath{\Varid{path}\;(\Varid{vpat}_{\Varid{k}},\Varid{v})\plus \Varid{vpath}})) \\
&\myspace \myindentS\myindentS \mid \ensuremath{\Varid{v}} \in \ensuremath{\Conid{Vars}\;\Varid{vpat}_{\Varid{k}}}, \ensuremath{((\Varid{spat},\Varid{spath}),(\Varid{vpat},\Varid{vpath}))} \in \ensuremath{\Varid{snd}\;(\Varid{vls}\;\Varid{v})} \,\big\} \ \text{.}
\end{align*}
The auxiliary function $\ensuremath{\Varid{path}} : (\ensuremath{\Varid{p}} \in \ensuremath{\Conid{Pattern}}) \times \ensuremath{\Conid{Vars}\;\Varid{p}} \rightarrow \ensuremath{\Conid{Path}}$ returns the path from the root of a pattern to one of its variables, and \ensuremath{(\plus )} is path concatenation.
While the recursive call is written as $\ensuremath{\Varid{get}} \circ \ensuremath{\Varid{decompose}\;(\Varid{spat}_{\Varid{k}},\Varid{s})}$ in the definition above, to be precise, \ensuremath{\Varid{get}} should have different subscripts \ensuremath{\Conid{TypeOf}\;(\Varid{spat}_{\Varid{k}},\Varid{v})} and \ensuremath{\Conid{TypeOf}\;(\Varid{vpat}_{\Varid{k}},\Varid{v})} for different $\ensuremath{\Varid{v}} \in \ensuremath{\Conid{Vars}\;\Varid{spat}_{\Varid{k}}}$.

\subsubsection{Put Semantics}\label{sec:putsem}
For a consistency relation $\ensuremath{\Conid{S}} \leftrightarrow \ensuremath{\Conid{V}}$ defined in our DSL as $R = \{\, \ensuremath{\Varid{spat}_{\Varid{k}}} \sim \ensuremath{\Varid{vpat}_{\Varid{k}}} \mid 1 \leq \ensuremath{\Varid{k}} \leq n\,\}$, its corresponding $\ensuremath{\Varid{put}}_{\ensuremath{\Conid{SV}}}$ function has the following type:
\[\ensuremath{\Varid{put}}_{\ensuremath{\Conid{SV}}} : \ensuremath{\Conid{Tree}} \times \ensuremath{\Conid{V}} \times \ensuremath{\Conid{Links}} \pfun \ensuremath{\Conid{S}} \ \text{.} \]
The source argument of \ensuremath{\Varid{put}} is given the generic type \ensuremath{\Conid{Tree}} since the type of the old source may be different from the type of the result that \ensuremath{\Varid{put}} is supposed to produce.
Given arguments $(\ensuremath{\Varid{s}}, \ensuremath{\Varid{v}}, \ensuremath{\Varid{ls}})$, \ensuremath{\Varid{put}} is defined by two cases depending on whether the root of the view is within a region referred to by the input links, i.e.~whether there is some $(\ensuremath{\anonymous }\, , (\ensuremath{\anonymous }\, , \ensuremath{[\mskip1.5mu \mskip1.5mu]})) \in \ensuremath{\Varid{ls}}$.

\begin{itemize}
\item {
  In the first case where the root of the view is not within any region of the input links, \ensuremath{\Varid{put}} selects a rule $\ensuremath{\Varid{spat}_{\Varid{k}}} \sim \ensuremath{\Varid{vpat}_{\Varid{k}}} \in R$ whose \ensuremath{\Varid{vpat}_{\Varid{k}}} matches \ensuremath{\Varid{v}}---our restriction on view patterns implies that at least one such rule exists for all \ensuremath{\Varid{v}}---and uses \ensuremath{\Varid{spat}_{\Varid{k}}} to build the top portion of the new source:
  wildcards in \ensuremath{\Varid{spat}_{\Varid{k}}} are filled with default values and variables in \ensuremath{\Varid{spat}_{\Varid{k}}} are filled with trees recursively constructed from their corresponding parts of the view.
  \begin{align*}
  &\ensuremath{\Varid{put}}_{\ensuremath{\Conid{SV}}}~\ensuremath{(\Varid{s},\Varid{v},\Varid{ls})} = \ensuremath{\Varid{reconstruct}\;(\Varid{spat}_{\Varid{k}}',\Varid{ss})} \numberthis \label{equ:put1}\\
  &\mywhere \text{find } k \text{ such that } \ensuremath{\Varid{spat}_{\Varid{k}}} \sim \ensuremath{\Varid{vpat}_{\Varid{k}}} \in R \text{ and } \ensuremath{\Varid{isMatch}\;(\Varid{vpat}_{\Varid{k}},\Varid{v})} \\
    &\myspace \phantom{\text{find } k \text{ such }} \text{ and } k \text{ satisfies the extra condition below} \\
  &\myspace \ensuremath{\Varid{vs}} = \ensuremath{\Varid{decompose}\;(\Varid{vpat}_{\Varid{k}},\Varid{v})}\\
  &\myspace \ensuremath{\Varid{ss}} = \lambda\,(t \in \ensuremath{\Conid{Vars}\;\Varid{spat}_{\Varid{k}}}) \rightarrow \\
  &\myspace\myindentS \ensuremath{\Varid{put}}(\ensuremath{\Varid{s}}, \ensuremath{\Varid{vs}\;\Varid{t}}, \ensuremath{\Varid{divide}\;(\Varid{path}\;(\Varid{vpat}_{\Varid{k}},\Varid{t}))}, \ensuremath{\Varid{ls}}) \numberthis \label{equ:rec1} \\
  &\myspace \ensuremath{\Varid{spat}}'_k = \ensuremath{\Varid{fillWildcardsWD}\;\Varid{spat}_{\Varid{k}}} \\
  &\ensuremath{\Varid{divide}\;(\Varid{prefix},\Varid{ls})} = \myset{(\ensuremath{\Varid{r}}_\ensuremath{\Varid{s}}, \ensuremath{(\Varid{vpat},\Varid{vpath})}) \mid (\ensuremath{\Varid{r}}_\ensuremath{\Varid{s}}, (\ensuremath{\Varid{vpat}}, \ensuremath{\Varid{prefix}\plus \Varid{vpath}}) \in \ensuremath{\Varid{ls}})}\numberthis \label{equ:divide} 
  \end{align*}
  The omitted subscripts of \ensuremath{\Varid{put}} in (\ref{equ:rec1}) are \ensuremath{\Conid{TypeOf}\;(\Varid{spat}_{\Varid{k}},\Varid{t})} and \ensuremath{\Conid{TypeOf}\;(\Varid{vpat}_{\Varid{k}},\Varid{t})}.
  Additionally, if there is more than one rule whose view pattern matches \ensuremath{\Varid{v}}, the first rule whose view pattern is \emph{not} a bare variable pattern is preferred for avoiding infinite recursive calls: if \ensuremath{\Varid{vpat}_{\Varid{k}}\mathrel{=}\Varid{x}}, the size of the input of the recursive call in~(\ref{equ:rec1}) does not decrease because \ensuremath{\Varid{vs}\;\Varid{t}\mathrel{=}\Varid{v}} and \ensuremath{\Varid{path}\;(\Varid{t},\Varid{vpat}_{\Varid{k}})\mathrel{=}[\mskip1.5mu \mskip1.5mu]}.
    For example, when the view patterns of both \lstinline{Plus _ x y ~ Add x y} and \lstinline{FromT _ t ~ t} match a view tree, the former is preferred.
    This helps to avoid non-termination of \ensuremath{\Varid{put}} as mentioned in Case~I in \autoref{sec:dslByExamples}.
}
\item {
  In the case where the root of the view is an endpoint of some link, \ensuremath{\Varid{put}} uses the source region (pattern) of the link as the top portion of the new source.%
  \begin{align*}
  &\ensuremath{\Varid{put}}_{\ensuremath{\Conid{SV}}}~\ensuremath{(\Varid{s},\Varid{v},\Varid{ls})} = \ensuremath{\Varid{inj}}_{\ensuremath{\Conid{TypeOf}\;\Varid{spat}_{\Varid{k}}} \rightarrow \ensuremath{\Conid{S}}@V} \ensuremath{(\Varid{reconstruct}}(\ensuremath{\Varid{spat}_{\Varid{k}}'},\ensuremath{\Varid{ss}}))  \numberthis \label{equ:put2}\\
  &\mywhere l = ((\ensuremath{\Varid{spat}}, \ensuremath{\Varid{spath}}) , (\ensuremath{\Varid{vpat}}, \ensuremath{\Varid{vpath}})) \in \ensuremath{\Varid{ls}} \\
  &\myspace\myindentSS \text{ such that } \ensuremath{\Varid{vpath}} = \ensuremath{[\mskip1.5mu \mskip1.5mu]},\,\ensuremath{\Varid{spath}}\text{ is the shortest } \\
  &\myspace \text{find } k \text{ such that } \ensuremath{\Varid{spat}_{\Varid{k}}} \sim \ensuremath{\Varid{vpat}_{\Varid{k}}} \in R \text{\ \ and\ \ \ensuremath{\Varid{spat}}} \\
  &\myspace\myindentSS \text{is \ensuremath{\Varid{eraseVars}\;(\Varid{fillWildcards}\;(\Varid{spat}_{\Varid{k}},\Varid{t}))} for some } t\\
  &\myspace \ensuremath{\Varid{spat}_{\Varid{k}}'} =  \ensuremath{\Varid{fillWildcards}\;(\Varid{spat}_{\Varid{k}},\Varid{spat})} \\
  &\myspace \ensuremath{\Varid{vs}} = \ensuremath{\Varid{decompose}\;(\Varid{vpat}_{\Varid{k}},\Varid{v})}\\
  &\myspace \ensuremath{\Varid{ss}} = \lambda\,(t \in \ensuremath{\Conid{Vars}\;\Varid{spat}_{\Varid{k}}}) \rightarrow \\
  &\myspace\myindentS \ensuremath{\Varid{put}}(\ensuremath{\Varid{s}}, \ensuremath{\Varid{vs}\;\Varid{t}}, \ensuremath{\Varid{divide}\;(\Varid{path}\;(\Varid{vpat}_{\Varid{k}},\Varid{t}))}, \ensuremath{\Varid{ls}} \setminus \myset{l})) \numberthis \label{equ:rec2}
  \end{align*}
  When there is more than one source region linked to the root of the view, to avoid Case~III in \autoref{sec:dslByExamples}, \ensuremath{\Varid{put}} chooses the source region whose path is the shortest, which ensures that the preserved region patterns in the new source will have the same relative positions as those in the old source, as the following figure shows.
  \begin{center}
  \includegraphics[scale=0.65,trim={6cm 8.4cm 4.5cm 8cm},clip]{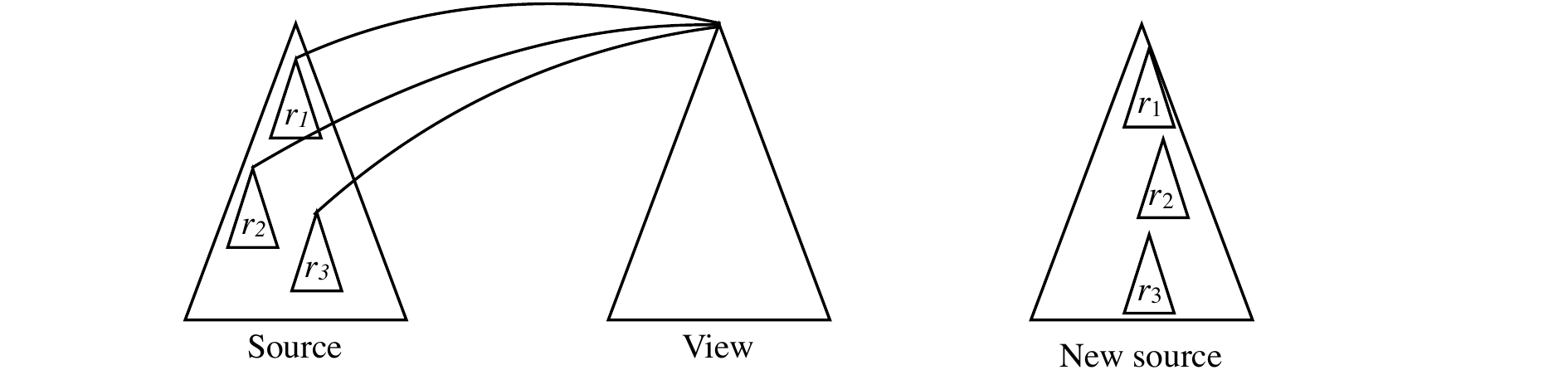}
  \end{center}
    Since the linked source region (pattern) does not necessarily have type \ensuremath{\Conid{S}}, we need to use the function $\ensuremath{\Varid{inj}}_{\ensuremath{\Conid{TypeOf}\;\Varid{spat}_{\Varid{k}}} \rightarrow \ensuremath{\Conid{S}}}@V$ (Equation \ref{equ:inj}) to convert it to type \ensuremath{\Conid{S}}; this function is available due to our requirement on interchangeable data types (see \hyperref[sec:synres]{Syntax Restrictions} in \autoref{sec:syntax}).
}

\end{itemize}

\subsubsection{Domain of \ensuremath{\Varid{put}}}
To avoid Case~IV in \autoref{sec:dslByExamples}, in the actual implementation of \ensuremath{\Varid{put}} there are runtime checks for detecting invalid input links, but these checks are omitted in the above definition of \ensuremath{\Varid{put}} for clarity.
We extract these checks into a separate function \ensuremath{\Varid{check}} below, which also serves as a decision procedure for the domain of \ensuremath{\Varid{put}}.
\begin{align*}
&\ensuremath{\Varid{check}} : \ensuremath{\Conid{Tree}} \times \ensuremath{\Conid{V}} \times \ensuremath{\Conid{Links}} \rightarrow \ensuremath{\Conid{Bool}} \\
&\ensuremath{\Varid{check}\;(\Varid{s},\Varid{v},\Varid{ls})} =
  \begin{cases}
  \ensuremath{\Varid{chkWithLink}\;(\Varid{s},\Varid{v},\Varid{ls})}  \quad & \text{if some } \ensuremath{((\anonymous ,\anonymous ),(\anonymous ,[\mskip1.5mu \mskip1.5mu]))} \in \ensuremath{\Varid{ls}}\\
  \ensuremath{\Varid{chkNoLink}\;(\Varid{s},\Varid{v},\Varid{ls})} & \text{otherwise}
  \end{cases}
\end{align*}
\ensuremath{\Varid{chkNoLink}} corresponds to the first case of \ensuremath{\Varid{put}} (\ref{equ:put1}).
\begin{align*}
&\ensuremath{\Varid{chkNoLink}\;(\Varid{s},\Varid{v},\Varid{ls})} = \ensuremath{\Varid{cond}}_1 \wedge \ensuremath{\Varid{cond}}_2 \wedge \ensuremath{\Varid{cond}}_3 \\
&\mywhere \text{find } k \text{ such that } \ensuremath{\Varid{spat}_{\Varid{k}}} \sim \ensuremath{\Varid{vpat}_{\Varid{k}}} \in R \text{ and } \ensuremath{\Varid{isMatch}\;(\Varid{vpat}_{\Varid{k}},\Varid{v})} \\
  &\myspace\myindentS \text{ and } k \text{ satisfies the same condition as in (\ref{equ:put1})} \\
&\myspace \ensuremath{\Varid{vs}} = \ensuremath{\Varid{decompose}\;(\Varid{vpat}_{\Varid{k}},\Varid{v})}\\
&\myspace \ensuremath{\Varid{vp}\;\Varid{t}} = \ensuremath{\Varid{path}\;(\Varid{vpat}_{\Varid{k}},\Varid{t})} \\
&\myspace \ensuremath{\Varid{cond}}_1 = \ensuremath{\Varid{ls}\mathop{\shorteq\kern 1pt\shorteq}} \left(\bigcup_{t \in \ensuremath{\Conid{Vars}\;\Varid{spat}_{\Varid{k}}}} \ensuremath{\Varid{addVPrefix}\;(\Varid{vp}\;\Varid{t},\Varid{divide}\;(\Varid{vp}\;\Varid{t},\Varid{ls}))}\right) \\
&\myspace \ensuremath{\Varid{cond}}_2 = \bigwedge_{t \in \ensuremath{\Conid{Vars}\;\Varid{spat}_{\Varid{k}}}} \ensuremath{\Varid{check}\;(\Varid{s},\Varid{vs}\;\Varid{t},\Varid{divide}\;(\Varid{vp}\;\Varid{t},\Varid{ls}))} \\
&\myspace \ensuremath{\Varid{cond}}_3 = \text{\textbf{if} } \ensuremath{\Varid{vpat}}_k \text{ is some bare variable pattern `$x$'} \text{ \textbf{then} } \\
&\myspace\qquad\qquad\ \ensuremath{\Conid{TypeOf}\;(\Varid{spat}_{\Varid{k}},\Varid{x})\leftrightarrow\Conid{V}}  \text{ has a rule } \ensuremath{\Varid{spat}_{\Varid{j}}} \sim \ensuremath{\Varid{vpat}_{\Varid{j}}} \text{ such that}\\
&\myspace\qquad\qquad\ \myindentSS \ensuremath{\Varid{isMatch}\;(\Varid{vpat}_{\Varid{j}},\Varid{v})} \text{ and } \ensuremath{\Varid{vpat}_{\Varid{j}}} \text{ is not a bare variable pattern} \\
&\myspace \ensuremath{\Varid{addVPrefix}\;(\Varid{prefix},\Varid{rs})} = \myset{\ensuremath{((\Varid{a},\Varid{b}),(\Varid{c},\Varid{prefix}\plus \Varid{d}))} \mid \ensuremath{((\Varid{a},\Varid{b}),(\Varid{c},\Varid{d}))} \in \ensuremath{\Varid{rs}}}
\end{align*}
The \ensuremath{\Varid{divide}} function is defined as in \autoref{equ:divide}.
Condition \ensuremath{\Varid{cond}_{\mathrm{1}}} checks that every link in \ensuremath{\Varid{ls}} is processed in one of the recursive calls, i.e.\ the path of every view region of \ensuremath{\Varid{ls}} starts with \ensuremath{\Varid{path}\;(\Varid{vpat}_{\Varid{k}},\Varid{t})} for some \ensuremath{\Varid{t}}.
(Specifically, if \ensuremath{\Conid{Vars}\;\Varid{spat}_{\Varid{k}}} is empty, \ensuremath{\Varid{ls}} in \ensuremath{\Varid{cond}_{\mathrm{1}}} should also be empty meaning that all the links have already been processed.)
\ensuremath{\Varid{cond}_{\mathrm{2}}} summarises the results of \ensuremath{\Varid{check}} for recursive calls.
\ensuremath{\Varid{cond}_{\mathrm{3}}} guarantees the termination of recursion: When \ensuremath{\Varid{vpat}_{\Varid{k}}} is a bare variable pattern, the recursive call in \autoref{equ:rec1} does not decrease the size of any of its arguments;
\ensuremath{\Varid{cond}_{\mathrm{3}}} makes sure that such non-decreasing recursion will not happen in the next round\footnote{For presentation purposes we only check two rounds here, but in general we should check \ensuremath{\Conid{N}\mathbin{+}\mathrm{1}} rounds where \ensuremath{\Conid{N}} is the number data types defined in the program.} for avoiding infinite recursive calls.

For \ensuremath{\Varid{chkWithLink}}, as in the corresponding case of \ensuremath{\Varid{put}} (\autoref{equ:put2}), let $\ensuremath{\Varid{l}\mathrel{=}((\Varid{spat},\Varid{spath}),} $
$\ensuremath{(\Varid{vpat},\Varid{vpath}))} \in \ensuremath{\Varid{ls}}$ such that \ensuremath{\Varid{vpath}\mathrel{=}[\mskip1.5mu \mskip1.5mu]} and \ensuremath{\Varid{spath}} is the shortest when there is more than one such link.
\phantomsection\label{def:chkWithLink}
\begin{align*}
&\ensuremath{\Varid{chkWithLink}\;(\Varid{s},\Varid{v},\Varid{ls})} = \ensuremath{\Varid{cond}}_1 \wedge \ensuremath{\Varid{cond}}_2 \wedge \ensuremath{\Varid{cond}}_3 \wedge \ensuremath{\Varid{cond}}_4 \\
& \mywhere \\
& \myindentSS \ensuremath{\Varid{cond}}_1 = \ensuremath{\Varid{isMatch}\;(\Varid{spat},\Varid{sel}\;(\Varid{s},\Varid{spath}))} \wedge \ensuremath{\Varid{isMatch}\;(\Varid{vpat},\Varid{sel}\;(\Varid{v},\Varid{vpath}))} \\
& \myindentSS \ensuremath{\Varid{cond}}_2 = \exists ! (\ensuremath{\Varid{spat}_{\Varid{k}}}, \ensuremath{\Varid{vpat}}_k) \in \ensuremath{\Conid{R}}.\ \ensuremath{\Varid{vpat}} = \ensuremath{\Varid{eraseVars}\;\Varid{vpat}_{\Varid{k}}} \\
& \myindent \wedge \ensuremath{\Varid{spat}} \text{ is \ensuremath{\Varid{eraseVars}\;(\Varid{fillWildcards}\;(\Varid{spat}_{\Varid{k}},\Varid{t}))} for some } t\\
& \myindentSS \ensuremath{\Varid{cond}}_3 = \ensuremath{\Varid{ls}} \ensuremath{\mathop{\shorteq\kern 1pt\shorteq}} (\myset{l} \cup \bigcup_{t \in \ensuremath{\Conid{Vars}\;\Varid{spat}_{\Varid{k}}}} \ensuremath{\Varid{addVPrefix}\;(\Varid{path}\;(\Varid{vpat}_{\Varid{k}},\Varid{t}),}\\
& \myindentSS\myindent\myindent\ensuremath{\Varid{divide}\;(\Varid{path}\;(\Varid{vpat}_{\Varid{k}},\Varid{t}),\Varid{ls}\setminus \myset{l}))}) \\
& \myindentSS \ensuremath{\Varid{cond}}_4 = \bigwedge_{t \in \ensuremath{\Conid{Vars}\;\Varid{spat}_{\Varid{k}}}}\ensuremath{\Varid{check}}~(s, \ensuremath{\Varid{vs}\;\Varid{t}}, \ensuremath{\Varid{divide}}~\ensuremath{(\Varid{path}\;(\Varid{vpat}_{\Varid{k}},\Varid{t}))}, \ensuremath{\Varid{ls}} \setminus \myset{l}))
\end{align*}
\ensuremath{\Varid{cond}_{\mathrm{1}}} makes sure that the link \ensuremath{\Varid{l}} is valid (\autoref{def:validLinks}) and \ensuremath{\Varid{cond}_{\mathrm{2}}} further checks that it can be generated from some rule of the consistency relations.
\ensuremath{\Varid{cond}_{\mathrm{3}}} and \ensuremath{\Varid{cond}_{\mathrm{4}}} are for recursive calls: the latter summarises the results for the subtrees and the former guarantees that no link will be missed.
It is \ensuremath{\Varid{cond}_{\mathrm{3}}} that rejects the subtle case of overlapping view regions as described at the end of Case~IV in \autoref{sec:dslByExamples}.%

\subsubsection{Main Theorem}
We can now state our main theorem in terms of the definitions of \ensuremath{\Varid{get}} and \ensuremath{\Varid{put}} above.

\begin{theorem}
\label{thm:mainthm}
Let \ensuremath{\Varid{put'}} be \ensuremath{\Varid{put}} with its domain intersected with $\ensuremath{\Conid{S}} \times \ensuremath{\Conid{V}} \times \ensuremath{\Conid{Links}}$. Then \ensuremath{\Varid{get}} and \ensuremath{\Varid{put'}} form a retentive lens as in \autoref{def:retLens}.
\end{theorem}
The proof goes by induction on the size of the arguments to \ensuremath{\Varid{put}} or \ensuremath{\Varid{get}} and can be found in the appendix (\autoref{app:proof}).

\section{Edit Operations and Link Maintenance}
\label{sec:editOperations}
Our \ensuremath{\Varid{get}} function only produces horizontal links between a source and its consistent view, while the input links to a \ensuremath{\Varid{put}} function are the ones between a source and a modified view.
To bridge the gap, in this section, we demonstrate how to update the view while maintaining the links using a set of typical \emph{edit operations} (on views).
These edit operations will be used in the three case studies in the next section.



\begin{figure}[t]
\centering
\includegraphics[scale=0.75,trim={7cm 10.17cm 7cm 2cm},clip]{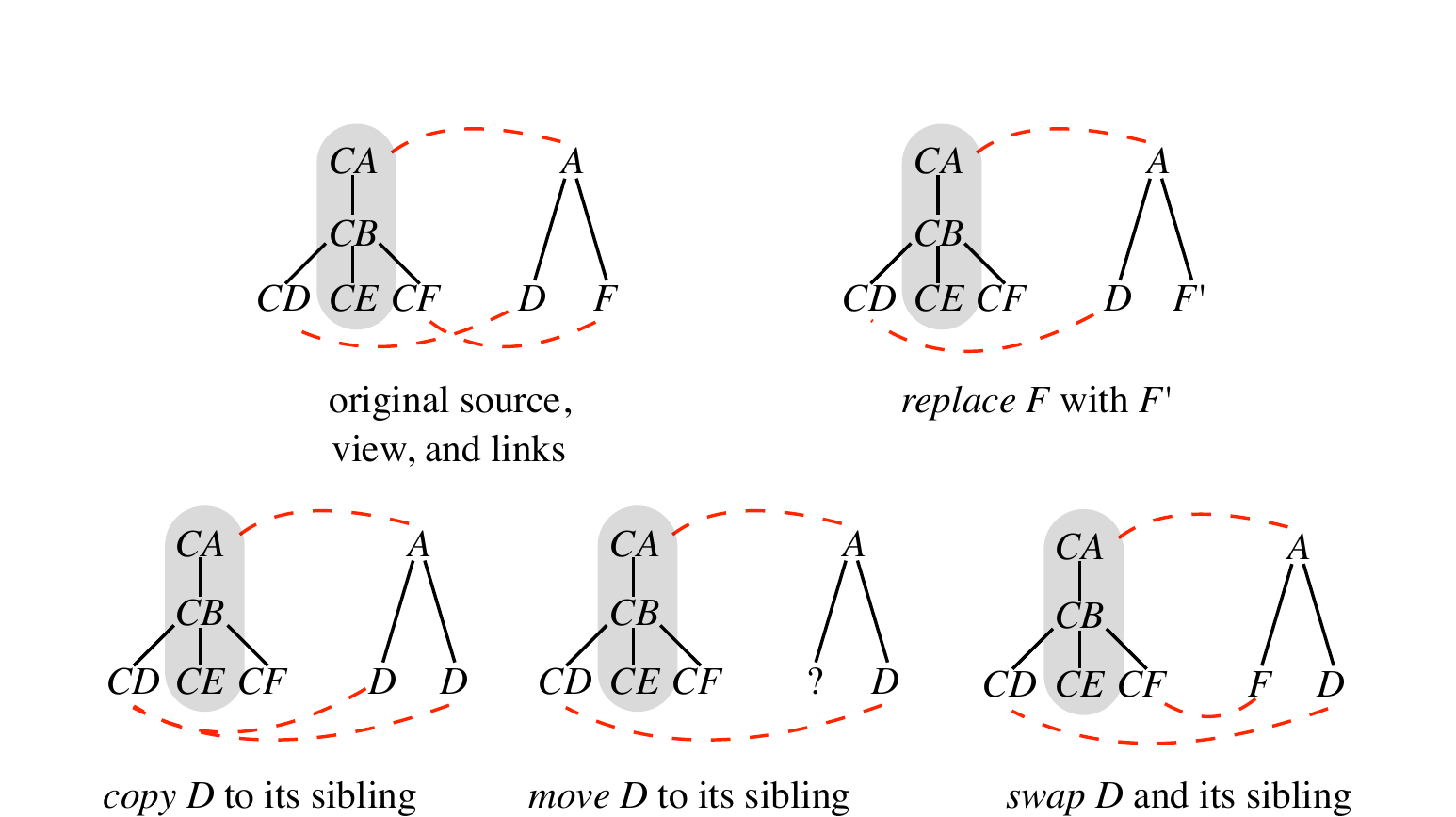}
\caption{How edit operations \ensuremath{\Varid{replace}}, \ensuremath{\Varid{copy}}, \ensuremath{\Varid{move}}, and \ensuremath{\Varid{swap}} main links.}
\label{fig:editOperations}
\end{figure}

We define four edit operations, \ensuremath{\Varid{replace}}, \ensuremath{\Varid{copy}}, \ensuremath{\Varid{move}}, and \ensuremath{\Varid{swap}}, of which \ensuremath{\Varid{move}} and \ensuremath{\Varid{swap}} are defined in terms of \ensuremath{\Varid{copy}} and \ensuremath{\Varid{replace}}.
The edit operations accept not only an AST but also a set of links, which is updated along with the AST.
The interface has been designed in a way that the last argument of an edit operation is the pair of the AST and links, so that the user can use \scName{Haskell}'s ordinary function composition to compose a sequence of edits (partially applied to all the other arguments).
The implementation of the four edit operations takes less than 40 lines of \scName{Haskell} code, as our DSL already generates useful auxiliary functions such as fetching a subtree according to a path in some tree.

We briefly explain how the edit operations update links, as illustrated in \autoref{fig:editOperations}:
Replacing a subtree at path \ensuremath{\Varid{p}} will destroy all the links previously connecting to path \ensuremath{\Varid{p}}.
Copying a subtree from path \ensuremath{\Varid{p}} to path \ensuremath{\Varid{p'}} will duplicate the set of links previously connecting to \ensuremath{\Varid{p}} and redirect the duplicated links to connect to \ensuremath{\Varid{p'}}.
Moving a subtree from \ensuremath{\Varid{p}} to \ensuremath{\Varid{p'}} will destroy links connecting to \ensuremath{\Varid{p'}} and redirect the links (previously) connecting to \ensuremath{\Varid{p}} to connect to \ensuremath{\Varid{p'}}.
Swapping subtrees at \ensuremath{\Varid{p}} and \ensuremath{\Varid{p'}} will also swap the links connecting to \ensuremath{\Varid{p}} and \ensuremath{\Varid{p'}}.


\section{Case Studies}
\label{sec:application}
We demonstrate how our DSL works for the problems of code refactoring~\cite{Fowler1999Refactoring}, resugaring~\cite{Pombrio2014Resugaring, Pombrio2015Hygienic}, and XML synchronisation~\cite{Pacheco2014BiFluX}, all of which require that we constantly make modifications to ASTs and synchronise them with CSTs.
For all these problems, retentive lenses provide a systematic way for the user to preserve information of interest in the original CST after synchronisation.
The source code for these case studies can be found on the first author's web page: \url{http://www.prg.nii.ac.jp/members/zhu/}.

\subsection{Refactoring}
\label{sec:refactoring}

As we will report below, we have programmed the consistency relations between CSTs and ASTs for a small subset of Java~8~\cite{Gosling2014Java} and tested the generated retentive lens on a particular refactoring.
Even though the case study is small, we believe that our framework is general enough:
We have surveyed the standard set of refactoring operations for Java~8 provided by Eclipse Oxygen (with Java Development Tools) and found that all the 23 refactoring operations can be represented as the combinations of our edit operations defined in \autoref{sec:editOperations}.
A summary can be found in the appendix (\autoref{sec:refactOptAsEditSeq}).





\begin{figure}[t]
\centering
\includegraphics[scale=0.6,trim={4cm 4cm 4cm 3cm},clip]{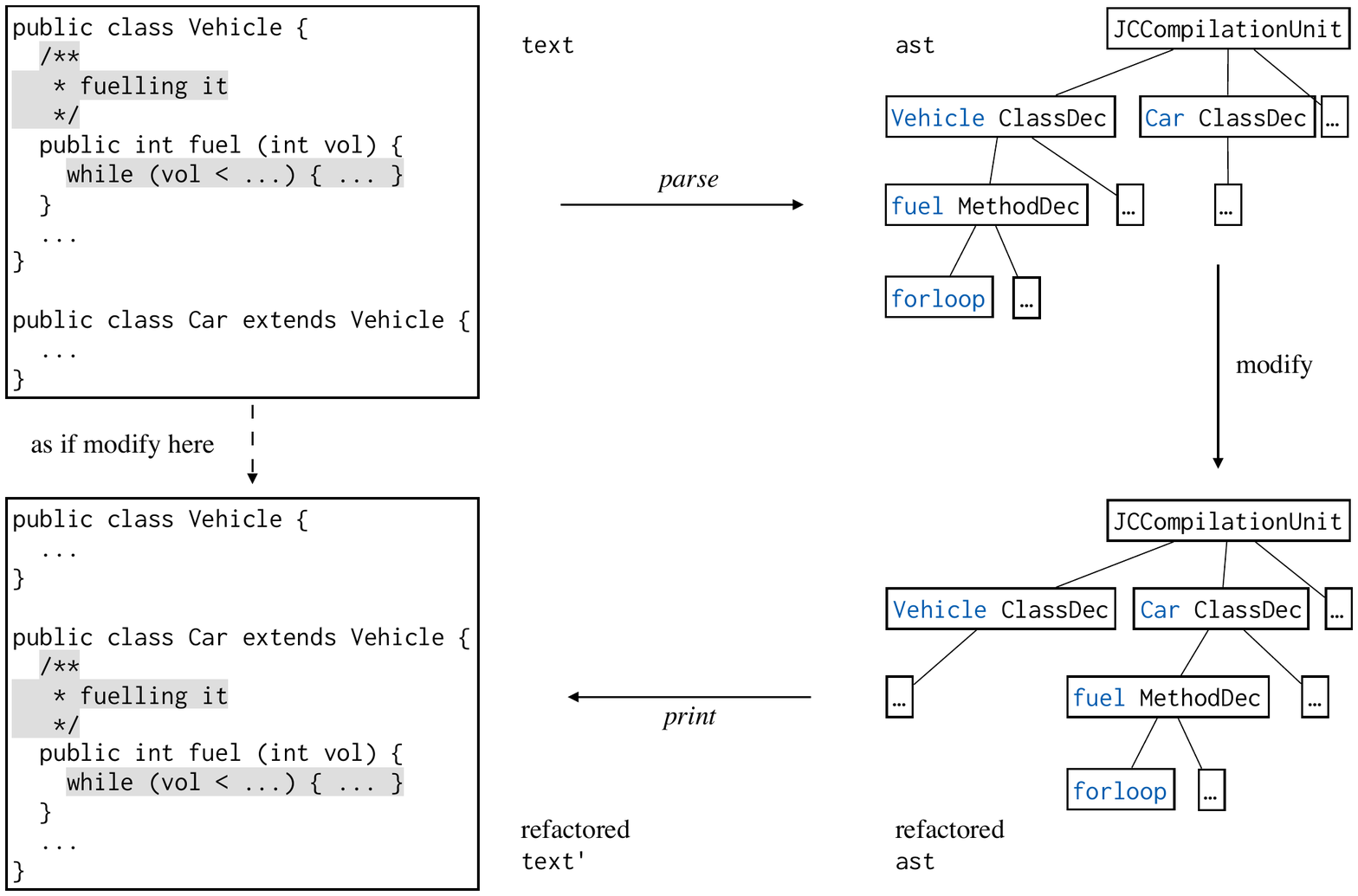}
\caption[An example of the \emph{push-down} code refactoring.]{An example of the \emph{push-down} code refactoring. \footnotesize{(Subclasses \lstinline{Bus} and \lstinline{Bicycle} are omitted due to space limitation.)}}
\label{fig:refactoringTextAndAst}
\end{figure}


\subsubsection{The \emph{Push-Down} Code Refactoring}
An example of the \emph{push-down} code refactoring is illustrated in \autoref{fig:refactoringTextAndAst}.
At first, the user designed a \lstinline{Vehicle} class and thought that it should possess a \lstinline{fuel} method for all the vehicles.
The \lstinline{fuel} method has a JavaDoc-style comment and contains a \lstinline{while} loop, which can be seen as syntactic sugar and is converted to a standard \lstinline{for} loop during parsing.
However, when later designing \lstinline{Vehicle}'s subclasses, the user realises that bicycles cannot be fuelled and decides to do the push-down code refactoring, which removes the \lstinline{fuel} method from \lstinline{Vehicle} and pushes the method definition down to subclasses \lstinline{Bus} and \lstinline{Car} but not \lstinline{Bicycle}.
Instead of directly modifying the (program) \lstinline{text}, most refactoring tools choose to parse the program text into its \lstinline{ast}, perform code refactoring on the \lstinline{ast}, and regenerate new (program) \lstinline{text'}.
The bottom-left corner of \autoref{fig:refactoringTextAndAst} shows the desired (program) \lstinline{text'} after refactoring, where we see that the comment associated with \lstinline{fuel} is also pushed down, and the \lstinline{while} sugar is kept.
However, the preservation of the comment and syntactic sugar does not come for free actually, as the \lstinline{ast}---being a concise and compact representation of the program \lstinline{text}---includes neither comments nor the form of the original \lstinline{while} loop.
So if the user implements the \ensuremath{\Varid{parse}} and \ensuremath{\Varid{print}} functions as back-and-forth conversions between CSTs ASTs (or even as a well-behaved lens), they may produce unsatisfactory results in which the comment and the \lstinline{while} syntactic sugar are lost.

\subsubsection{Implementation in Our DSL}
\label{sec:java8SyncImpl}

Following the grammar of Java~8, we define data types for a simplified version of its concrete syntax, which consists of definitions of classes, methods, and variables; arithmetic expressions (including assignment and method invocation); and conditional and loop statements.
For convenience, we also restrict the occurrence of statements and expressions to exactly once in some cases (such as variable declarations).
Then we define the corresponding simplified version of the abstract syntax that follows the one defined by the JDT parser~\cite{OpenJDK}.
This subset of Java~8 has around 80 CST constructs (production rules) and 30 AST constructs; the 70 consistency relations among them generate about 3000 lines of code for the retentive lenses and auxiliary functions (such as the ones for conversions between interchangeable data types and edit operations).

\subsubsection{Demo}
\label{sec:refactorSysRun}
We can now perform some experiments on \autoref{fig:refactoringTextAndAst}.
\begin{itemize}

\item First we test \lstinline{put cst ast ls}, where \lstinline{(ast, ls) = get cst}.
We get back the same \lstinline{cst}, showing that the generated lenses do satisfy Hippocraticness.

\item As a special case of Correctness, we let \lstinline{cst' = put cst ast []} and check \lstinline{fst (get cst')} \ensuremath{\mathop{\shorteq\kern 1pt\shorteq}} \lstinline{ast}.
In \lstinline{cst'}, the \lstinline{while} loop becomes a basic \lstinline{for} loop and all the comments disappear.
This shows that \ensuremath{\Varid{put}} will create a new source solely from the view if links are missing.

\item Then we change \lstinline{ast} to \lstinline{ast'} and the set of links \lstinline{ls} to \lstinline{ls'} using our edit operations, simulating the \emph{push-down} code refactoring for the \lstinline{fuel} method.
To show the effect of Retentiveness more clearly, when building \lstinline{ast'}, the \lstinline{fuel} method in the \lstinline{Car} class is copied from the \lstinline{Vehicle} class, while the \lstinline{fuel} method in the \lstinline{Bus} class is built from scratch (i.e.~replaced with a `new' \lstinline{fuel} method).
Let \lstinline{cst' = put cst ast' ls'}.
In the \lstinline{fuel} method of the \lstinline{Car} class, the \lstinline{while} loop and its associated comments are preserved; but in the \lstinline{fuel} method of the \lstinline{Bus} class, there is only a \lstinline{for} loop without any associated comments.
This is where Retentiveness helps the user to retain information on demand. Finally, we also check that Correctness holds: \lstinline{fst (get cst')} \ensuremath{\mathop{\shorteq\kern 1pt\shorteq}} \lstinline{ast'}.

\end{itemize}

\subsection{Resugaring}
We have seen syntactic sugar such as negation and \lstinline{while} loops.
The idea of \emph{resugaring} is to print evaluation sequences in a core language using the constructs of its surface syntax (which contains sugar)~\cite{Pombrio2014Resugaring,Pombrio2015Hygienic}.
To solve the problem,~\citeauthor{Pombrio2014Resugaring}~\cite{Pombrio2014Resugaring} enrich the AST to incorporate fields for holding tags that mark from which syntactic object an AST construct comes.
Using retentive lenses, we can also solve the problem while leaving the AST clean---we can write consistency relations between the surface syntax and the abstract syntax and passing the generated \ensuremath{\Varid{put}} function proper links for retaining syntactic sugar, which we have already seen in the arithmetic expression example (where we retain the negation) and in the code refactoring example (where we retain the \lstinline{while} loop).
Both \citeauthor{Pombrio2014Resugaring}'s `tag approach' and our `link approach', in actuality, identifies where an AST construct comes from; however, the link approach has an advantage that it leaves ASTs clean and unmodified so that we do not need to patch up the existing compiler to deal with tags.

\subsection{XML Synchronisation}
\label{sec:XMLSync}

In this subsection, we present a case study on XML synchronisation, which is pervasive in the real world.
The specific example used here is adapted from \citeauthor{Pacheco2014BiFluX}'s paper~\cite{Pacheco2014BiFluX}, where they use their DSL, \scName{BiFluX}, to synchronise address books.

As for their example, both the source address book and the view address book are grouped by social relationships; however, the source address book (defined by \lstinline{AddrBook}) contains names, emails, and telephone numbers whereas the view (social) address book (defined by \lstinline{SocialBook}) contains names only.

To synchronise \lstinline{AddrBook} and \lstinline{SocialBook}, we write consistency relations in our DSL and the core ones are
\begin{centerTab}
\begin{minipage}[t]{.55\textwidth}
\begin{lstlisting}
AddrGroup  <---> SocialGroup
  AddrGroup grp p  ~  SocialGroup grp p

List Person <---> List Name
  Nil  ~  Nil
  Cons p xs  ~  Cons p xs
\end{lstlisting}
\end{minipage}
\begin{minipage}[t]{.45\textwidth}
\begin{lstlisting}
Person <---> Name
  Person t  ~  t

Triple Name Email Tel <---> Name
  Triple name  _ _  ~  name (*@.@*)
\end{lstlisting}
\end{minipage}
\end{centerTab}
\\
The consistency relations will compile to a pair of \lstinline{get} and \lstinline{put}.

As \autoref{fig:XMLSync} shows, the original source is \lstinline{addrBook}
and its consistent view is \lstinline{socialBook}, both of which have two relationship groups: \lstinline{coworkers} and \lstinline{friends}.
The source has a record \lstinline{Person (Triple "Alice" "alice@abc.xyz" "000111")} in the group \lstinline{coworkers}, and we will see how this record changes in the new source after we update the view \lstinline{socialBook} in the following way and propagate the changes back:
we (i) reorder the two groups; (ii) change Alice's group from \lstinline{coworkers} to \lstinline{friends}; (iii) create a new social relationship group \lstinline{family} for family members.

In our case, to produce a new source \lstinline{socialBook'}, we handle the three updates using our basic edit operations (in this case, only \ensuremath{\Varid{swap}}, \ensuremath{\Varid{move}}, and \ensuremath{\Varid{replace}}) which also maintain the links.
Feeding the original source \lstinline{addrBook}, updated view \lstinline{socialBook'} and links \lstinline{hls'} to the (generated) \lstinline{put} function, we obtain the updated \lstinline{addrBook'}.
In \autoref{fig:XMLSync}, it is clearly seen that carefully maintained links help us to preserve email addresses and telephone numbers associated with each person during the \ensuremath{\Varid{put}} process; note that well-behavedness does not guarantee the retention of this information, since the input view is not consistent with the input source in this case.

\begin{figure}[t]
\centering
\includegraphics[scale=0.6,trim={4cm 8cm 4cm 3cm},clip]{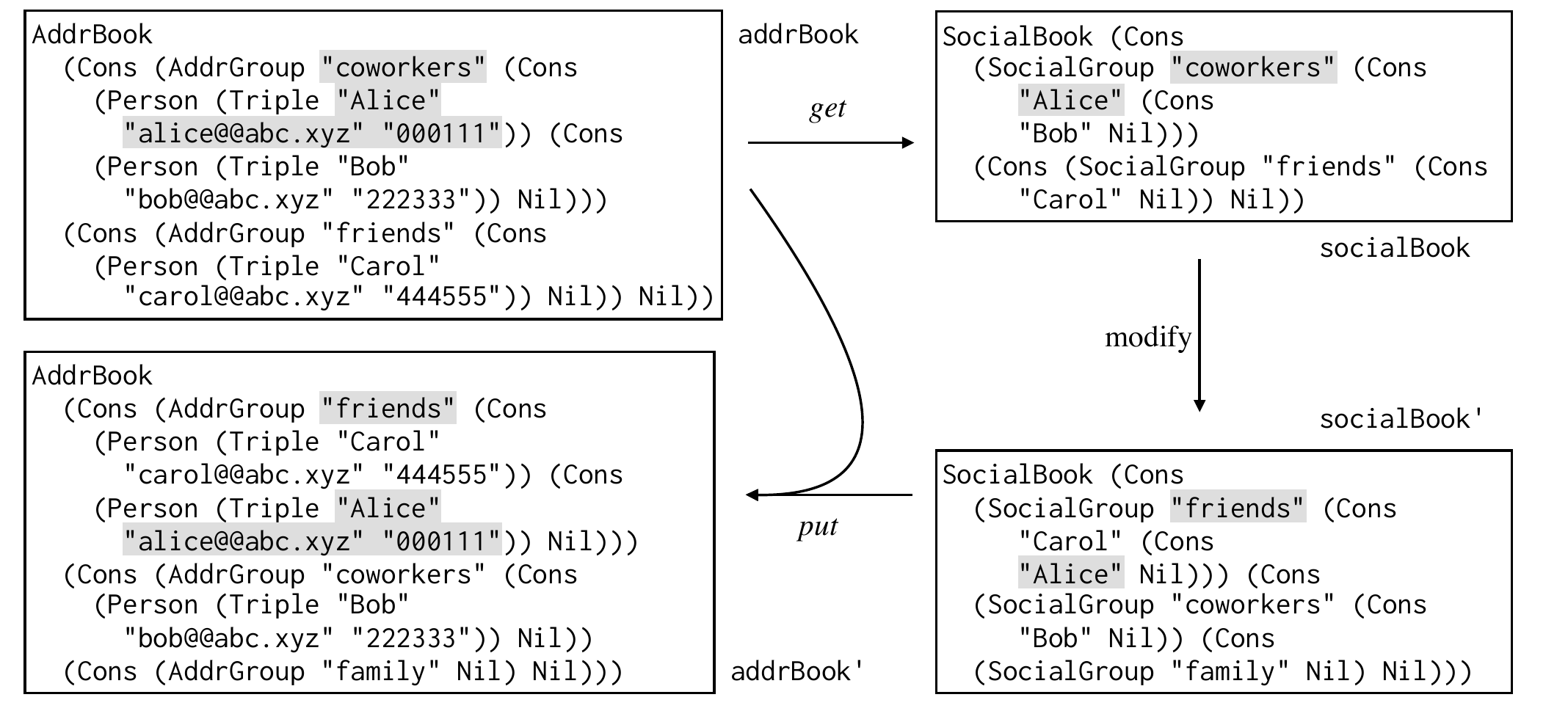}
\caption[An example of XML synchronisation.]{An example of XML synchronisation. (Grey areas highlight how the record \lstinline{Alice} is updated.)}
\label{fig:XMLSync}
\end{figure}

As pointed out by \citeauthor{Pacheco2014BiFluX}, examples of this kind motivate extensions to (combinator-based) alignment-aware languages such as \scName{Boomerang}~\cite{Bohannon2008Boomerang} and matching lenses~\cite{Barbosa2010Matching}.
In fact, it is hard for those languages to handle source-view alignment where some view elements are moved out of its original list-like structure (or \emph{chunk}~\cite{Barbosa2010Matching}) and put into a new list-like structure, probably far away---because when using those languages, we usually lift a lens combinator \ensuremath{\Varid{k}} handling a single element to $\ensuremath{\Varid{k}}^*$ dealing with a list of elements, so that the `scope' of the alignment performed by $\ensuremath{\Varid{k}}^*$ is always within that single list (which it currently works on).

\section{Related Work}
\label{sec:relatedWork}

\subsection{Alignment}
\label{sub:alignment}

\emph{Alignment} has been recognised as an important problem when we need to synchronise two lists. Our work is closely related.

\subsubsection{Alignment for Lists}
The earliest lenses~\cite{Foster2007Combinators} only allow source and view elements to be matched positionally---the $n$-th source element is simply updated using the $n$-th element in the modified view. Later, lenses with more powerful matching strategies are proposed, such as dictionary lenses~\cite{Bohannon2008Boomerang} and their successor matching lenses~\cite{Barbosa2010Matching}.
As for matching lenses, when a \ensuremath{\Varid{put}} is invoked, it will first find the correspondence between \emph{chunks} (data structures that are reorderable, such as lists) of the old and new views using some predefined strategies; based on the correspondence, the chunks in the source are aligned to match the chunks in the new view. Then element-wise updates are performed on the aligned chunks. Matching lenses are designed to be practically easy to use, so they are equipped with a few fixed matching strategies (such as greedy align) from which the user can choose. However, whether the information is retained or not, still depends on the lens applied after matching. As a result, the more complex the applied lens is, the more difficult to reason about the information retained in the new source.
Moreover, it suffers a disadvantage that the alignment is only between a single source list and a single view list, as already discussed in the last paragraph of \autoref{sec:XMLSync}.
\scName{BiFluX}~\cite{Pacheco2014BiFluX} overcomes the disadvantage by providing the functionality that allows the user to write alignment strategies manually; in this way, when we see several lists at once, we are free to search for elements and match them in all the lists. But this alignment still has the limitation that each source element and each view element can only be matched at most once---after that they are classified as either \emph{matched pair}, \emph{unmatched source element}, or \emph{unmatched view element}. Assuming that an element in the view has been copied several times, there is no way to align all the copies with the same source element. (However, it is possible to reuse an element several times for the handling of unmatched elements.)

By contrast, retentive lenses are designed to abstract out matching strategies (alignment) and are more like taking the result of matching as an additional input. This matching is not a one-layer matching but rather, a global one that produces (possibly all the) links between a source's and a view's unchanged parts. The information contained in the linked parts is preserved independently of any further applied lenses.

 \subsubsection{Alignment for Containers}
 To generalise list alignment, a more general notion of data structures called \emph{containers}~\cite{Abbott2005Containers} is used~\cite{Hofmann2012Edit}.
 In the container framework, a data structure is decomposed into a \emph{shape} and its \emph{content}; the shape encodes a set of positions, and the content is a mapping from those positions to the elements in the data structure.
 The existing approaches to container alignment take advantage of this decomposition and treat shapes and contents separately.
 For example, if the shape of a view container changes, \citeauthor{Hofmann2012Edit}'s approach will update the source shape by a fixed strategy that makes insertions or deletions at the rear positions of the (source) containers.
 By contrast, \citeauthor{Pacheco2012Delta}'s method permits more flexible shape changes, and they call it \emph{shape alignment}~\cite{Pacheco2012Delta}.
 In our setting, both the consistency on data and the consistency on shapes are specified by the same set of consistency declarations.
 In the \ensuremath{\Varid{put}} direction, both the data and shape of a new source is determined by (computed from) the data and shape of a view, so there is no need to have separated data and shape alignments.

 Container-based approaches have the same situation (as list alignment) that the retention of information is dependent on the basic lens applied after alignment. Moreover, the container-based approaches face another serious problem:
 they always translate a change on data in the view to another change on data in the source, without affecting the shape of a container. This is wrong in some cases, especially when the decomposition into shape and data is inadequate.
 For example, let the source be \lstinline{Neg (Lit 100)} and the view \lstinline{Sub (Num 0) (Num 100)}. If we modify the view by changing the integer \lstinline{0}~to~\lstinline{1} (so that the view becomes \lstinline{Sub (Num 1) (Num 100)}), the container-based approach would not produce a correct source \lstinline{Minus ...}, as this data change in the view must not result in a shape change in the source.
 In general, the essence of container-based approaches is the decomposition into shape and data such that they can be processed independently (at least to some extent), but when it comes to scenarios where such decomposition is unnatural (like the example above), container-based approaches hardly help.

\subsection{Provenance and Origin}
\label{sec:provenance}

Our idea of links is inspired by research on provenance~\cite{Cheney2009Provenance} in database communities and origin tracking~\cite{vanDeursen1993Origin} in the rewriting communities.

\citeauthor{Cheney2009Provenance} classify provenance into three kinds, \emph{why}, \emph{how}, and \emph{where}: \emph{why-provenance} is the information about which data in the view is from which rows in the source; \emph{how-provenance} additionally counts the number of times a row is used (in the source); \emph{where-provenance} in addition records the column where a piece of data is from. In our setting, we require that two pieces of data linked by vertical correspondence be equal (under a specific pattern), and hence the vertical correspondence resembles where-provenance.
However, the above-mentioned provenance is not powerful enough as they are mostly restricted to relational data, namely rows of tuples---in functional programming, the algebraic data types are more complex. For this need, \emph{dependency provenance}~\cite{Cheney2011Provenance} is proposed; it tells the user on which parts of a source the computation of a part of a view depends. In this sense, our consistency links are closer to dependency provenance.

The idea of inferring consistency links can be found in the work on origin tracking for term rewriting systems~\cite{vanDeursen1993Origin}, in which the origin relations between rewritten terms can be calculated by analysing the rewrite rules statically.
However, it was developed solely for building trace between intermediate terms rather than using trace information to update a tree further. Based on origin tracking, \citeauthor{deJonge2012Algorithm} implemented an algorithm for code refactoring systems, which `preserves formatting for terms that are not changed in the (AST) transformation, although they may have changes in their subterms'~\cite{deJonge2012Algorithm}.
This description shows that the algorithm also decomposes large terms into smaller ones resembling our regions.
In terms of the formatting aspect, we think that retentiveness can in effect be the same as their theorem if
we include vertical correspondence (representing view updates) in the theory, rather than dealing with it implicitly and externally as in \autoref{sec:editOperations}.

The use of consistency links can also be found in \citeauthor{Wang2011Incremental}'s work, where the authors extend state-based lenses and use links for tracing data in a view to its origin in a source~\cite{Wang2011Incremental}.
When a sub-term in the view is edited locally, they use links to identify a sub-term in the source that `contains' the edited sub-term in the view. When updating the old source, it is sufficient to only perform state-based \ensuremath{\Varid{put}} on the identified sub-term (in the source) so that the update becomes an incremental one. Since lenses generated by our DSL also create consistency links (albeit for a different purpose), they can be naturally incrementalised using the same technique.

\subsection{Operation-based BX}
\label{sec:operationBX}

Our work is closely relevant to the operation-based approaches to BX, in particular, the delta-based BX model~\cite{Diskin2011Asymmetric,Diskin2011Symmetric} and edit lenses~\cite{Hofmann2012Edit}. The (asymmetric) delta-based BX model regards the differences between a view state \ensuremath{\Varid{v}} and \ensuremath{\Varid{v'}} as \emph{deltas}, which are abstractly represented as arrows (from the old view to the new view).
The main law of the framework can be described as `given a source state \ensuremath{\Varid{s}} and a view delta $\ensuremath{\Varid{det}}_\ensuremath{\Varid{v}}$, $\ensuremath{\Varid{det}}_\ensuremath{\Varid{v}}$ should be translated to a source delta $\ensuremath{\Varid{det}}_\ensuremath{\Varid{s}}$ between \ensuremath{\Varid{s}} and $s'$ satisfying \ensuremath{\Varid{get}\;\Varid{s'}\mathrel{=}\Varid{v'}}'.
As the law only guarantees the existence of a source delta $\ensuremath{\Varid{det}}_\ensuremath{\Varid{s}}$ that updates the old source to a correct state, it is yet not sufficient to derive Retentiveness in their model, for there are infinite numbers of translated delta $\ensuremath{\Varid{det}}_\ensuremath{\Varid{s}}$ which can take the old source to a correct state, of which only a few are `retentive'.
To illustrate, \citeauthor{Diskin2011Asymmetric} tend to represent deltas as edit operations such as \emph{create}, \emph{delete}, and \emph{change}; representing deltas in this way will only tell the user what must be changed in the new source, while it requires additional work to reason about what is retained.
However, it is possible to exhibit Retentiveness if we represent deltas in some other proper form. Compared to \citeauthor{Diskin2011Asymmetric}'s work, \citeauthor{Hofmann2012Edit} give concrete definitions and implementations for propagating edit operations (in a symmetric setting).

\section{Conclusion}
\label{sec:conclusion}

In this paper, we showed that well-behavedness is not sufficient for retaining information after an update and it may cause problems in many real-world applications.
To address the issue, we illustrated how to use links to preserve desired data fragments of the original source, and developed a semantic framework of (asymmetric) retentive lenses.
Then we presented a small DSL tailored for describing consistency relations between syntax trees; we showed its syntax, semantics, and proved that the pair of \ensuremath{\Varid{get}} and \ensuremath{\Varid{put}} functions generated from any program in the DSL form a retentive lens.
We provide four edit operations which can update a view together with the links between the view and the original source, and demonstrated the practical use of retentive lenses for code refactoring, resugaring, and XML synchronisation;
we discussed related work about alignment, origin tracking, and operation-based BX.
Some further discussions can be found in the appendix (\autoref{sec:discussions}).


\begin{acks}                            
  We thank Jeremy Gibbons, Meng Wang for useful discussions and comments.
  Yongzhe Zhang helped us to create a nice figure in the introduction of the last submission, although the figure was later revised.
  This work is partially supported by the \grantsponsor{GS501100001691}{Japan Society for the Promotion of Science}{https://doi.org/10.13039/501100001691} (JSPS) Grant-in-Aid for Scientific Research (S)~No.~\grantnum{GS501100001691}{17H06099}.
\end{acks}

\newpage

\renewcommand*{\bibfont}{\small}
\bibliography{retentive}

\newpage
\appendix

\section{Programming Examples in Our DSL}
\label{sec:programmingExamples}

Although the DSL is tailored for describing consistency relations between syntax trees, it is also possible to handle general tree transformations and the following are small but typical programming examples other than syntax tree synchronisation.
\begin{itemize}
  \item Let us consider the binary trees
\begin{centerTab}
\begin{lstlisting}
data BinT a = Tip | Node a (BinT a) (BinT a) (*@.@*)
\end{lstlisting}
\end{centerTab}
We can concisely define the \ensuremath{\Varid{mirror}} consistency relation between a tree and its mirroring as
\begin{centerTab}
\begin{lstlisting}
BinT Int <---> BinT Int
  Tip         ~  Tip
  Node i x y  ~  Node i y x (*@.@*)
\end{lstlisting}
\end{centerTab}

  \item We demonstrate the implicit use of some other consistency relations when defining a new one. Suppose that we have defined the following consistency relation between natural numbers and boolean values:
\begin{centerTab}
\begin{lstlisting}
Nat <---> Bool
  Succ _  ~  True
  Zero    ~  False (*@.@*)
\end{lstlisting}
\end{centerTab}
Then we can easily describe the consistency relation between a binary tree over natural numbers and a binary tree over boolean values:
\begin{centerTab}
\begin{lstlisting}
BinT Nat  <---> BinT Bool
  Tip           ~  Tip
  Node x ls rs  ~  Node x ls rs (*@.@*)
\end{lstlisting}
\end{centerTab}

  \item Let us consider rose trees, a data structure mutually defined with lists:
\begin{centerTab}
\begin{lstlisting}
data RTree a = RNode a (List (RTree a))
data List  a = Nil | Cons a (List a) (*@.@*)
\end{lstlisting}
\end{centerTab}
We can define the following consistency relation to associate the left spine of a tree with a list:
\begin{centerTab}
\begin{lstlisting}
RTree Int <---> List Int
  RNode i Nil         ~  Cons i Nil
  RNode i (Cons x _)  ~  Cons i x (*@.@*)
\end{lstlisting}
\end{centerTab}
\end{itemize}

\section{Discussions of the Paper}
\label{sec:discussions}
We will briefly discuss Strong Retentiveness (that subsumes Hippocraticness), our thought on (retentive) lens composition, the feasibility of retaining code styles for refactoring tools, and our choice of the word `retentive'.

\subsection{Strong Retentiveness}
\label{sec:strongRetentiveness}
Through our research into Retentiveness, we also tried a different theory, which we call \emph{Strong Retentiveness} now, that requires that the consistency links generated by \ensuremath{\Varid{get}} should additionally capture all the `information' of the source and \emph{uniquely identify} it. Strong Retentiveness is appealing in the sense that (we proved that) it subsumes Hippocraticness: the more information we require that the new source have, the more restrictions we impose on the possible forms of the new source; in the extreme case where the input links capture all the information and are only valid for at most one source, the new source has to be the same as the original one. However, using Strong Retentiveness demands extra effort in practice, for a set of region patterns can never uniquely identify a tree; as a result, much more information is required. For instance, \lstinline[mathescape]!cst$_\texttt{1}$ = Minus "" (Lit 1) (Lit 2)! has region patterns \lstinline[mathescape]!reg$_{\texttt{1}}$ = Minus "" _ _!, \lstinline[mathescape]!reg$_{\texttt{2}}$ = Lit 1!, and \lstinline[mathescape]!reg$_{\texttt{3}}$ = Lit 2!, which, however, are also satisfied by \lstinline[mathescape]!cst$_{\texttt{2}}$ = Minus "" (Lit 2) (Lit 1)! in which regions are assembled in a different way.

This observation inspires us to generalise region patterns to \emph{properties} in order for holding more information (that can eventually uniquely identify a tree) and generalise links connecting $\ensuremath{\Conid{Pattern}} \times \ensuremath{\Conid{Path}}$ to links connecting $\ensuremath{\Conid{Property}} \times \ensuremath{\Conid{Path}}$ accordingly. We eventually formalised three kinds of properties that are sufficient to capture all the information of a tree (e.g.~\lstinline[mathescape]!cst$_{\texttt{1}}$!): \emph{region patterns} (e.g.~\lstinline[mathescape]!reg$_{\texttt{1}}$!, \lstinline[mathescape]!reg$_{\texttt{2}}$!, and \lstinline[mathescape]!reg$_{\texttt{3}}$!), \emph{relative positions} between two regions (e.g.~\lstinline[mathescape]!reg$_{\texttt{2}}$! is the first child of \lstinline[mathescape]!reg$_{\texttt{1}}$! and \lstinline[mathescape]!reg$_{\texttt{3}}$! is the second child of \lstinline[mathescape]!reg$_{\texttt{1}}$!), and \emph{top} that marks the top of a tree (e.g.~\lstinline[mathescape]!reg$_{\texttt{1}}$! is the top).
Worse still, observant readers might have found that properties need to be named so that they can be referred to by other properties; for instance, the region pattern \lstinline{Minus "" _ _} is named \lstinline[mathescape]!reg$_{\texttt{1}}$! and is referred to as the top of \lstinline[mathescape]!cst$_{\texttt{1}}$!. This will additionally cause many difficulties in lens composition, as different lenses might assign the same region different names and we need to do `alpha conversion'. Take everything into consideration, finally, we opted for the `weaker' but simpler version of Retentiveness.

\subsection{Rethinking Lens Composition}
\label{sec:rethinkRLensComp}
We defined retentive lens composition (\autoref{def:rlensComp}) in which we treat link composition as relation composition. In this case, however, the composition of two lenses \ensuremath{lens_{AB}} and \ensuremath{lens_{BC}} may not be satisfactory because the link composition might (trivially) produce an empty set as the result, if \ensuremath{lens_{AB}} and \ensuremath{lens_{BC}} decompose a tree \ensuremath{\Varid{b}} (of type \ensuremath{\Conid{B}}) in a different way, as the following example shows:
\begin{center}
\includegraphics[scale=0.7,trim={7.5cm 8.5cm 7.5cm 7.5cm},clip]{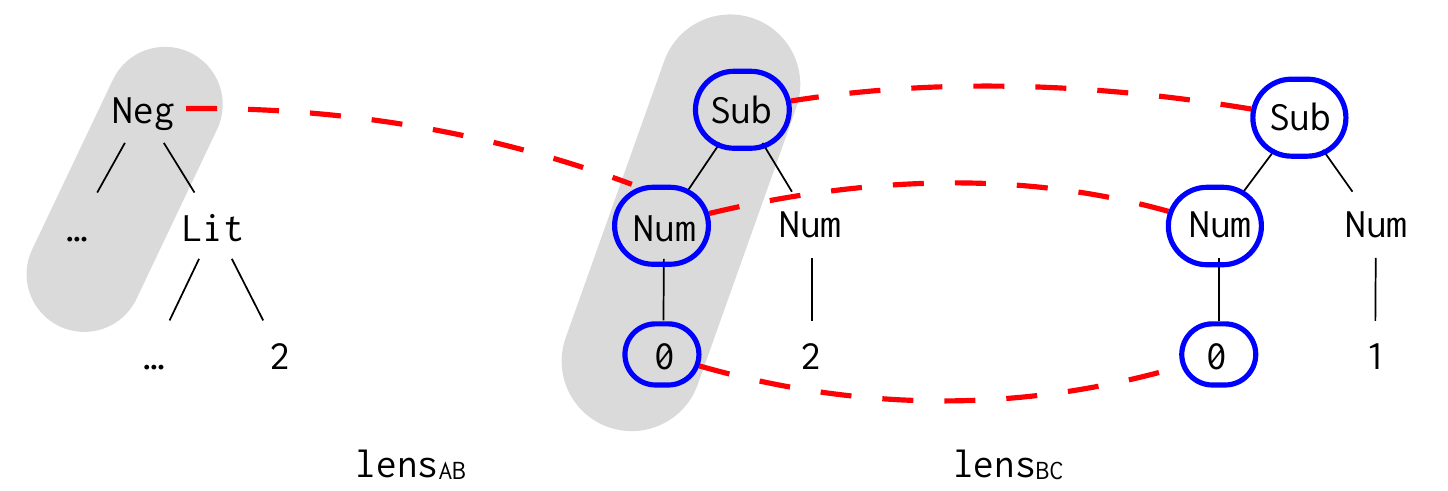}
\end{center}
In the above figure, \lstinline[mathescape]!lens$_{\texttt{AB}}$! connects the region (pattern) \lstinline{Neg a _} with \lstinline{Sub (Num 0) _} (the grey parts); while \lstinline[mathescape]!lens$_{\texttt{BC}}$! decomposes \lstinline{Sub (Num 0) _} into three small regions and establishes links for them respectively. For this case, our current link composition simply produces an empty set as the result.

Coincidentally, similar problems can also be found in quotient lenses~\cite{Foster2008Quotient}: A quotient lens operates on sources and views that are divided into many equivalent classes, and the well-behavedness is defined on those equivalent classes rather than a particular pair of source and view. In order to establish a sequential composition \ensuremath{\Varid{l};\Varid{k}}, the authors require that the abstract (view-side) equivalence relation of lens \ensuremath{\Varid{l}} is identical to the concrete (source-side) equivalence of lens \ensuremath{\Varid{k}}. We leave other possibilities of link composition to future work.

As for our DSL, the lack of composition does not cause problems because of the philosophy of design. Take the scenario of writing a parser for example where there are two main approaches for the user to choose: to use parser combinators (such as \textsc{Parsec}) or to use parser generators (such as \textsc{Happy}). While parser combinators offer the user many small composable components, parser generators usually provide the user with a high-level syntax for describing the grammar of a language using production rules (associated with semantic actions). Then the generated parser is used as a `standalone black box' and usually will not be composed with some others (although it is still possible to be composed `externally'). Our DSL is designed to be a `lens generator' and we have no difficulty in writing bidirectional transformations for the subset of Java~8 in \autoref{sec:refactoring}.

\subsection{Retaining Code Styles}
A challenge to refactoring tools is to retain the style of program text such as indentation, vertical alignment of identifiers, and the place of line breaks. For example, an argument of a function application may be vertically aligned with a previous argument; when a refactoring tool moves the application to a different place, what should be retained is not the absolute number of spaces preceding the arguments but the \emph{property} that these two arguments are vertically aligned.

Although not implemented in the DSL, these properties can be added to the set of \ensuremath{\Conid{Property}} as introduced in \autoref{sec:strongRetentiveness}.
For instance, we may have $\ensuremath{\mathsf{VertAligned}\;\Varid{x}\;\Varid{y}} \in \ensuremath{\Conid{Property}}$ for $\ensuremath{\Varid{x}}, \ensuremath{\Varid{y}} \in \ensuremath{\Conid{Name}}$ (i.e.~\ensuremath{\Varid{x}} and \ensuremath{\Varid{y}} are names of some regions); a CST satisfies such a property if region \ensuremath{\Varid{y}} is vertically aligned with region \ensuremath{\Varid{x}}.When \ensuremath{\Varid{get}} computes an AST from such a vertically aligned argument and produces consistency links, the links will not include (real) spaces preceding the argument as a part of the source region; instead, the links connect the property $\ensuremath{\mathsf{VertAligned}\;\Varid{x}\;\Varid{y}}$ (and the corresponding AST region).In the \ensuremath{\Varid{put}} direction, such links serve as directives to adjust the number of spaces preceding the argument to conform to the styling rule. In general, handling code styles can be very language-specific and is beyond the scope of this thesis but could be considered a direction of future work.

\section{Proofs About Retentive Lenses}
\label{sec:rLensProof}

\subsection{Composability}
\label{prf:retPrsv}

In this section, we show the proof of \autoref{thm:retPrsv} with the help of \autoref{fig:rlensComp} and the definition of retentive lens composition (\autoref{def:rlensComp}).

\begin{proof}[Hippocraticness Preservation]
We prove that the composite lens satisfies Hippocraticness with the help of \autoref{fig:rlensComp} and the definition of retentive lens composition (\autoref{def:rlensComp}).\\
Let \ensuremath{get_{AC}\;\Varid{a}\mathrel{=}(\Varid{c},ls_{ac})}. We prove \ensuremath{put_{AC}\;(\Varid{a},\Varid{c'},ls_{ac^{'}})\mathrel{=}\Varid{a'}\mathrel{=}\Varid{a}}. In this case, \ensuremath{\Varid{c'}\mathrel{=}\Varid{c}} and \ensuremath{ls_{ac^{'}}\mathrel{=}ls_{ac}}.
\begin{align*}
&  \ensuremath{put_{AC}\;(\Varid{a},\Varid{c'},ls_{ac^{'}})}\\
= & \mcond{ \ensuremath{put_{AC}\;(\Varid{a},\Varid{c'},ls_{ac^{'}})\mathrel{=}\Varid{a'}\mathrel{=}put_{AB}\;(\Varid{a},\Varid{b'},ls_{ab^{'}})} } \\[-0.8ex]
& \ensuremath{put_{AB}\;(\Varid{a},\Varid{b'},ls_{ab^{'}})} \\
= & \mcond{ \ensuremath{ls_{ab^{'}}} = \ensuremath{ls_{ac^{'}}} \cdot \ensuremath{ls_{b^{'}c^{'}}}^{\circ}\ } \\[-0.8ex]
& \ensuremath{put_{AB}}~(\ensuremath{\Varid{a}}, \ensuremath{\Varid{b'}}, \ensuremath{ls_{ac^{'}}} \cdot \ensuremath{ls_{b^{'}c^{'}}}^{\circ}) \\
= & \mcond{ \text{Since } \ensuremath{\Varid{c'}} = \ensuremath{\Varid{c}}, \text{ we have } \ensuremath{ls_{ac^{'}}} = \ensuremath{ls_{ac}} \text{ and } \ensuremath{ls_{b^{'}c^{'}}} = \ensuremath{ls_{b^{'}c}} } \\[-0.8ex]
& \ensuremath{put_{AB}}~(\ensuremath{\Varid{a}}, \ensuremath{\Varid{b'}}, (\ensuremath{ls_{ac}} \cdot \ensuremath{ls_{b^{'}c}}^{\circ})^{\circ}) \\
= & \mcond{  \ensuremath{\Varid{b'}} = \ensuremath{put_{BC}\;(\Varid{b},\Varid{c'},ls_{bc^{'}})} \text{ and } \ensuremath{\Varid{c'}} = \ensuremath{\Varid{c}} } \\[-0.8ex]
& \ensuremath{put_{AB}}~(\ensuremath{\Varid{a}}, \ensuremath{put_{BC}\;(\Varid{b},\Varid{c},ls_{bc})}, \ensuremath{ls_{ac}} \cdot \ensuremath{ls_{b^{'}c}}^{\circ}) \\
= & \mcond{ \text{By Hippocraticness of } \ensuremath{lens_{BC}}, \ensuremath{\Varid{b'}} = \ensuremath{put_{BC}\;(\Varid{b},\Varid{c},ls_{bc})} = \ensuremath{\Varid{b}} } \\[-0.8ex]
& \ensuremath{put_{AB}}~(\ensuremath{\Varid{a}}, \ensuremath{\Varid{b}}, \ensuremath{ls_{ac}} \cdot \ensuremath{ls_{bc}}^{\circ}) \\
= & \mcond{ \text{The link composition is \numcircledmod{6} in \autoref{fig:rlensComp}, and } \ensuremath{\Varid{b'}} = \ensuremath{\Varid{b}} } \\[-0.8ex]
& \ensuremath{put_{AB}}~(\ensuremath{\Varid{a}}, \ensuremath{\Varid{b}}, \ensuremath{ls_{ab}}) \\
= & \mcond{ \text{Hippocraticness of } \ensuremath{lens_{AB}}} \\[-0.8ex]
& \ensuremath{\Varid{a}} \ \text{.}
\end{align*}
\end{proof}

\begin{proof}[Correctness Preservation]
We prove that the composite lens satisfies Correctness with the help of \autoref{fig:rlensComp} and the definition of retentive lens composition (\autoref{def:rlensComp}).\\
Let \ensuremath{\Varid{a'}\mathrel{=}put_{AC}\;(\Varid{a},\Varid{c'},ls_{ac^{'}})}, we prove \ensuremath{\Varid{fst}\;(get_{AC}\;\Varid{a'})\mathrel{=}\Varid{c'}}.
\begin{align*}
& \ensuremath{\Varid{fst}\;(get_{AC}\;\Varid{a'})} \\
= & \mcond{ \text{Definition of \ensuremath{get_{AC}}}} \\[-0.8ex]
& \ensuremath{\Varid{fst}\;(get_{BC}\;(\Varid{fst}\;(get_{AB}\;\Varid{a'})))} \\
= & \mcond{ \text{\ensuremath{\Varid{a'}\mathrel{=}put_{AC}\;(\Varid{a},\Varid{c'},ls_{ac^{'}})}}} \\[-0.8ex]
& \ensuremath{\Varid{fst}\;(get_{BC}\;(\Varid{fst}\;(get_{AB}\;(put_{AC}\;(\Varid{a},\Varid{c'},ls_{ac^{'}})))))} \\
= & \mcond{ \ensuremath{put_{AC}\;(\Varid{a},\Varid{c'},ls_{ac^{'}})\mathrel{=}\Varid{a'}\mathrel{=}put_{AB}\;(\Varid{a},\Varid{b'},ls_{ab^{'}})} } \\[-0.8ex]
& \ensuremath{\Varid{fst}\;(get_{BC}\;(\Varid{fst}\;(get_{AB}\;(put_{AB}\;(\Varid{a},\Varid{b'},ls_{ab^{'}})))))} \\
= & \mcond{ \text{Correctness of } \ensuremath{lens_{AB}}} \\[-0.8ex]
& \ensuremath{\Varid{fst}\;(get_{BC}\;(\Varid{b'}))} \\
= & \mcond{ \ensuremath{\Varid{b'}\mathrel{=}put_{BC}\;(\Varid{b},\Varid{c'},ls_{bc^{'}})} } \\[-0.8ex]
& \ensuremath{\Varid{fst}\;(get_{BC}\;(put_{BC}\;(\Varid{b},\Varid{c'},ls_{bc^{'}})))} \\
= & \mcond{ \text{Correctness of } \ensuremath{lens_{BC}}} \\[-0.8ex]
& \ensuremath{\Varid{c'}} \ \text{.}
\end{align*}
\end{proof}

\begin{proof}[Retentiveness Preservation]
In \autoref{fig:rlensComp}, we prove $\ensuremath{\Varid{fst}} \cdot \ensuremath{ls_{ac}} \subseteq \ensuremath{\Varid{fst}} \cdot \ensuremath{ls_{a^{'}c^{'}}}$.

To finish the proof, we need the following lemma.
\begin{lemma}
\label{prop:fstPrsvDom}
Given a relation \ensuremath{\Conid{R}} and a function \ensuremath{\Varid{f}}, we have
\begin{align*}
\rdom (\ensuremath{\Varid{f}} \cdot R) &= \rdom R \quad\text{if } \ldom R \subseteq \rdom f \, \text{, and}\\
\ldom (R \cdot \ensuremath{\Varid{f}}) &= \ldom R \quad\text{if } \rdom R \subseteq \ldom f \ \text{.}
\end{align*}
\end{lemma}
\begin{proof}
We prove the first equation; the second equation is symmetric.\\
Suppose $f : X \to Y$ and $R : Y \sim Z$.
By definition, $\rdom R = \{\, z \in Z \mid \exists y \in Y, \ y \mathrel{R} z \,\}$ and $\rdom (f \cdot R) = \{\, z \in Z \mid \exists y \in Y,\ \exists \ x \in X, \ x \mathrel{f} y \mathrel{R} z \,\}$.
Since $\ldom R \subseteq \rdom f$, we know that $\forall y. \ y \in \ldom R \Rightarrow y \in \rdom f$;
on the other hand, we also have $y \in \rdom f \Rightarrow \exists x . \ x \in X$.
Therefore, $\forall y. \ y \in \ldom R \Rightarrow \exists x . \ x \in X$ and thus $\rdom (f \cdot R) = \{\, z \in Z \mid \exists y \in Y,\ \exists \ x \in X, \ x \mathrel{f} y \mathrel{R} z \,\} = \{\, z \in Z \mid \exists y \in Y,\ y \mathrel{R} z \,\} = \rdom R$.
\end{proof}

Now, we present the main proof:
\begin{align*}
& \ensuremath{\Varid{fst}} \cdot \ensuremath{ls_{ac}} \\
= & \mcond{ R = R \cdot id_{\rdom R}} \\[-0.8ex]
& \ensuremath{\Varid{fst}} \cdot \ensuremath{ls_{ac^{'}}} \cdot \ensuremath{\Varid{id}}_{\rdom (\ensuremath{ls_{ac^{'}}})} \\
\subseteq & \mcond{ \ensuremath{\Varid{id}}_{\rdom (\ensuremath{ls_{ac^{'}}})} \subseteq \ensuremath{ls_{c^{'}b^{'}}} \cdot \ensuremath{ls_{c^{'}b^{'}}}^{\circ} \text{ by sub-proof-1 below} } \\[-0.8ex]
& \ensuremath{\Varid{fst}} \cdot \ensuremath{ls_{ac^{'}}} \cdot (\ensuremath{ls_{c^{'}b^{'}}} \cdot \ensuremath{ls_{c^{'}b^{'}}}^{\circ}) \\
= & \mcond{ \text{Relation composition is associative} } \\[-0.8ex]
& \ensuremath{\Varid{fst}} \cdot (\ensuremath{ls_{ac^{'}}} \cdot \ensuremath{ls_{c^{'}b^{'}}}) \cdot \ensuremath{ls_{c^{'}b^{'}}}^{\circ} \\
= & \mcond{ \ensuremath{ls_{ab^{'}}} = \ensuremath{ls_{ac^{'}}} \cdot \ensuremath{ls_{c^{'}b^{'}}} \text{ (\numcircledmod{6} in \autoref{fig:rlensComp})}  } \\[-0.8ex]
& \ensuremath{\Varid{fst}} \cdot \ensuremath{ls_{ab^{'}}} \cdot \ensuremath{ls_{c^{'}b^{'}}}^{\circ} \\
\subseteq & \mcond{ \text{Retentiveness of } \ensuremath{lens_{AB}} \text{ and } \ensuremath{ls_{c^{'}b^{'}}}^{\circ} = \ensuremath{ls_{b^{'}c^{'}}} } \\[-0.8ex]
& \ensuremath{\Varid{fst}} \cdot \ensuremath{ls_{a^{'}b^{'}}} \cdot \ensuremath{ls_{b^{'}c^{'}}} \\
= & \mcond{ \ensuremath{ls_{a^{'}c^{'}}} = \ensuremath{ls_{a^{'}b^{'}}} \cdot \ensuremath{ls_{b^{'}c^{'}}} } \\[-0.8ex]
& \ensuremath{\Varid{fst}} \cdot \ensuremath{ls_{a^{'}c^{'}}} \ \text{.}
\end{align*}

sub-proof-1: $\ensuremath{\Varid{id}}_{\rdom (\ensuremath{ls_{ac^{'}}})} \subseteq \ensuremath{ls_{c^{'}b^{'}}} \cdot \ensuremath{ls_{c^{'}b^{'}}}^{\circ} \Leftrightarrow \rdom (\ensuremath{ls_{ac^{'}}}) \subseteq \ldom (\ensuremath{ls_{c^{'}b^{'}}})$ and we prove the latter using linear proofs.
The right column of each line gives the reason how it is derived.

\begin{align*}
1.\quad & \ldom (\ensuremath{ls_{c^{'}b^{'}}}) = \rdom (\ensuremath{ls_{b^{'}c^{'}}})   \quad&\quad   \text{definition of relations} \\
2.\quad & \ensuremath{\Varid{fst}} \cdot \ensuremath{ls_{bc^{'}}} \subseteq \ensuremath{\Varid{fst}} \cdot \ensuremath{ls_{b^{'}c^{'}}}   \quad&\quad   \text{Retentiveness of \ensuremath{lens_{BC}}}\\
3.\quad & \rdom (\ensuremath{\Varid{fst}} \cdot \ensuremath{ls_{bc^{'}}}) \subseteq \rdom (\ensuremath{\Varid{fst}} \cdot \ensuremath{ls_{b^{'}c^{'}}})   \quad&\quad   \text{2 and definition of relation inclusion} \\
4.\quad & \rdom (\ensuremath{ls_{bc^{'}}}) \subseteq \rdom (\ensuremath{ls_{b^{'}c^{'}}})   \quad&\quad   \text{3 and \autoref{prop:fstPrsvDom}} \\
5.\quad & \ensuremath{ls_{bc^{'}}} = (\ensuremath{ls_{ac^{'}}}^{\circ} \cdot \ensuremath{ls_{ab}})^{\circ} = (\ensuremath{ls_{c^{'}a}} \cdot \ensuremath{ls_{ab}})^{\circ}   \quad&\quad   \text{\numcircledmod{3} in \autoref{fig:rlensComp}} \\
6.\quad & \rdom (\ensuremath{ls_{bc^{'}}}) = \rdom (\ensuremath{ls_{c^{'}a}} \cdot \ensuremath{ls_{ab}})^{\circ}   \quad&\quad   \text{5}\\
7.\quad & \rdom (\ensuremath{ls_{c^{'}a}} \cdot \ensuremath{ls_{ab}})^{\circ} = \ldom (\ensuremath{ls_{c^{'}a}} \cdot \ensuremath{ls_{ab}})   \quad&\quad    \text{definition of converse relation} \\
8. \quad & \ldom (\ensuremath{ls_{c^{'}a}} \cdot \ensuremath{ls_{ab}}) \subseteq \ldom (\ensuremath{ls_{c^{'}a}})   \quad&\quad    \text{definition of relation composition} \\
9. \quad & \ldom (\ensuremath{ls_{c^{'}a}}) = \rdom (\ensuremath{ls_{ac^{'}}})   \quad&\quad    \text{definition of converse relation} \\
10. \quad & \rdom (\ensuremath{ls_{bc^{'}}}) = \rdom (\ensuremath{ls_{ac^{'}}})   \quad&\quad    \text{6, 7, 8, and 9} \\
11. \quad & \rdom (\ensuremath{ls_{ac^{'}}}) \subseteq \ldom (\ensuremath{ls_{c^{'}b^{'}}})   \quad&\quad    \text{10, 4, and 1}
\end{align*}

\end{proof}

\subsection{Retentiveness of the DSL}
\label{app:proof}

In this section, we prove that the \ensuremath{\Varid{get}} and \ensuremath{\Varid{put}} semantics given in \autoref{sec:DSLSem} does satisfy the three properties (\autoref{def:retLens}) of a retentive lens.
Most of the proofs are proved by induction on the size of the trees.

\begin{lemma}
  The \ensuremath{\Varid{get}} function described in \autoref{sec:getsem} is total.
\end{lemma}
\begin{proof}
  Because we require source pattern coverage, \ensuremath{\Varid{get}} is defined for all the input data.
  Besides, since our DSL syntactically restricts source pattern \ensuremath{\Varid{spat}_{\Varid{k}}} to not being a bare variable pattern, for any $\ensuremath{\Varid{v}} \in \ensuremath{\Conid{Vars}}(\ensuremath{\Varid{spat}_{\Varid{k}}})$, $\ensuremath{\Varid{decompose}}(\ensuremath{\Varid{spat}_{\Varid{k}}}, \ensuremath{\Varid{s}})$ is a proper subtree of \ensuremath{\Varid{s}}.
  So the recursion always decreases the size of the \ensuremath{\Varid{s}} parameter and thus terminates.
\end{proof}

\begin{lemma}
  For a pair of \ensuremath{\Varid{get}} and \ensuremath{\Varid{put}} described in \autoref{sec:putsem} and any \ensuremath{\Varid{s}\mathbin{:}\Conid{S}}, $\ensuremath{\Varid{check}}(\ensuremath{\Varid{s}}, \ensuremath{\Varid{get}}(\ensuremath{\Varid{s}})) = \ensuremath{\Conid{True}}$.
\end{lemma}
\begin{proof}
  We prove the lemma by induction on the structure of \ensuremath{\Varid{s}}.
  By the definition of \ensuremath{\Varid{get}} and \ensuremath{\Varid{check}},
  \begin{align*}
      &\ensuremath{\Varid{check}}(\ensuremath{\Varid{s}}, \ensuremath{\Varid{get}}(\ensuremath{\Varid{s}})) \\
    = & \mcond{ \ensuremath{\Varid{get}}(\ensuremath{\Varid{s}}) \text{ produces consistency links}} \\[-0.8ex]
    & \ensuremath{\Varid{chkWithLink}}(\ensuremath{\Varid{s}}, \ensuremath{\Varid{get}}(\ensuremath{\Varid{s}})) \\
    = & \mcond{ \text{Unfolding } \ensuremath{\Varid{get}}(s) }\\[-0.8ex]
     &\ensuremath{\Varid{chkWithLink}}(\ensuremath{\Varid{s}}, \ensuremath{\Varid{reconstruct}}(\ensuremath{\Varid{vpat}_{\Varid{k}}}, \ensuremath{\Varid{fst}} \circ \ensuremath{\Varid{vls}}), l_\ensuremath{\Varid{root}} \cup \ensuremath{\Varid{links}})
  \end{align*}
  where \ensuremath{\Varid{vpat}_{\Varid{k}}}, \ensuremath{\Varid{fst}}, \ensuremath{\Varid{vls}}, $l_\ensuremath{\Varid{root}}$ and \ensuremath{\Varid{links}} are those in the definition of \ensuremath{\Varid{get}} (\ref{equ:get}).
  In \hyperref[def:chkWithLink]{\ensuremath{\Varid{chkWithLink}}}, \ensuremath{\Varid{cond}_{\mathrm{1}}} and \ensuremath{\Varid{cond}_{\mathrm{2}}} are true by the evident semantics of pattern matching functions such as \ensuremath{\Varid{isMatch}} and \ensuremath{\Varid{reconstruct}}.
  \ensuremath{\Varid{cond}_{\mathrm{3}}} is true following the definition of $l_\ensuremath{\Varid{root}}$, \ensuremath{\Varid{links}}, and \ensuremath{\Varid{divide}}.
  Finally, \ensuremath{\Varid{cond}_{\mathrm{4}}} is true by the inductive hypothesis.
\end{proof}

\begin{lemma}{(Focusing)}\label{lem:focus}
  If $\ensuremath{\Varid{sel}}(\ensuremath{\Varid{s}}, p) = s'$ and for any $((\_, \ensuremath{\Varid{spath}}), (\_, \_)) \in \ensuremath{\Varid{ls}}$, \ensuremath{\Varid{p}} is a prefix of \ensuremath{\Varid{spath}}, then
  \[\ensuremath{\Varid{put}}(\ensuremath{\Varid{s}}, \ensuremath{\Varid{v}}, \ensuremath{\Varid{ls}}) = \ensuremath{\Varid{put}}(\ensuremath{\Varid{s'}}, \ensuremath{\Varid{v}}, \ensuremath{\Varid{ls'}}) \,\text{ and }\, \ensuremath{\Varid{check}}(\ensuremath{\Varid{s}}, \ensuremath{\Varid{v}}, \ensuremath{\Varid{ls}}) = \ensuremath{\Varid{check}}(\ensuremath{\Varid{s'}}, \ensuremath{\Varid{v}}, \ensuremath{\Varid{ls'}})\]
  where $\ensuremath{\Varid{ls'}} = \myset{((a, b), (c, d)) \mid ((a, p \ensuremath{\plus } b), (c, d))}$.
\end{lemma}
\begin{proof}
  From the definitions of \ensuremath{\Varid{put}} and \ensuremath{\Varid{check}}, we find that their first argument (of type \ensuremath{\Conid{S}}) is invariant during the recursive process.
  In fact, the first argument is only used when checking whether a link in \ensuremath{\Varid{ls}} is valid with respect to the source tree.
  Since all links in \ensuremath{\Varid{ls}} connect to the subtree \ensuremath{\Varid{s'}}, the parts in \ensuremath{\Varid{s}} above \ensuremath{\Varid{s'}} can be trimmed and the identity holds.
\end{proof}

\begin{theorem}{(Hippocraticness of the DSL)}
  For any \ensuremath{\Varid{s}} of type \ensuremath{\Conid{S}},
    \[ \ensuremath{\Varid{put}}(\ensuremath{\Varid{s}}, \ensuremath{\Varid{get}}(\ensuremath{\Varid{s}}))\footnotemark = s \ \text{.} \]
\end{theorem}
\footnotetext{For simplicity, we regard \ensuremath{(\Varid{a},(\Varid{b},\Varid{c}))} the same as \ensuremath{(\Varid{a},\Varid{b},\Varid{c})}.}
\begin{proof}[Proof of Hippocraticness]
  Also by induction on the structure of \ensuremath{\Varid{s}},
  \begin{align*}
    & \ensuremath{\Varid{put}}(\ensuremath{\Varid{s}}, \ensuremath{\Varid{get}}(\ensuremath{\Varid{s}})) \\
    =& \mcond{\text{Unfolding } \ensuremath{\Varid{get}}(s)} \\[-0.8ex]
    & \ensuremath{\Varid{put}}(\ensuremath{\Varid{s}}, \ensuremath{\Varid{reconstruct}}(\ensuremath{\Varid{vpat}_{\Varid{k}}}, \ensuremath{\Varid{fst}} \circ \ensuremath{\Varid{vls}}), l_\ensuremath{\Varid{root}} \cup \ensuremath{\Varid{links}}) \ \text{,}
  \end{align*}
  where $\ensuremath{\Varid{spat}_{\Varid{k}}} \sim \ensuremath{\Varid{vpat}_{\Varid{k}}} \in \ensuremath{\Conid{R}}$ is the unique rule such that \ensuremath{\Varid{spat}_{\Varid{k}}} matches \ensuremath{\Varid{s}}.
  $l_\ensuremath{\Varid{root}}$, \ensuremath{\Varid{links}}, and \ensuremath{\Varid{vls}} are defined exactly the same as in \ensuremath{\Varid{get}} (\ref{equ:get}).

  Now we expand \ensuremath{\Varid{put}}. Because $l_\ensuremath{\Varid{root}}$ links to the root of the view, \ensuremath{\Varid{put}} falls to its second case.
  \begin{equation} \label{equ:hippo}
       \ensuremath{\Varid{put}}(\ensuremath{\Varid{s}}, \ensuremath{\Varid{get}}(\ensuremath{\Varid{s}})) = \ensuremath{\Varid{inj}}(\ensuremath{\Varid{reconstruct}}(\ensuremath{\Varid{spat}_{\Varid{k}}}', \ensuremath{\Varid{ss}}))
  \end{equation}
  where
  \begin{align*}
  & \ensuremath{\Varid{spat}_{\Varid{k}}}' \\
  &= \mcond{ \ensuremath{\Varid{spat}} \text{ in (\ref{equ:put2}) is } \ensuremath{\Varid{eraseVars}}(\ensuremath{\Varid{fillWildcards}}(\ensuremath{\Varid{spat}_{\Varid{k}}}, \ensuremath{\Varid{s}})) } \\[-0.8ex]
  & \ensuremath{\Varid{fillWildcards}}(\ensuremath{\Varid{spat}_{\Varid{k}}}, \ensuremath{\Varid{eraseVars}}(\ensuremath{\Varid{fillWildcards}}(\ensuremath{\Varid{spat}_{\Varid{k}}}, \ensuremath{\Varid{s}}))) \\
  &= \mcond{ \text{See \autoref{fig:fillWildcards}}  } \\[-0.8ex]
  & \ensuremath{\Varid{fillWildcards}}(\ensuremath{\Varid{spat}_{\Varid{k}}}, \ensuremath{\Varid{s}}) \ \text{.}
  \end{align*}
  and
  \begin{align*}
    \ensuremath{\Varid{ss}} = \lambda (t \in \ensuremath{\Conid{Vars}}(\ensuremath{\Varid{spat}_{\Varid{k}}})) \rightarrow \ensuremath{\Varid{put}}(\ensuremath{\Varid{s}}, \ensuremath{\Varid{vs}}(t), \ensuremath{\Varid{divide}}(\ensuremath{\Conid{Path}}(\ensuremath{\Varid{vpat}_{\Varid{k}}}, t), \ensuremath{\Varid{links}}))
  \end{align*}
  where $\ensuremath{\Varid{vs}} = \ensuremath{\Varid{decompose}}(\ensuremath{\Varid{vpat}_{\Varid{k}}}, \ensuremath{\Varid{reconstruct}}(\ensuremath{\Varid{vpat}_{\Varid{k}}}, \ensuremath{\Varid{fst}} \circ \ensuremath{\Varid{vls}})) = \ensuremath{\Varid{fst}} \circ \ensuremath{\Varid{vls}}$. (See the beginning of the proof.)
  Since $\ensuremath{\Varid{vls}} = \ensuremath{\Varid{get}} \circ \ensuremath{\Varid{decompose}}(\ensuremath{\Varid{spat}_{\Varid{k}}}, s)$, we have
  \begin{align*}
    \ensuremath{\Varid{ss}} &= \lambda (t \in \ensuremath{\Conid{Vars}}(\ensuremath{\Varid{spat}_{\Varid{k}}})) \rightarrow \\
    &\myindent \ensuremath{\Varid{put}}(\ensuremath{\Varid{s}}, \ensuremath{\Varid{fst}}(\ensuremath{\Varid{get}}(\ensuremath{\Varid{decompose}}(\ensuremath{\Varid{spat}_{\Varid{k}}}, s)(t))), \ensuremath{\Varid{divide}}(\ensuremath{\Conid{Path}}(\ensuremath{\Varid{vpat}_{\Varid{k}}}, t), \ensuremath{\Varid{links}}))
  \end{align*}
  By \autoref{lem:focus}, we have
  \begin{align*}
    \ensuremath{\Varid{ss}} &= \lambda (t \in \ensuremath{\Conid{Vars}}(\ensuremath{\Varid{spat}_{\Varid{k}}})) \rightarrow \\
    &\myindent \ensuremath{\Varid{put}}(\ensuremath{\Varid{decompose}}(\ensuremath{\Varid{spat}_{\Varid{k}}}, s)(t), \ensuremath{\Varid{fst}}(\ensuremath{\Varid{get}}(\ensuremath{\Varid{decompose}}(\ensuremath{\Varid{spat}_{\Varid{k}}}, s)(t))), \\
    &\myindent \myindentS \ensuremath{\Varid{snd}}(\ensuremath{\Varid{get}}(\ensuremath{\Varid{decompose}}(\ensuremath{\Varid{spat}_{\Varid{k}}}, s)(t)))) \\
    &=\mcond{\text{Inductive hypothesis for }\ensuremath{\Varid{decompose}}(\ensuremath{\Varid{spat}_{\Varid{k}}}, s)(t)} \\
    &\quad \lambda (t \in \ensuremath{\Conid{Vars}}(\ensuremath{\Varid{spat}_{\Varid{k}}})) \rightarrow \ensuremath{\Varid{decompose}}(\ensuremath{\Varid{spat}_{\Varid{k}}}, s)(t) \\
    &= \ensuremath{\Varid{decompose}}(\ensuremath{\Varid{spat}_{\Varid{k}}}, s) \ \text{.}
  \end{align*}
  Now, we substitute $\ensuremath{\Varid{fillWildcards}}(\ensuremath{\Varid{spat}_{\Varid{k}}}, \ensuremath{\Varid{s}})$ for $\ensuremath{\Varid{spat}_{\Varid{k}}}'$ and $\ensuremath{\Varid{decompose}}(\ensuremath{\Varid{spat}_{\Varid{k}}}, s)$ for \ensuremath{\Varid{ss}} in equation (\ref{equ:hippo}), and obtain
  \begin{align*}
       & \ensuremath{\Varid{put}}(\ensuremath{\Varid{s}}, \ensuremath{\Varid{get}}(\ensuremath{\Varid{s}})) \\
      =& \mcond{ Equation (\ref{equ:hippo}) } \\[-0.8ex]
      & \ensuremath{\Varid{inj}}_{\ensuremath{\Conid{S}}\rightarrow\ensuremath{\Conid{S}}@V}(\ensuremath{\Varid{reconstruct}}(\ensuremath{\Varid{spat}_{\Varid{k}}}', \ensuremath{\Varid{ss}}))\\
      %
      = & \ensuremath{\Varid{inj}}_{\ensuremath{\Conid{S}}\rightarrow\ensuremath{\Conid{S}}@V}(\ensuremath{\Varid{reconstruct}}(\ensuremath{\Varid{fillWildcards}}(\ensuremath{\Varid{spat}_{\Varid{k}}}, \ensuremath{\Varid{s}}), \ensuremath{\Varid{decompose}}(\ensuremath{\Varid{spat}_{\Varid{k}}}, \ensuremath{\Varid{s}})))\\
      =& \mcond{ \text{See \autoref{fig:reconstructDecompose}} } \\[0.8ex]
      & \ensuremath{\Varid{inj}}_{\ensuremath{\Conid{S}} \rightarrow \ensuremath{\Conid{S}}@V}(s) \\
      =& \ensuremath{\Varid{s}} \ \text{.}
  \end{align*}
  This completes the proof of Hippocraticness.
\end{proof}

\begin{figure}[t]
\centering
\includegraphics[scale=0.7,trim={5cm 8.2cm 3.5cm 9cm},clip]{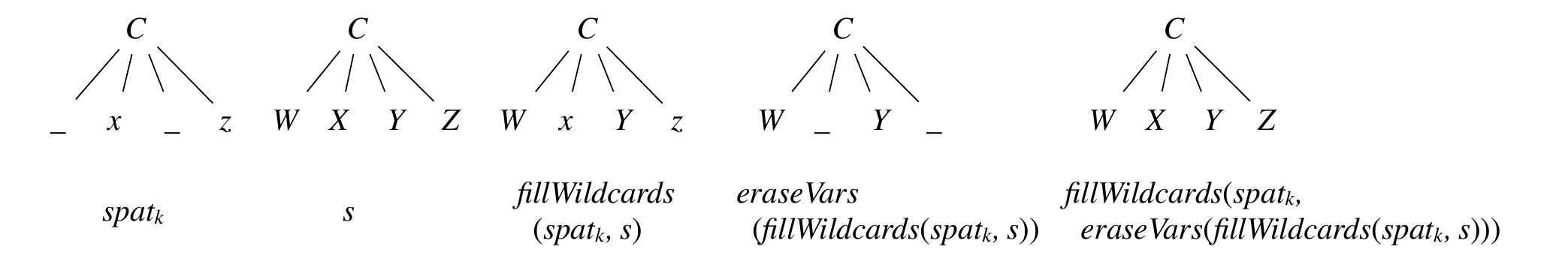}
\caption{A property regarding \ensuremath{\Varid{fillWildcards}}.}
\label{fig:fillWildcards}
\end{figure}

\begin{figure}[t]
\centering
\includegraphics[scale=0.7,trim={5cm 8.2cm 3cm 9cm},clip]{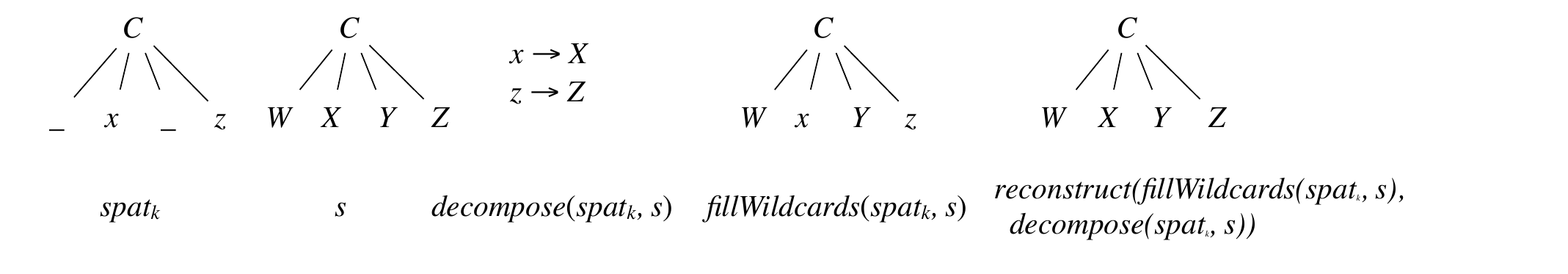}
\caption{A property regarding Reconstruct-Decompose.}
\label{fig:reconstructDecompose}
\end{figure}

\begin{theorem}{(Correctness of the DSL)}
  For any \ensuremath{(\Varid{s},\Varid{v},\Varid{ls})} that makes $\ensuremath{\Varid{check}}(s, v, \ensuremath{\Varid{ls}}) = \ensuremath{\Conid{True}}$, $\ensuremath{\Varid{get}}(\ensuremath{\Varid{put}}(\ensuremath{\Varid{s}}, \ensuremath{\Varid{v}}, \ensuremath{\Varid{ls}})) = (\ensuremath{\Varid{v}}, \ensuremath{\Varid{ls'}})$, for some \ensuremath{\Varid{ls'}}.
\end{theorem}
\begin{proof}[Proof of Correctness]
  We prove Correctness by induction on the size of $(\ensuremath{\Varid{v}}, \ensuremath{\Varid{ls}})$.
  The proofs of the two cases of \ensuremath{\Varid{put}} are quite similar, and therefore we only present the first one, in which $\ensuremath{\Varid{put}}(\ensuremath{\Varid{s}}, \ensuremath{\Varid{v}}, \ensuremath{\Varid{ls}})$ falls into the first case of \ensuremath{\Varid{put}}: i.e.~ $\ensuremath{\Varid{put}}(\ensuremath{\Varid{s}}, \ensuremath{\Varid{v}}, \ensuremath{\Varid{ls}}) = \ensuremath{\Varid{reconstruct}}(\ensuremath{\Varid{fillWildcardsWithDefaults}}(\ensuremath{\Varid{spat}_{\Varid{k}}}), \ensuremath{\Varid{ss}})$. Then
  \begin{align*}
    \ensuremath{\Varid{get}}(\ensuremath{\Varid{put}}(\ensuremath{\Varid{s}}, \ensuremath{\Varid{v}}, \ensuremath{\Varid{ls}})) = \ensuremath{\Varid{get}}(\ensuremath{\Varid{reconstruct}}(\ensuremath{\Varid{fillWildcardsWithDefaults}}(\ensuremath{\Varid{spat}_{\Varid{k}}}), \ensuremath{\Varid{ss}}))
  \end{align*}
  where $\ensuremath{\Varid{spat}_{\Varid{k}}} \sim \ensuremath{\Varid{vpat}_{\Varid{k}}} \in \ensuremath{\Conid{R}}$, $\ensuremath{\Varid{isMatch}}(\ensuremath{\Varid{vpat}_{\Varid{k}}}, \ensuremath{\Varid{v}}) = \ensuremath{\Conid{True}}$, and
  \begin{align*}
    & ss = \lambda (t \in \ensuremath{\Conid{Vars}}(\ensuremath{\Varid{spat}_{\Varid{k}}})) \rightarrow \\
    & \myindent \ensuremath{\Varid{put}}(\ensuremath{\Varid{s}}, \ensuremath{\Varid{decompose}}(\ensuremath{\Varid{vpat}}_k, \ensuremath{\Varid{v}})(t), \ensuremath{\Varid{divide}}(\ensuremath{\Conid{Path}}(\ensuremath{\Varid{vpat}_{\Varid{k}}}, t), \ensuremath{\Varid{ls}})) \ \text{.}
  \end{align*}

  Now expanding the definition of \ensuremath{\Varid{get}}, because of the disjointness of source patterns, the same $\ensuremath{\Varid{spat}_{\Varid{k}}} \sim \ensuremath{\Varid{vpat}_{\Varid{k}}} \in \ensuremath{\Conid{R}}$ will be select again. Thus
  \begin{align*}
     \ensuremath{\Varid{get}}(\ensuremath{\Varid{put}}(\ensuremath{\Varid{s}}, \ensuremath{\Varid{v}}, \ensuremath{\Varid{ls}})) = (\ensuremath{\Varid{reconstruct}}(\ensuremath{\Varid{vpat}_{\Varid{k}}}, \ensuremath{\Varid{fst}} \circ \ensuremath{\Varid{vls}}), \cdots)
  \end{align*}
  where
  \begin{align*}
    \ensuremath{\Varid{vls}} &= \ensuremath{\Varid{get}} \circ \ensuremath{\Varid{decompose}}(\ensuremath{\Varid{spat}_{\Varid{k}}},  \ensuremath{\Varid{put}}(\ensuremath{\Varid{s}},\ensuremath{\Varid{v}},\ensuremath{\Varid{ls}}))\\
          &= \ensuremath{\Varid{get}} \circ \ensuremath{\Varid{decompose}}(\ensuremath{\Varid{spat}_{\Varid{k}}},  \ensuremath{\Varid{reconstruct}}(\ensuremath{\Varid{fillWildcardsWithDefaults}}(\ensuremath{\Varid{spat}_{\Varid{k}}}), \ensuremath{\Varid{ss}}))\\
          &= \mcond{\text{See \autoref{fig:decomposeReconstruct}}}\\[-0.8ex]
          & \ensuremath{\Varid{get}} \circ \ensuremath{\Varid{ss}} \\
          &= \lambda (t \in \ensuremath{\Conid{Vars}}(\ensuremath{\Varid{spat}_{\Varid{k}}})) \rightarrow \ensuremath{\Varid{get}}(\ensuremath{\Varid{put}}(\ensuremath{\Varid{s}}, \ensuremath{\Varid{decompose}}(\ensuremath{\Varid{vpat}}_k, \ensuremath{\Varid{v}})(t), \ensuremath{\Varid{divide}}(\ensuremath{\Conid{Path}}(\ensuremath{\Varid{vpat}_{\Varid{k}}}, t), \ensuremath{\Varid{ls}})))
  \end{align*}

  To proceed, we want to use the inductive hypothesis to simplify $\ensuremath{\Varid{get}}(\ensuremath{\Varid{put}}(\cdots))$.
  When \ensuremath{\Varid{vpat}_{\Varid{k}}} is not a bare variable pattern, $\ensuremath{\Varid{decompose}}(\ensuremath{\Varid{vpat}}_k, \ensuremath{\Varid{v}})(t)$ is a proper subtree of \ensuremath{\Varid{v}} and the size of the third argument (i.e.~links \ensuremath{\Varid{ls}}) is non-increasing; thus the inductive hypothesis is applicable.
  On the other hand, if \ensuremath{\Varid{vpat}_{\Varid{k}}} is a bare variable pattern, the sizes of all the arguments stays the same; but \ensuremath{\Varid{cond}_{\mathrm{3}}} in \ensuremath{\Varid{chkNoLink}} guarantees that in the next round of the recursion, a pattern \ensuremath{\Varid{vpat}_{\Varid{k}}} that is not a bare variable pattern will be selected.
  Therefore we can still apply the inductive hypothesis.
  Applying the inductive hypothesis, we get
  \begin{align*}
     \ensuremath{\Varid{vls}} = \lambda (t \in \ensuremath{\Conid{Vars}}(\ensuremath{\Varid{spat}_{\Varid{k}}})) \rightarrow (\ensuremath{\Varid{decompose}}(\ensuremath{\Varid{vpat}}_k, \ensuremath{\Varid{v}})(t), \cdots)
  \end{align*}
  Thus $\ensuremath{\Varid{get}}(\ensuremath{\Varid{put}}(\ensuremath{\Varid{s}}, \ensuremath{\Varid{v}}, \ensuremath{\Varid{ls}})) = (\ensuremath{\Varid{reconstruct}}(\ensuremath{\Varid{vpat}_{\Varid{k}}}, \ensuremath{\Varid{decompose}}(\ensuremath{\Varid{vpat}_{\Varid{k}}}, \ensuremath{\Varid{v}})),\, \cdots) = (\ensuremath{\Varid{v}},\, \cdots)$, which completes the proof of Correctness.
\end{proof}

\begin{figure}[t]
\includegraphics[scale=0.7,trim={5cm 8.2cm 3.5cm 9cm},clip]{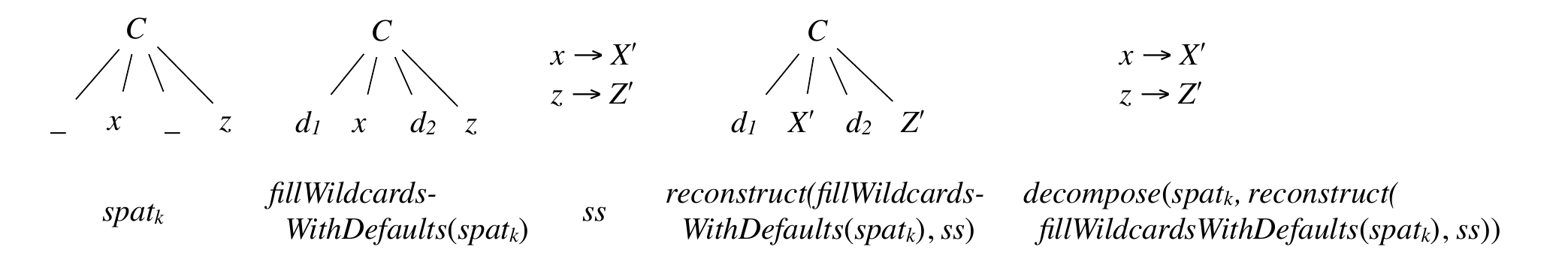}
\caption{A property regarding Decompose-Reconstruct.}
\label{fig:decomposeReconstruct}
\end{figure}

\begin{theorem}{(Retentiveness of the DSL)}
  For any \ensuremath{(\Varid{s},\Varid{v},\Varid{ls})} that $\ensuremath{\Varid{check}}(s, \ensuremath{\Varid{v}}, \ensuremath{\Varid{ls}}) = \ensuremath{\Conid{True}}$, $\ensuremath{\Varid{get}}(\ensuremath{\Varid{put}}(s, \ensuremath{\Varid{v}}, \ensuremath{\Varid{ls}})) = (\ensuremath{\Varid{v'}}, \ensuremath{\Varid{ls'}})$, for some \ensuremath{\Varid{v'}} and \ensuremath{\Varid{ls'}} such that
  \begin{align*}
       &\myset{(\ensuremath{\Varid{spat}}, (\ensuremath{\Varid{vpat}}, \ensuremath{\Varid{vpath}})) \mid ((\ensuremath{\Varid{spat}}, \ensuremath{\Varid{spath}}),\, (\ensuremath{\Varid{vpat}},\ensuremath{\Varid{vpath}})) \in \ensuremath{\Varid{ls}}} \\
    \subseteq &\myset{(\ensuremath{\Varid{spat}}, (\ensuremath{\Varid{vpat}}, \ensuremath{\Varid{vpath}})) \mid ((\ensuremath{\Varid{spat}}, \ensuremath{\Varid{spath}}),\, (\ensuremath{\Varid{vpat}},\ensuremath{\Varid{vpath}})) \in \ensuremath{\Varid{ls'}}}
  \end{align*}
\end{theorem}

\begin{proof}[Proof of Retentiveness]
  Again, we prove Retentiveness by induction on the size of $(\ensuremath{\Varid{v}}, \ensuremath{\Varid{ls}})$.
  The proofs of the two cases of \ensuremath{\Varid{put}} are similar, and thus we only show the second one here.

  If there is some $l = ((\ensuremath{\Varid{spat}}, \ensuremath{\Varid{spath}}),\, (\ensuremath{\Varid{vpat}}, \ensuremath{[\mskip1.5mu \mskip1.5mu]})) \in \ensuremath{\Varid{ls}}$, let $\ensuremath{\Varid{spat}_{\Varid{k}}} \sim \ensuremath{\Varid{vpat}_{\Varid{k}}}$ be the unique rule in $\ensuremath{\Conid{S}} \sim \ensuremath{\Conid{V}}$ that $\ensuremath{\Varid{isMatch}}(\ensuremath{\Varid{spat}_{\Varid{k}}}, \ensuremath{\Varid{spat}}) = \ensuremath{\Conid{True}}$.
  We have
  \begin{align*}
    &  \ensuremath{\Varid{get}}(\ensuremath{\Varid{put}}(\ensuremath{\Varid{s}}, \ensuremath{\Varid{v}}, \ensuremath{\Varid{ls}})) \\
    =& \mcond{\text{Definition of } \ensuremath{\Varid{put}}} \\[-0.8ex]
    & \ensuremath{\Varid{get}}(\ensuremath{\Varid{inj}}_{\ensuremath{\Conid{TypeOf}}(\ensuremath{\Varid{spat}}_k) \rightarrow S}(s')) \\
    =& \mcond{ \ensuremath{\Varid{get}}(\ensuremath{\Varid{inj}}(\ensuremath{\Varid{s}})) = \ensuremath{\Varid{get}}(\ensuremath{\Varid{s}}) \text{ as shown in (\autoref{sec:synres})}} \\[-0.8ex]
    & \ensuremath{\Varid{get}}(s')
  \end{align*}
  where $s' = \ensuremath{\Varid{reconstruct}}(\ensuremath{\Varid{fillWildcards}}(\ensuremath{\Varid{spat}_{\Varid{k}}}, \ensuremath{\Varid{spat}}), \ensuremath{\Varid{ss}})$ and
  \begin{align*}
    & ss = \lambda (t \in \ensuremath{\Conid{Vars}}(\ensuremath{\Varid{spat}_{\Varid{k}}})) \rightarrow \\
    & \myindent \ensuremath{\Varid{put}}(\ensuremath{\Varid{s}}, \ensuremath{\Varid{decompose}}(\ensuremath{\Varid{vpat}}_k, \ensuremath{\Varid{v}})(t), \ensuremath{\Varid{divide}}(\ensuremath{\Conid{Path}}(\ensuremath{\Varid{vpat}_{\Varid{k}}}, t), \ensuremath{\Varid{ls}} \setminus \myset{l}))
  \end{align*}
  Now we expand the definition of \ensuremath{\Varid{get}} (and focus on the links)
  \begin{align*}
    \ensuremath{\Varid{get}}(\ensuremath{\Varid{put}}(\ensuremath{\Varid{s}}, \ensuremath{\Varid{v}}, \ensuremath{\Varid{ls}})) = (\cdots, \myset{l_\ensuremath{\Varid{root}}} \cup \ensuremath{\Varid{links}})
  \end{align*}
  where $l_\ensuremath{\Varid{root}} = \left(\left(\ensuremath{\Varid{eraseVars}}(\ensuremath{\Varid{fillWildcards}}(\ensuremath{\Varid{spat}_{\Varid{k}}}, \ensuremath{\Varid{s'}})), []\right),\, (\ensuremath{\Varid{eraseVars}}(\ensuremath{\Varid{vpat}_{\Varid{k}}}), [])\right)$,
  \begin{align*}
    &\ensuremath{\Varid{links}} = \{\, ((a, \ensuremath{\Conid{Path}}(\ensuremath{\Varid{spat}_{\Varid{k}}}, t) \ensuremath{\plus } b),\,(c, \ensuremath{\Conid{Path}}(\ensuremath{\Varid{vpat}_{\Varid{k}}}, t) \ensuremath{\plus } d)) \numberthis \label{equ:links}\\
    & \myindent \mid t \in \ensuremath{\Conid{Vars}}(\ensuremath{\Varid{vpat}_{\Varid{k}}}), ((a, b),\,(c, d)) \in \ensuremath{\Varid{snd}}(\ensuremath{\Varid{vls}}(t))\,\} \text{ ,and}
  \end{align*}
  \begin{align*}
    & \ensuremath{\Varid{vls}}(t) \\
    = & \mcond{\text{Unfolding } \ensuremath{\Varid{vls}}} \\[-0.8ex]
    & (\ensuremath{\Varid{get}}\circ \ensuremath{\Varid{decompose}}(\ensuremath{\Varid{spat}_{\Varid{k}}}, s'))(t) \\
    = & \mcond{\text{Unfolding } \ensuremath{\Varid{s'}}} \\[-0.8ex]
    & (\ensuremath{\Varid{get}} \circ \ensuremath{\Varid{decompose}}(\ensuremath{\Varid{spat}_{\Varid{k}}}, \ensuremath{\Varid{reconstruct}}(\ensuremath{\Varid{fillWildcards}}(\ensuremath{\Varid{spat}_{\Varid{k}}}, \ensuremath{\Varid{spat}}), \ensuremath{\Varid{ss}}))) (t) \\
    = & \mcond{ \text{Similar to the case shown in \autoref{fig:decomposeReconstruct}} }\\[-0.8ex]
    & (\ensuremath{\Varid{get}} \circ \ensuremath{\Varid{ss}})(t) \\
    = & \mcond{\text{Definition of } \ensuremath{\Varid{ss}}}\\[-0.8ex]
    & \ensuremath{\Varid{get}}(\ensuremath{\Varid{put}}(\ensuremath{\Varid{s}}, \ensuremath{\Varid{decompose}}(\ensuremath{\Varid{vpat}}_k, \ensuremath{\Varid{v}})(t), \ensuremath{\Varid{divide}}(\ensuremath{\Conid{Path}}(\ensuremath{\Varid{vpat}_{\Varid{k}}}, t), \ensuremath{\Varid{ls}} \setminus \myset{l}))) \ \text{.}
  \end{align*}
  For $l_\ensuremath{\Varid{root}}$, we have
  \begin{align*}
      &\ensuremath{\Varid{eraseVars}}(\ensuremath{\Varid{fillWildcards}}(\ensuremath{\Varid{spat}_{\Varid{k}}}, \ensuremath{\Varid{s'}})) \\
    = &\mcond{\text{Unfolding } \ensuremath{\Varid{s'}}}\\[-0.8ex]
    &\ensuremath{\Varid{eraseVars}}(\ensuremath{\Varid{fillWildcards}}(\ensuremath{\Varid{spat}_{\Varid{k}}}, \ensuremath{\Varid{reconstruct}}(\ensuremath{\Varid{fillWildcards}}(\ensuremath{\Varid{spat}_{\Varid{k}}}, \ensuremath{\Varid{spat}}), \ensuremath{\Varid{ss}}))) \\
    = &\mcond{\text{See \autoref{fig:fillWildcards2}}}\\[-0.8ex]
    &\ensuremath{\Varid{eraseVars}}(\ensuremath{\Varid{fillWildcards}}(\ensuremath{\Varid{spat}_{\Varid{k}}}, \ensuremath{\Varid{spat}})) \\
    = &\mcond{\text{By }\ensuremath{\Varid{cond}_{\mathrm{2}}} \text{ in } \hyperref[def:chkWithLink]{\ensuremath{\Varid{chkWithLink}}}} \\[-0.8ex]
    &\ensuremath{\Varid{spat}}
  \end{align*}
   Use the first clause of \ensuremath{\Varid{cond}_{\mathrm{2}}}, we have $\ensuremath{\Varid{vpat}} = \ensuremath{\Varid{eraseVars}}(\ensuremath{\Varid{vpat}_{\Varid{k}}})$.
   Thus \[ l_\ensuremath{\Varid{root}} = \left(\left(\ensuremath{\Varid{eraseVars}}(\ensuremath{\Varid{fillWildcards}}(\ensuremath{\Varid{spat}_{\Varid{k}}}, \ensuremath{\Varid{s'}})), []\right),\, (\ensuremath{\Varid{eraseVars}}(\ensuremath{\Varid{vpat}_{\Varid{k}}}), [])\right) =
   ((\ensuremath{\Varid{spat}}, \ensuremath{[\mskip1.5mu \mskip1.5mu]}),\,(\ensuremath{\Varid{vpat}}, \ensuremath{[\mskip1.5mu \mskip1.5mu]})) \ \text{,} \]
   and therefore the input link $l = ((\ensuremath{\Varid{spat}}, \ensuremath{\Varid{spath}}),\, (\ensuremath{\Varid{vpat}}, \ensuremath{[\mskip1.5mu \mskip1.5mu]}))$ is `preserved' by $l_\ensuremath{\Varid{root}}$, i.e.~$\ensuremath{\Varid{fst}} \cdot \{l\} = \ensuremath{\Varid{fst}} \cdot \{\ensuremath{\Varid{l}}_{\ensuremath{\Varid{root}}}\}$ .

   For the links in $\ensuremath{\Varid{ls}} \setminus \myset{\ensuremath{\Varid{l}}}$, we show that they are preserved in \ensuremath{\Varid{links}} (\ref{equ:links}) above.
   By \ensuremath{\Varid{cond}_{\mathrm{3}}} in \hyperref[def:chkWithLink]{\ensuremath{\Varid{chkWithLink}}}, for every link $\ensuremath{\Varid{m}} \in \ensuremath{\Varid{ls}} \setminus \myset{\ensuremath{\Varid{l}}}$, there is some \ensuremath{\Varid{t\char95 m}} in $\ensuremath{\Conid{Vars}}(\ensuremath{\Varid{spat}_{\Varid{k}}})$ such that \[\ensuremath{\Varid{m}} \in \ensuremath{\Varid{addVPrefix}}(\ensuremath{\Conid{Path}}(\ensuremath{\Varid{vpat}_{\Varid{k}}}, t_m), \ensuremath{\Varid{divide}}(\ensuremath{\Conid{Path}}(\ensuremath{\Varid{vpat}}_k, t_m), \ensuremath{\Varid{ls}} \setminus \myset{\ensuremath{\Varid{l}}})). \]
   If $m = ((a, b), (c, \ensuremath{\Conid{Path}}(\ensuremath{\Varid{vpat}}_k, t_m) \ensuremath{\plus } d))$, then
   \[m' = ((a, b),\, (c,d)) \in \ensuremath{\Varid{divide}}(\ensuremath{\Conid{Path}}(\ensuremath{\Varid{vpat}}_k, t_m), \ensuremath{\Varid{ls}} \setminus \myset{\ensuremath{\Varid{l}}}). \]
   By the inductive hypothesis for $\ensuremath{\Varid{snd}}(\ensuremath{\Varid{vls}}(t_m))$, \ensuremath{\Varid{m'}} is `preserved', that is
   \begin{align*}
     \exists b'.\, ((a, b'),\ (c, d)) \in \ensuremath{\Varid{snd}}(\ensuremath{\Varid{vls}}(t_m))
   \end{align*}
   Now by the definition of \ensuremath{\Varid{links}} (\ref{equ:links}), $((a, \ensuremath{\Conid{Path}}(\ensuremath{\Varid{spat}}_k, t_m) \ensuremath{\plus } b'), (c, \ensuremath{\Conid{Path}}(\ensuremath{\Varid{vpat}}_k, t_m) \ensuremath{\plus } d)) \in \ensuremath{\Varid{links}}$, therefore $m$ is also preserved.
\end{proof}

\begin{corollary}
  Let $\ensuremath{\Varid{put'}} = \ensuremath{\Varid{put}}$ with its domain intersected with $\ensuremath{\Conid{S}} \times \ensuremath{\Conid{V}} \times \ensuremath{\Conid{LinkSet}}$, \ensuremath{\Varid{get}} and \ensuremath{\Varid{put'}} form a retentive lens as in \autoref{def:retLens} since they satisfy Hippocraticness (\ref{law:get}), Correctness (\ref{law:correct}) and Retentiveness (\ref{law:retain}).
\end{corollary}

\begin{figure}[t]
\centering
\includegraphics[scale=0.7,trim={5cm 4cm 10cm 5cm},clip]{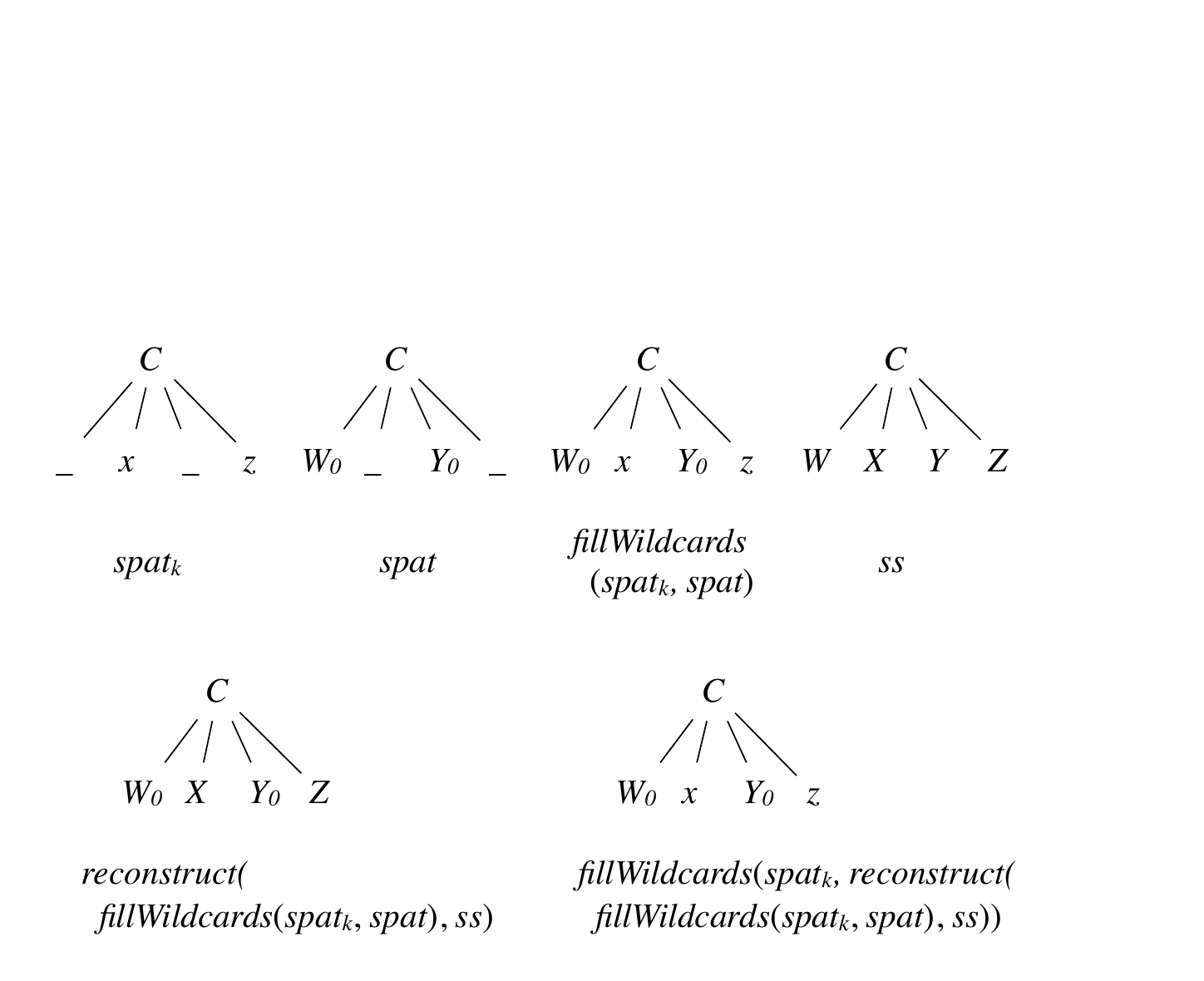}
\caption{Another property regarding \ensuremath{\Varid{fillWildcards}}.}
\label{fig:fillWildcards2}
\end{figure}

\section{Refactoring Operations as Edit Operation Sequences}
\label{sec:refactOptAsEditSeq}

We summarise how the 23 refactoring operations for Java~8 in Eclipse Oxygen could be described by \ensuremath{\Varid{replace}}, \ensuremath{\Varid{copy}}, \ensuremath{\Varid{move}}, \ensuremath{\Varid{swap}}, \ensuremath{\Varid{insert}}, and \ensuremath{\Varid{delete}}, where the \ensuremath{\Varid{insert}} and \ensuremath{\Varid{delete}} operations on lists can be implemented in terms of the first four.
For instance, to \ensuremath{\Varid{insert}} an element \ensuremath{\Varid{e}} at position \ensuremath{\Varid{i}} in a list of length \ensuremath{\Varid{n}} (where $1 \leq i \leq n$), we can follow these steps:
(i) Change the list to length $n+1$.
(ii) Starting from the tail of the list, \ensuremath{\Varid{move}} each element at position \ensuremath{\Varid{j}} such that \ensuremath{\Varid{j}\mathbin{>}\Varid{i}} to position \ensuremath{\Varid{j}\mathbin{+}\mathrm{1}}.
(iii) \ensuremath{\Varid{replace}} the element at position \ensuremath{\Varid{i}} with \ensuremath{\Varid{e}}.
Deleting the element at position \ensuremath{\Varid{i}} is almost as simple as moving each element after \ensuremath{\Varid{i}} one position ahead and decrease the length of the list by one.

\begin{small}
\begin{longtable}[c]{| p{0.115\textwidth} | p{0.4\textwidth} | p{0.4\textwidth} |}
\caption{Refactoring Operations as Edit Operation Sequences.}\\
\hline
Refactor Operation & Description & Edit Operations\\
\hline
Rename & Renames the selected element and (if enabled) corrects all references to the elements & \ensuremath{\Varid{replace}} the selected element and all references with the new name.\\
\hline
Use Supertype Where Possible & Replaces occurrences of a type with one of its supertypes after identifying all places where this replacement is possible. & \ensuremath{\Varid{replace}} all occurrences. \\
\hline
Generalize Declared Type & Allows the user to choose a supertype of the reference's current type. If the reference can be safely changed to the new type, it is. & \ensuremath{\Varid{replace}} all occurrences. \\
\hline
Infer Generic Type Arguments & Replaces raw type occurrences of generic types by parameterized types after identifying all places where this replacement is possible. & \ensuremath{\Varid{replace}} all occurrences. \\
\hline
Encapsulate Field & Replaces all references to a field with getter and setter methods. & \ensuremath{\Varid{insert}} getters and setters; \ensuremath{\Varid{replace}} all occurrences (with getters or setters respectively). \\
\hline
Change Method Signature & Changes parameter names, parameter types, parameter order and updates all references to the corresponding method. & \ensuremath{\Varid{replace}} all occurrences. Use \ensuremath{\Varid{swap}} if we need to change the parameter order. \\
\hline
Extract Method & Creates a new method containing the statements or expression currently selected and replaces the selection with a reference to the new method. & \ensuremath{\Varid{insert}} a new method; \ensuremath{\Varid{move}} selected code; \ensuremath{\Varid{replace}} the selection.\\
\hline
Extract Local Variable & Creates a new variable assigned to the expression currently selected and replaces the selection with a reference to the new variable. & \ensuremath{\Varid{insert}} a new variable; \ensuremath{\Varid{copy}} the selected expression to the variable assignment; \ensuremath{\Varid{replace}} the selected expression. \\
\hline
Extract Constant & Creates a static final field from the selected expression and substitutes a field reference, and optionally rewrites other places where the same expression occurs. & \ensuremath{\Varid{insert}} a field; \ensuremath{\Varid{copy}} the selected expression; \ensuremath{\Varid{replace}} the selected expression. \\
\hline
Introduce Parameter & Replaces an expression with a reference to a new method parameter, and updates all callers of the method to pass the expression as the value of that parameter. & \ensuremath{\Varid{insert}} a method parameter; \ensuremath{\Varid{insert}} the selected expression to all the callers (use \ensuremath{\Varid{copy}} if we want to preserve the information attached to the expression); \ensuremath{\Varid{replace}} the expression with the new method parameter.  \\
\hline
Introduce Factory & Creates a new factory method, which will call a selected constructor and return the created object. All references to the constructor will be replaced by calls to the new factory method. & \ensuremath{\Varid{insert}} a factory method; \ensuremath{\Varid{replace}} all the references to the constructor. \\
\hline
Introduce Indirection & Creates a static indirection method delegating to the selected method. & \ensuremath{\Varid{insert}} a method. \\
\hline
Convert  to Nested & Converts an anonymous inner class to a member class. & \ensuremath{\Varid{insert}} a member class; \ensuremath{\Varid{move}} the code within the anonymous class to the member class; \ensuremath{\Varid{delete}} the anonymous class.  \\
\hline
Move Type to New File & Creates a new Java compilation unit for the selected member type or the selected secondary type, updating all references as needed. & \ensuremath{\Conid{Move}} the selected code to the new file; \ensuremath{\Varid{replace}} all references. \\
\hline
Convert Local Variable to Field & Turn a local variable into a field. If the variable is initialized on creation, then the operation moves the initialization to the new field's declaration or to the class's constructors. & \ensuremath{\Varid{insert}} a field; \ensuremath{\Varid{copy}} the initialization; \ensuremath{\Varid{delete}} the variable declaration. \\
\hline
Extract Superclass & Extracts a common superclass from a set of sibling types. The selected sibling types become direct subclasses of the extracted superclass after applying the refactoring. & \ensuremath{\Varid{insert}} a superclass; \ensuremath{\Varid{move}} fields to the superclass; \ensuremath{\Varid{replace}} declarations of sibling types (classes) so that they extend the superclass; \ensuremath{\Varid{insert}} lacking fields into sibling classes. \\
\hline
Extract Interface & Creates a new interface with a set of methods and makes the selected class implement the interface. & generally the same as above. \\
\hline
Move & Moves the selected elements and (if enabled) corrects all references to the elements (also in other files). & \ensuremath{\Varid{move}} the selected elements; \ensuremath{\Varid{replace}} all references. \\
\hline
Push Down & Moves a set of methods and fields from a class to its subclasses. & \ensuremath{\Varid{move}} the methods and fields. \\
\hline
Pull Up & Moves a field or method to a superclass of its declaring class or (in the case of methods) declares the method as abstract in the superclass. & \ensuremath{\Varid{move}} the field or \ensuremath{\Varid{insert}} an abstract method declaration. \\
\hline
Introduce Parameter Object & Replaces a set of parameters with a new class, and updates all callers of the method to pass an instance of the new class as the value to the introduce parameter. & \ensuremath{\Varid{insert}} a class definition; \ensuremath{\Varid{move}} the parameters to the class; in callers' definitions, \ensuremath{\Varid{delete}} the set of parameters and \ensuremath{\Varid{insert}} the class type as a new parameter; for callers' arguments, \ensuremath{\Varid{delete}} the arguments corresponding to the set of parameters and \ensuremath{\Varid{insert}} a class instance. \\
\hline
Extract Class & Replaces a set of fields with new container object. All references to the fields are updated to access the new container object. & \ensuremath{\Varid{insert}} a new class; \ensuremath{\Varid{move}} the fields; \ensuremath{\Varid{replace}} references to the fields with references to the container object and field names. \\
\hline
Inline & Inline local variables, methods or constants. & For a variable or a constant, \ensuremath{\Varid{replace}} the occurrences with the value; \ensuremath{\Varid{delete}} the definition. For a method, \ensuremath{\Varid{replace}} all occurrences of parameters within the method body with real arguments; \ensuremath{\Varid{replace}} the method call with the (new) method body; \ensuremath{\Varid{delete}} the method definition. \\
\hline
\end{longtable}
\end{small}



\end{document}